\documentclass[11pt]{article}
\usepackage[english]{babel}
\usepackage[margin=1in]{geometry}
\usepackage{setspace}
\usepackage{lmodern}

\usepackage{amsmath}
\usepackage{graphicx}
\usepackage[colorlinks=true, allcolors=blue]{hyperref}
\usepackage{mathrsfs}
\usepackage{multirow}

\usepackage{algorithm}
\usepackage{algpseudocode}
\usepackage{calc}%Makes good indentations
\usepackage{amsthm}
\usepackage{amsfonts}
\usepackage{tikz}
\usepackage{multicol}
\usetikzlibrary {arrows.meta, shapes}

\usepackage{thmtools}
\usepackage{thm-restate}

\usepackage{hyperref}

\usepackage{cleveref}

\newcommand\restr[2]{{% we make the whole thing an ordinary symbol
  \left.\kern-\nulldelimiterspace % automatically resize the bar with \right
  #1 % the function
  \littletaller % pretend it's a little taller at normal size
  \right|_{#2} % this is the delimiter
  }}
\newcommand{\myindent}[1]{
\newline\makebox[#1cm]{}}
\newcommand*{\Let}[2]{\State #1 $\gets$
\parbox[t]{\linewidth-\algorithmicindent-\widthof{ #1 $\gets$}}{#2\strut}}
\newcommand*{\LongState}[1]{\State
\parbox[t]{\linewidth-\algorithmicindent}{#1\strut}}

\newtheorem{lemma}{Lemma}
\newtheorem{theorem}{Theorem}
\newtheorem{claim}{Claim}
\newtheorem{corollary}{Corollary}

\newtheorem{definition}{Definition}
\newtheorem{observation}{Observation}

\author{Michael Elkin\\
        \vspace{-0.15cm}\small Department of Computer Science,\\
        \vspace{-0.15cm}\small Ben-Gurion University of the Negev$^{*}$,\\
        \vspace{-0.1cm}\small Beer-Sheva, Israel.\\
        \href{elkinm@cs.bgu.ac.il}{elkinm@cs.bgu.ac.il}
        \and Ariel Khuzman\\
        \vspace{-0.15cm}\small Department of Computer Science,\\
        \vspace{-0.15cm}\small Ben-Gurion University of the Negev$^{*}$,\\
        \vspace{-0.1cm}\small Beer-Sheva, Israel.\\
        \href{huzmana@post.bgu.ac.il}{huzmana@post.bgu.ac.il}}

\title{Efficient Parallel $\left(\Delta+1\right)$-Edge-Coloring}

\date{}
\begin{document}
\maketitle

\begin{abstract}
We study the $(\Delta+1)$-edge-coloring problem in the parallel ($\mathrm{PRAM}$) model of computation. The celebrated Vizing's theorem \cite{vizing1964estimate} states that every simple graph $G = (V,E)$ can be properly $(\Delta+1)$-edge-colored.
In a seminal paper, Karloff and Shmoys \cite{karloff1987efficient} devised a parallel algorithm with time $O\left(\Delta^5\cdot\log n\cdot\left(\log^3 n+\Delta^2\right)\right)$ and $O(m\cdot\Delta)$ processors. This result was improved by Liang et al. \cite{liang1996parallel} to time $O\left(\Delta^{4.5}\cdot \log^3\Delta\cdot \log n + \Delta^4 \cdot\log^4 n\right)$ and $O\left(n\cdot\Delta^{3} +n^2\right)$ processors. 
\cite{liang1996parallel} claimed $O\left(\Delta^{3.5} \cdot\log^3\Delta\cdot \log n + \Delta^3\cdot \log^4 n\right)$ time, but we point out a flaw in their analysis, which once corrected, results in the above bound.
We devise a faster parallel algorithm for this fundamental problem. Specifically, our algorithm uses $O\left(\Delta^4\cdot \log^4 n\right)$ time and $O(m\cdot \Delta)$ processors. Another variant of our algorithm requires $O\left(\Delta^{4+o(1)}\cdot\log^2 n\right)$ time, and $O\left(m\cdot\Delta\cdot\log n\cdot\log^{\delta}\Delta\right)$ processors, for an arbitrarily small $\delta>0$.
We also devise a few other tradeoffs between the time and the number of processors, and devise an improved algorithm for graphs with small arboricity.
On the way to these results, we also provide a very fast parallel algorithm for updating $(\Delta+1)$-edge-coloring. Our algorithm for this problem is dramatically faster and simpler than the previous state-of-the-art algorithm (due to \cite{liang1996parallel}) for this problem.
\end{abstract}
\renewcommand{\thefootnote}{*}
\footnotetext{This research is supported by the ISF grant 3413/25.}
\renewcommand{\thefootnote}{\arabic{footnote}}

\newpage
\section{Introduction}

\subsection{General Graphs}
Given an $n$-vertex $m$-edge undirected simple graph $G=(V,E)$, an edge-coloring 
$\varphi:E\to\mathbb{N}$ is called \emph{proper} if 
$\varphi(e)\neq\varphi(e')$ for every pair of distinct edges $e\neq e'\in E$ that share an endpoint. 
If a (proper) coloring $\varphi$ employs only colors 
$\{1,2,\ldots,t\}$, for a positive integer $t$, 
then it is called a (proper) \emph{$t$-edge-coloring} of $G$.
The celebrated Vizing's theorem \cite{vizing1964estimate} states that any (simple) graph $G$ admits a $(\Delta+1)$-edge-coloring, where $\Delta$ is the maximum degree of $G$. 
Efficient computation of 
$(\Delta+1)$-edge-coloring in various computational models~\cite{misra1992constructive, bernshteyn2022fast, bernshteyn2023fast, gabow1985algorithms, bhattacharya2025even, assadi2024faster, karloff1987efficient, liang1996parallel,furer1996parallel, liang1997parallel, barenboim2011distributed, barenboim2017deterministic, ghaffari2020improved, kowalik2025planar, chrobak1989fast, chrobak1990improved, kowalik2024edge} constitutes a fundamentally important algorithmic problem.

In this paper we focus on the edge-coloring problem in the \emph{parallel} ($\mathrm{PRAM}$) model of computation. In a classical paper, Karloff and Shmoys~\cite{karloff1987efficient} devised a $(\Delta+1)$-edge-coloring algorithm of running time $O\left(\Delta^5\cdot\log n\cdot\left(\log^3 n+\Delta^2\right)\right)$ and $O(m\cdot\Delta)$ processors, assuming the fastest currently known algorithm for computing maximal independent sets (MIS)\footnote{See Definition \ref{def: independent set} for the definition of independent set.} \cite{goldberg1989constructing} is employed as a subroutine.
Further progress was achieved by Liang et al.~\cite{liang1996parallel}. Their algorithm requires $O\left(\Delta^{4.5}\cdot\log^3\Delta\cdot\log n+\Delta^4\cdot\log^4 n\right)$ time and $O\left(n\cdot\Delta^3+n^2\right)$ processors. (In fact, they claimed time $O\left(\Delta^{3.5}\cdot\log^3\Delta\cdot\log n+\Delta^3\log^4 n\right)$, but we believe that there is a flaw in their argument (see Appendix \ref{app: Computing a Large Collection of Pairwise-Disjoint Fans}). Once corrected, their running time becomes as stated above.)
A different tradeoff was provided by~\cite{liang1997parallel}: their algorithm has running time $O\left(\Delta^9\cdot\log^2 n\right)$, but uses $O(m\cdot\Delta)$ processors.

We significantly improve upon previous bounds, and devise a host of $(\Delta+1)$-edge-coloring algorithms whose respective running times and numbers of processors are summarized in Table \ref{table: A summary of (Delta+1)-edge-coloring routines}. In particular, one variant of our algorithm has running time $O\left(\Delta^4\cdot\log^4 n\right)$ and uses $O(m\cdot\Delta)$ processors.
For $\Delta=\omega\left(\log^6 n\right)$ we have $\Delta^4\cdot\log^4 n\ll\Delta^{4.5}\cdot\log^3\Delta\cdot\log n$, i.e., our running time improves that of~\cite{liang1996parallel} in that range. The number of processors that this algorithm uses ($O(m \cdot \Delta)$) is also smaller than the number of processors $O\left(n\cdot\Delta^3+n^2 \cdot\log\Delta\right)$ of~\cite{liang1996parallel}. (Indeed, $n\cdot\Delta^3\geq m\cdot\Delta^2\gg m\cdot\Delta$.)
For polylogarithmic $\Delta$, i.e., $\Delta=\log^{O(1)} n$, one can use another variant of our algorithm that has running time $\Delta^{4+o(1)}\cdot\log^2 n$ and employs $O\left(m\cdot\Delta\cdot\frac{\log^{\delta}\Delta\cdot\log n}{\log(\Delta\cdot\log n)}\right)$ processors. In this range, our running time becomes $\Delta^4\cdot\log^{2+o(1)} n$, and the number of processors is $\tilde{O}(n)$, while the algorithm of~\cite{liang1996parallel} has running time $O\left(\Delta^4\cdot\log^4 n\right)$ and uses $O\left(n^2\right)$ processors. Another variant of our algorithm has running time $O\left(\Delta^{6}\cdot\log^2\Delta\cdot\log n\right)$ and employs $O(m \cdot\Delta)$ processors. This running time is strictly better than $O\left(\Delta^9\cdot \log^2 n\right)$, which is the running time of~\cite{liang1997parallel}. (The latter is incomparable with the aforementioned result of~\cite{liang1996parallel}.) The number of processors in the algorithm of~\cite{liang1997parallel} is also $O(m\cdot\Delta)$, like in this variant of our algorithm.

\begin{table}[h!]
    \begin{center}
    \addtolength{\leftskip} {-2cm}
    \addtolength{\rightskip}{-2cm}
    \begin{tabular}{|c | c | c | c|} 
    \hline
     Algorithm & Running Time & \# Processors \\ [0.5ex] 
     \hline
     \cite{karloff1987efficient} &  $O\left(\Delta^5\cdot\log n\cdot\left(\log^3 n+\Delta^2\right)\right)$ & $O\left(m\cdot\Delta\right)$  \\
     \hline
     \cite{liang1996parallel} &  $O\left(\Delta^{4.5}\cdot\log^3\Delta\cdot\log n+\Delta^4\cdot\log^4 n\right)$ & $O\left(n\cdot\Delta^3+n^2\right)$  \\
     \hline
     \cite{liang1997parallel} &  $O\left(\Delta^{9}\cdot\log^2 n\right)$ & $O\left(m\cdot\Delta\right)\,\,\,\,\,^{(*)}$  \\\hline
     \textbf{Ours} & $O\left(\Delta^4\cdot\log^4 n\right)$ & $O\left(m\cdot\Delta\right)$\\
     \hline
     \textbf{Ours} & $O\left(\Delta^4\cdot\log^2 n+\Delta^6\cdot\log^2\Delta\cdot\log n\right)$ & $O\left(m\cdot\Delta\right)$\\
     \hline
     \textbf{Ours} & $O\left(a^2\cdot\Delta^4\cdot\log\Delta\cdot\log^2 n\right)$ & $O\left(m\cdot\Delta\right)$\\
     \hline
     \textbf{Ours} & $O\left(\Delta^{3+o(1)}\cdot a^{1+o(1)}\cdot\log^2 n\right)=O\left(\Delta^{4+o(1)}\cdot\log^2 n\right)$ & $O\left(m\cdot\Delta\cdot\frac{\log^{\delta}\Delta\cdot\log n}{\log(\Delta\cdot\log n)}\right)\,\,\,\,\,\,^{(**)}$ \\
     \hline
     \textbf{Ours} & $O\left(\Delta^5\cdot\log^2 n\right)$ & $O\left(m\cdot\left(\Delta\cdot\log\Delta+\frac{\sqrt{\Delta}\cdot\log n}{\log(\Delta\cdot\log n)}\right)\right)$\\
     [1ex] 
     \hline
    \end{tabular}
    \end{center}
    \caption{A summary of new and existing parallel $(\Delta+1)$-edge-coloring algorithms for $n$-vertex, $m$-edge\ graphs with maximum degree $\Delta$ and arboricity $a$.\\
    $^{(*)}$ In~\cite{liang1997parallel}, the authors claim that the number of processors is $O(n+m)$.\\
    $^{(**)}$ The result holds for any constant $\delta>0$.}
    \label{table: A summary of (Delta+1)-edge-coloring routines}
\end{table}

\subsection{Graphs with Bounded Arboricity}

\emph{Arboricity} $a=a(G)$ of a graph $G=(V,E)$ is defined by
$$a(G) = \max_{\substack{U \subseteq V\\ |U| \ge 2}} \frac{|E(U)|}{|U| - 1}.$$
By Nash-Williams’ theorem, it is equal to the minimum number of edge-disjoint forests required to cover the edge set $E$ of $G$.
Graphs with bounded arboricity is an important graph family that contains planar graphs, graphs with bounded genus, graphs that exclude any fixed minor, graphs that admit sublinear recursive separators, and many other families of sparse graphs.
Zhou et al.~\cite{zhou1994edge} devised a parallel $O\left(\log^3 n\right)$-time $\Tilde{O}(n)$-work algorithm for $\Delta$-edge-coloring graphs with $\Delta\geq\left\lceil\frac{(a+2)^2}{2}\right\rceil-1$, as long as $\Delta=O(1)$.
We provide improved parallel $(\Delta+1)$-edge-coloring algorithms for graphs with $a = o(\Delta)$. Specifically, one variant of our algorithm provides time $O\left(a^2\cdot\Delta^4\cdot\log\Delta\cdot\log^2 n\right)$ and $O(m \cdot \Delta)$ processors, and another has time $O\left(\Delta^{3+o(1)}\cdot a^{1+o(1)}\cdot\log^2 n\right)$ and $O\left(m\cdot\Delta\cdot\frac{\log^{\delta}\Delta\cdot\log n}{\log(\Delta\cdot\log n)}\right)$ processors.
See Table \ref{table: A summary of (Delta+1)-edge-coloring routines} for a concise summary of these algorithms.

\subsection{$(1+\varepsilon)\Delta$-Edge-Coloring}

Building upon $(\Delta+1)$-edge-coloring algorithms of~\cite{liang1996parallel,liang1997parallel}, Liang~\cite{liang1995fast} and Furer and Raghavachari~\cite{furer1996parallel} devised two $(1+\varepsilon)\Delta$-edge-coloring algorithms, for a parameter $\varepsilon>0$. The algorithm of~\cite{liang1995fast} has running time $O\left(\varepsilon^{-4.5}\cdot\log^3\varepsilon^{-1}\cdot\log n+\varepsilon^{-4}\cdot\log^4 n\right)$ and $O\left(n\cdot\varepsilon^{-3}+n^2\right)$ processors, while that of~\cite{furer1996parallel} has time $O\left(\varepsilon^{-9}\cdot\log^2 n\right)$ and $O\left(m\cdot\varepsilon^{-1}\right)$ processors. We also employ our improved $(\Delta+1)$-edge-coloring algorithms to derive a number of $(1+\varepsilon)\Delta$-edge-coloring algorithms. In particular, we provide time $O\left(\varepsilon^{-4}\cdot\log^4 n\right)$ and $O\left(m\cdot\varepsilon^{-1}\right)$ processors, or time $O\left(\varepsilon^{-4-o(1)}\cdot\log^2 n\right)$ time and $O\left(m\cdot\varepsilon^{-1}\cdot\frac{\log^{\delta}\left(\varepsilon^{-1}\right)\cdot\log n}{\log\left(\varepsilon^{-1}\cdot\log n\right)}\right)$ processors.
See Table \ref{table: A summary of (1+eps)Delta-edge-coloring routines.} for other tradeoffs that we achieve, and a concise summary of previous and new bounds for the $(1+\varepsilon)\Delta$-edge-coloring problem.

\begin{table}[h!]
    \begin{center}
    \addtolength{\leftskip} {-2.37cm}
    \addtolength{\rightskip}{-2cm}
    \begin{tabular}{|c | c | c | c|} 
     \hline
     Algorithm & Running Time & \# Processors \\ [0.5ex] 
     \hline
     \cite{liang1995fast} &  $O\left(\varepsilon^{-4.5}\cdot\log^3\varepsilon^{-1}\cdot\log n+\varepsilon^{-4}\cdot\log^4 n\right)$ & $O\left(n\cdot\varepsilon^{-3}+n^2\right)$  \\
     \hline
     \cite{furer1996parallel}+\cite{liang1997parallel} $^{(*)}$& $O\left(\varepsilon^{-9}\cdot\log^2 n\right)$ & $O\left(m\cdot \varepsilon^{-1}\right)$\\
     \hline
     \textbf{Ours} & $O\left(\varepsilon^{-4}\cdot\log^4 n\right)$ & $O\left(m\cdot\varepsilon^{-1}\right)$\\
     \hline
     \textbf{Ours} & $O\left(\varepsilon^{-4}\cdot \log^2 n+\varepsilon^{-6}\cdot\log^2\varepsilon^{-1}\cdot\log n\right)$ & $O\left(m\cdot\varepsilon^{-1}\right)$\\
     \hline
     \textbf{Ours} & $O\left(\varepsilon^{-4-o(1)}\cdot\log^2 n\right)$ & $O\left(m\cdot\varepsilon^{-1}\cdot\frac{\log^{\delta}\left(\varepsilon^{-1}\right)\cdot\log n}{\log\left(\varepsilon^{-1}\cdot\log n\right)}\right)\,\,\,\,\,^{(**)}$ \\
     \hline
     \textbf{Ours} & $O\left(\varepsilon^{-5}\cdot\log^2 n\right)$ & $O\left(m\cdot\left(\varepsilon^{-1}\cdot\log\varepsilon^{-1}+\frac{\varepsilon^{-1/2}\cdot\log n}{\log\left(\varepsilon^{-1}\cdot\log n\right)}\right)\right)$\\
     [1ex] 
     \hline
    \end{tabular}
    \end{center}
    \caption{A summary of new and existing $(1+\varepsilon)\Delta$-edge-coloring algorithms.\\
    $^{(*)}$ The dependence on $\varepsilon$ is implicit in~\cite{furer1996parallel}. They show that when $\varepsilon$ is constant, the running time is $O\left(\log^2 n\right)$ and the number of processors is $O(m+n)$.\\
    $^{(**)}$ The result holds for any constant $\delta>0$.
}
    \label{table: A summary of (1+eps)Delta-edge-coloring routines.}
\end{table}

\subsection{Edge-Coloring Update Problem}\label{Sec: Edge-Coloring Update Problem}
A central ingredient in the algorithm of~\cite{liang1996parallel}, which is of interest by its own right, is the following dynamic version of the $(\Delta+1)$-edge-coloring problem: given a graph $G=(V,E)$ equipped with a proper $(\Delta+1)$-edge-coloring, suppose that a new vertex $v \notin V$ is added to the graph, along with at most $\Delta$ edges connecting it to the existing vertices of $G$. The \emph{dynamic update} algorithm of~\cite{liang1996parallel} solves this problem in $O\left(\Delta^{\frac{3}{2}}\cdot\log^3\Delta+\Delta\cdot\log n\right)$ time, using $O\left(n\cdot\Delta+\Delta^3\right)$ processors. We devise a completely different and a far more efficient solution for this problem. Our algorithm for it requires $O(\Delta\cdot\log n)$ time, using $O(n)$ processors. In fact, stronger than that, we can process every single edge update within $O(\log n)$ parallel time, using $O(n)$ processors.

\subsection{Technical Overview}\label{Section: Technical Overview}
\paragraph{The Algorithm of Karloff-Shmoys}
The pioneering parallel $(\Delta+1)$-edge-coloring algorithm of Karloff and Shmoys~\cite{karloff1987efficient} consists of $O\left(\Delta^5\cdot\log n\right)$ phases, each of which colors a large fraction ($\Omega\left(\frac{1}{\Delta^5}\right)$) of the remaining uncolored edges. The algorithm maintains a set $\mathcal{A}\subseteq V$ of \emph{active} vertices, that is, vertices that are incident to at least one uncolored edge. 
It constructs a graph $G_{\mathcal{A}}$ on these vertices: two vertices $u$ and $v$ are adjacent if they are at distance at most $2$ in the original graph $G$. Then the algorithm computes a maximal independent set $I$ in $G_A$, which has cardinality at least $\frac{|\mathcal{A}|}{\Delta^2}$. These vertices compute fans in parallel (see Section \ref{sec: Vizing Theorem} for the definition of fans and Figure \ref{fig:fan} for an illustration). Fan is the basic structure in Vizing's~\cite{vizing1964estimate} algorithm, and in all its subsequent efficient implementations (see, e.g.,~\cite{misra1992constructive}). Fans are characterized by a pair of admissible colors, and a useful property of fans characterized by the same pair of colors is that they can be processed in parallel. (Processing a fan colors at least one uncolored edge, while possibly recoloring some colored edges.) The algorithm of~\cite{karloff1987efficient} selects a pair of colors that characterizes the largest number of constructed fans. As there are $O\left(\Delta^2\right)$ possible pairs of admissible colors, the collection of selected fans has size $\Omega\left(\frac{|\mathcal{A}|}{\Delta^4}\right)$. These fans are then processed in parallel. As a result, $\Omega\left(\frac{|\mathcal{A}|}{\Delta^4}\right)$ uncolored edges become colored. As originally the number of uncolored edges was $O(|\mathcal{A}|\cdot\Delta)$, it follows that at least $\Omega\left(\frac{1}{\Delta^5}\right)$-fraction of uncolored edges are colored on each phase. Each phase of~\cite{karloff1987efficient} requires $O\left(\log^3 n+\Delta^2\right)$ time, resulting in an overall time of $O\left(\Delta^5\cdot\log n\cdot\left(\log^3 n+\Delta^2\right)\right)$.

\paragraph{The Algorithm of Liang-Shen-Hu}
The algorithm of Liang et al.~\cite{liang1996parallel} starts with splitting the graph into $\approx \Delta$ edge-disjoint subgraphs of constant degree. Each of these subgraphs $G_1,G_2,\ldots,G_\Delta$ is colored by a Vizing coloring (i.e., a coloring that employs $\Delta'+1$ colors, where $\Delta'$ is the maximum degree of the particular subgraph), and then these subgraphs are paired into $\left(G_1,G_2\right),\left(G_3,G_4\right),\ldots$. The edge-coloring obtained for each pair has now a surplus of at most two colors, and the algorithm of~\cite{liang1996parallel} eliminates one color from each such edge-coloring. Now we obtain $\approx \frac{\Delta}{2}$ subgraphs $G_{1,2},G_{3,4},\ldots$ of larger degree, each of which is colored by a Vizing coloring. They are again paired, one color is eliminated, etc. The most time-consuming step in this algorithm is color-elimination. Consider a graph $G=(V,E)$ equipped by a $(\Delta+2)$-edge-coloring. Edges of color $\Delta+2$ are uncolored. Denote this set of edges (that forms a matching) by $F$. Now the algorithm of Liang et al.~\cite{liang1996parallel} creates a graph on $F$, similar to the graph that the algorithm of~\cite{karloff1987efficient} creates on the active vertices, and computes an independent set $I_F$ in this graph. We believe that there is a mistake in the construction and analysis of~\cite{liang1996parallel} of this graph (see Appendix \ref{app: Computing a Large Collection of Pairwise-Disjoint Fans}), and as a result they bound $|I_F|$ by $\Omega\left(\frac{|F|}{\Delta}\right)$, while after correcting the construction, it becomes $\Omega\left(\frac{|F|}{\Delta^2}\right)$. The algorithm of~\cite{liang1996parallel} then creates fans around edges of $I_F$, and finds the largest subcollection of them that are characterized by the same pair of colors. This is done in a similar way to the algorithm of~\cite{karloff1987efficient}, and the fans are processed in a way analogous to that of~\cite{karloff1987efficient}. The cardinality of this collection is then $\Omega\left(\frac{|I_F|/\Delta^2}{\Delta^2}\right) = \Omega\left(\frac{|I_F|}{\Delta^4}\right)$. (They claim erroneously that it is $\Omega\left(\frac{|F|}{\Delta^3}\right)$.) Therefore, the algorithm of~\cite{liang1996parallel} requires $O\left(\Delta^4\cdot\log n\right)$ phases. On each phase they construct fans via an elaborate routine that reduces the problem to the edge-coloring update problem (see Section \ref{Sec: Edge-Coloring Update Problem}), reduce the latter problem to the problem of computing a set of maximal node-disjoint paths, and invoke an algorithm of~\cite{goldberg1993sublinear} for solving the latter problem. This results in running time of $\Tilde{O}\left(\sqrt{\Delta}\right)$ for this step, and overall complexity of $O\left(\Delta^{4.5}\cdot\log^3\Delta\cdot\log n+\Delta^4\cdot\log^4 n\right)$ (because of the aforementioned mistake, they claimed time $O\left(\Delta^{3.5}\cdot\log^3\Delta\cdot\log n+\Delta^4\cdot\log^4 n\right)$).

\paragraph{Our Algorithm}
Our algorithm employs the general framework of~\cite{liang1996parallel}, i.e., we also split the graph $G$ into subgraphs $G_1,\ldots,G_{\Delta}$, pair them, compute colorings with a surplus of two colors for each merged subgraph $G_{1,2},G_{3,4},\ldots$, reduce one color from the coloring of each of these subgraphs, and then proceed to the next iteration by pairing the resulting subgraphs, etc. We first diverge from~\cite{liang1996parallel} in the way that we define a graph $G^{(F)}$ on the uncolored edge set $F$: two edges $e,e'\in F$ are connected if they are at distance at most $2$ from one another (see Equation~(\ref{Equation: edges distance})). This definition guarantees that an independent set $I_F\subseteq F$ will have the property that two fans that correspond to two distinct edges of $I_F$ can be processed in parallel.
The second difference is in the way that the fans are constructed. We devise a direct and very efficient fan-constructing procedure (that also solves the edge update problem much faster than in~\cite{liang1996parallel}).
This procedure starts by constructing a graph $G_{\mathrm{fan}}^{(v,u)}$, where $(v,u)$ is an uncolored edge incident to the center vertex $v$ of the future fan. The vertex set of the graph is $\{u\} \cup \{1,2,\ldots,\Delta+1\}$. The vertex $u$ is connected to an arbitrary missing color $\beta_u$ of $u$ (i.e., a color not used by any edge incident to it) via an arc $\langle u,\beta_u\rangle$. Also, for any color $\alpha$ such that an edge $(v,w)$ incident on $v$ is $\alpha$-colored, we connect $\alpha$ (in $G_{\mathrm{fan}}^{(v,u)}$) to a missing color $\beta_w$ of $w$ via an arc $\langle \alpha,\beta_w\rangle$. It is not hard to see that a maximal path in this graph translates directly to a maximal fan centered at $u$. Such a path can also be very efficiently computed in $\mathrm{PRAM}$. Once this efficient procedure for building fans is employed, the dominating term in the running time of every single phase of 
our algorithm is the time required to compute an MIS in the graph $G^{(F)}$.
We observe that instead of an MIS, it is sufficient to construct a large independent set, and develop a number of efficient procedures for building large independent sets in $G^{(F)}$. To this end, we adapt a number of distributed vertex-coloring algorithms~\cite{barenboim2008sublogarithmic,barenboim2011deterministic, doi:10.1137/12088848X, barenboim2016deterministic, barenboim2018locally} to the parallel setting. These different procedures give rise to various tradeoffs that we obtain for the $(\Delta+1)$-edge-coloring problem.
In particular, if we use the fastest known parallel deterministic MIS algorithm due to~\cite{goldberg1989constructing}, which requires $O\left(\log^3 n\right)$ time and $O\left(\frac{n+ m}{\log n}\right)$ processors, we obtain $O\left(\Delta^4\cdot\log^4 n\right)$ time and $O(m\cdot\Delta)$ processors.
This already improves previous bounds in a wide range of parameters. (The algorithm of~\cite{liang1996parallel} also uses the MIS algorithm by~\cite{goldberg1989constructing}. Our analysis above of the algorithm of~\cite{karloff1987efficient} also assumes that the algorithm of~\cite{goldberg1989constructing} is used as a subroutine.)
But we can also compute an independent set of size $\frac{n}{a^{1+o(1)}}$, where $a$ is the arboricity, in time $O\left(\log^{2+\delta}a\cdot\log n\right)$ time, using $O\left(m\cdot\frac{\log^{\delta}a\cdot\log n}{\log(a\cdot\log n)}\right)$ processors. This results in $(\Delta+1)$-edge-coloring in $\Delta^{3+o(1)}\cdot a^{1+o(1)}\cdot\log^2 n$ time using $O\left(m\cdot\Delta\cdot\frac{\log^{\delta}\Delta\cdot\log n}{\log(\Delta\cdot\log n)}\right)$ processors, for an arbitrarily small $\delta>0$.
We also provide a few additional tradeoffs. See Table \ref{table: A summary of (Delta+1)-edge-coloring routines}.

Generally, time and work complexities of our algorithm depend on the respective complexities of the subroutine for computing large independent sets that it employs. For an $n$-vertex $m$-edge graph $G$ with maximum degree $\Delta$ and arboricity $a$, and a parameter $\lambda$, polynomial in $\Delta$ and/or $a$, we denote the time required for computing a $\lambda$-large independent set of $G$ (i.e., an independent set of size $\Omega\left(\frac{n}{\lambda}\right)$) by $IST_{\lambda}(n,\Delta,a)$, and denote the number of processors that it uses by $m\cdot ISP_{\lambda}(n,\Delta,\alpha)$. In terms of these expressions, our $(\Delta+1)$-edge-coloring algorithm requires $O\left(\lambda\left(\Delta^{2},a\cdot\Delta\right)\cdot \Delta^{2}\cdot\log n \cdot IST_{\lambda}\left(n,\Delta^{2},a\cdot\Delta\right)\right)$ time using $O\left(m\cdot \Delta \cdot ISP_{\lambda}\left(n,\Delta^{2},a\cdot\Delta\right)\right)$ processors.
In these terms, the algorithm of~\cite{karloff1987efficient} requires $O\left(\lambda\left(\Delta^2,\Delta^2\right)\cdot\Delta^3\cdot\log n\cdot\left(IST_{\lambda}\left(n,\Delta^2,\Delta^2\right)+\Delta^2\right)\right)$ time using $O\left(m\cdot\Delta\cdot ISP_{\lambda}\left(n,\Delta^2,\Delta^2\right)\right)$ processors, and the algorithm of~\cite{liang1996parallel} requires\\
$O\left(\lambda\left(\Delta^2,a\cdot\Delta\right)\cdot\Delta^2\cdot\log n\cdot\left(IST_{\lambda}\left(n,\Delta^2,a\cdot\Delta\right)+\sqrt{\Delta}\cdot\log^3\Delta\right)\right)$ time using\\
$O\left(n^2+n\cdot\Delta^3+m\cdot\Delta\cdot ISP_{\lambda}\left(n,\Delta^2,a\cdot\Delta\right)\right)=O\left(n^2+n\cdot\Delta^3\right)$ processors\footnote{The second term dominates the third one whenever $\Delta\geq ISP_{\lambda}\left(n,\Delta^2,a\cdot\Delta\right)$. In all Known routines (see Theorem~\ref{Large independent set alg}), $ISP_{\lambda}\left(n,\Delta^2,a\cdot\Delta\right)\leq \max\{\Delta,\textrm{poly}(\log n)\}$. If $\Delta\leq \textrm{poly}(\log n)$, then the first term dominates the second and the third ones.}. See Table~\ref{table: A summary of new and old parallel (Delta+1)-edge-coloring algorithms, in terms of IST and ISP}.

\begin{table}[h!]
    \begin{center}
    \addtolength{\leftskip} {-2.37cm}
    \addtolength{\rightskip}{-2cm}
    \begin{tabular}{|c | c | c | c|} 
     \hline
     Algorithm & Running Time & \# Processors \\ [0.5ex] 
     \hline
     \cite{karloff1987efficient} &  $O\left(\lambda\left(\Delta^2,\Delta^2\right)\cdot\Delta^3\cdot\log n\cdot\left(IST_{\lambda}\left(n,\Delta^2,\Delta^2\right)+\Delta^2\right)\right)$ & $O\left(m\cdot\Delta\cdot ISP_{\lambda}\left(n,\Delta^2,\Delta^2\right)\right)$  \\
     \hline
     \cite{liang1996parallel} & $O\left(\lambda\left(\Delta^2,a\cdot\Delta\right)\cdot\Delta^2\cdot\log n\cdot\left(IST_{\lambda}\left(n,\Delta^2,a\cdot\Delta\right)+\sqrt{\Delta}\cdot\log^3\Delta\right)\right)$ & $O\left(n^2+n\cdot\Delta^3\right)$\\
     \hline
     \textbf{Ours} & $O\left(\lambda\left(\Delta^{2},a\cdot\Delta\right)\cdot \Delta^{2}\cdot\log n \cdot IST_{\lambda}\left(n,\Delta^{2},a\cdot\Delta\right)\right)$ & $O\left(m\cdot \Delta \cdot ISP_{\lambda}\left(n,\Delta^{2},a\cdot\Delta\right)\right)$\\
     [1ex] 
     \hline
    \end{tabular}
    \end{center}
    \caption{A summary of new and old parallel $(\Delta+1)$-edge-coloring algorithms, in terms of the expressions $IST$ and $ISP$.}
    \label{table: A summary of new and old parallel (Delta+1)-edge-coloring algorithms, in terms of IST and ISP}
\end{table}

\subsection{Related Work}

Edge-coloring problem is a subject of very intensive investigation in the area of distributed computing. See, e.g., \cite{barenboim2011distributed, barenboim2017deterministic, bernshteyn2022fast, balliu2022distributed, bernshteyn2023fast, su2019towards, christiansen2023power}, and the references therein. See also \cite{barenboim2013distributed} for a survey of older work on this fascinating subject. However, only recently the first distributed $(\Delta+1)$-edge-coloring algorithms were devised~\cite{bernshteyn2022fast, christiansen2023power, bernshteyn2023fast}. To the best of our understanding, these algorithms do not translate into efficient parallel algorithms for this fundamental problem. Consider, for example, the state-of-the-art $(\Delta+1)$-edge-coloring algorithm of~\cite{bernshteyn2023fast}. The algorithm hinges on an auxiliary randomized routine (Theorem 8.1,~\cite{bernshteyn2023fast}) that in $O\left(\Delta^{16}\cdot\log n\right)$ distributed time outputs a subset $W$ of expected size $|W|=\Omega\left(\frac{|U|}{\Delta^{20}}\right)$, where $U$ is the set of edges that still need to be colored, along with connected pairwise disjoint $e$-augmenting subgraphs $H_e$ for every $e \in W$. (These subgraphs are \emph{multi-step Vizing chains} - for the sake of this discussion one can think of them as of extensions of classical Vizing chains. See~\cite{bernshteyn2023fast} for details.) Using these subgraphs, one can augment the current edge-coloring so that edges of $W$ will be colored too. As a result, in overall $O\left(\Delta^{20}\cdot\log n\right)$ iterations, each requiring at least $O\left(\Delta^{16}\cdot\log n\right)$ time, one would obtain a $(\Delta+1)$-edge-coloring. However, to compute this edge set $W$, one uses the full power of the distributed $\mathsf{LOCAL}$ model. For every edge $e \in E$, one collects its $O\left(\Delta^{16}\cdot\log n\right)$-neighborhood, and uses it to locally compute a multi-step Vizing chain, i.e., an $e$-augmenting subgraph. Even if this process can be efficiently implemented (say, in $O\left(\Delta^{16}\cdot\log n\right)$ time and $O(m)$ processors), one needs to execute it for all edges $e \in E$ in parallel, blowing up the number of processors to at least $\Omega\left(|E|^2\right)$. In addition, once these $e$-augmenting subgraphs $\{H_e\mid e\in E\}$ are computed, the algorithm of~\cite{bernshteyn2023fast} builds a graph in which two $H_e$'s are connected if and only if they intersect. As the graph contains $|E|$ vertices, computing an MIS on it would require (using the state-of-the-art parallel MIS algorithm of~\cite{goldberg1989constructing}) $O\left(\log^3 n\right)$ time and $O\left(|E|^2\right)$ processors. As a result, one could plausibly obtain a randomized parallel algorithm with running time $O\left(\Delta^{36}\cdot\log^5 n\right)$ and $O\left(|E|^2\right)$ processors, while we propose deterministic parallel algorithms with much smaller running time and number of processors.

\subsection{Structure of the Paper}
In Section~\ref{pram edge coloring} we focus on the $(\Delta+1)$-edge-coloring problem. In Section~\ref{Our alg} we use our new parallel $(\Delta+1)$-edge-coloring algorithm to build a more efficient $(1+\varepsilon)\Delta$-edge-coloring algorithm. Section~\ref{Sections: The Edge-Coloring Update Algorithm} is devoted to our edge-coloring update algorithm. In Appendix~\ref{app: Computing a Large Collection of Pairwise-Disjoint Fans} we describe the flaw in the algorithm of \cite{liang1996parallel}. Some proofs from Section 3 are deferred to Appendix~\ref{App: Some Proofs from Section 3}. Appendix~\ref{App: coloring} contains standard routines for edge-coloring paths and cycles, and Appendix~\ref{maximal path} is devoted to parallel computation of maximal paths. Our adaptations of distributed vertex-coloring algorithms to the parallel setting are provided in Appendix~\ref{Vertex-Coloring Algorithm}.

\section{Preliminaries}\label{sec preliminaries}

Unless stated otherwise, all the graphs in this paper are undirected.

Let $G=(V,E)$ be an undirected graph. For a vertex $v\in V$, denote the set neighbor of $v$ by $N(v)$, and its degree in $G$ by $\deg_G(v)=|N(v)|$.
We denote $|V|=n$, $|E|=m$ and the maximum degree of $G$ by $\Delta(G)=\max_{v\in V}\deg(v)$ (or $\Delta$, if the graph $G$ is clear from the context).

Let $v\in V$ and $e\in E$.
Denote by $G\setminus e$ the graph $G'=(V,E')$, where $E'=E\setminus\{e\}$.\\
Denote by $G\setminus v$ the graph $G'=(V',E')$, where $V'=V\setminus\{v\}$, and $E'=\{e\in E\mid v\notin e\}$.\\
We use the notation $v\in G$ if $v\in V$, and $e\in G$ if $e\in E$.

For a directed graph $G=(V,E)$, we say that a vertex $u$ is an \emph{outgoing neighbor} of $v$, if $\langle v,u\rangle\in E$. For $v\in V$, we denote by $\deg_{\text{out}}(v)$ the out-degree of $v$ in $G$, that is, the number of outgoing neighbors of $v$ in $G$.

\begin{definition}[Eulerian graph]
    A graph $G=(V,E)$ is called \emph{Eulerian} if and only if all its vertices have even degrees.
\end{definition}

\begin{definition}[Adjacent edges]
    Given a graph $G=(V,E)$, we say that two edges $e,e'\in E$ are adjacent, or neighbors of each other, if $e\neq e'$ and they share an endpoint.
\end{definition}

\begin{definition}[Proper edge-coloring]
    A proper $k$-edge-coloring of a graph $G=(V,E)$ is a map $\varphi:E\rightarrow\{1,2,...,k\}$, such that $\varphi(e)\neq\varphi(e')$ for every pair of adjacent edges $e,e'$.
    A proper partial $k$-edge-coloring of a graph $G=(V,E)$ is a proper edge-coloring of a graph $G'=(V,F)$ for some $F\subseteq E$.
\end{definition}

\begin{definition}[Arboricity]
    Given a graph $G=(V,E)$, the \emph{arboricity} $a(G)$ is the minimal number of edge-disjoint forests into which the graph $G$ can be decomposed. Equivalently, $a(G)=\max_{U\subseteq V,|U|\geq 2}\left\{\frac{|E(U)|}{|U|-1}\right\}$~\cite{nash1964decomposition}.
\end{definition}

\begin{definition}[Orientation]
Let $G=(V,E)$ be an undirected graph. An \emph{orientation} $\mu$ of $G$ is an assignment of directions either $\langle u,v\rangle$ or $\langle v,u\rangle$ to each edge $(u,v)$ of the graph.
For a vertex $v\in V$, the \emph{out-degree} of $v$ in the orientation is the number of outgoing edges incident to $v$.
The \emph{out-degree} of an orientation $\mu$ is the largest out-degree among all vertices in the graph.
\end{definition}

\begin{claim}\label{claim: arboricity and orientation}
    Let $G=(V,E)$ be a graph equipped with an orientation with out-degree at most $k$. The arboricity of $G$ is at most $2k$.
\end{claim}

\begin{proof}
Let $G=(V,E)$ be a graph equipped with an orientation with out-degree at most $k$.
Consider an induced subgraph $G[S]=(S,F)$ of $G$, for $S\subseteq V$, $|S|\geq 2$. Since each vertex in $S$ has at most $k$ outgoing edges, we have $|F|\leq k\cdot |S|$. We conclude that $$a(G)=\max_{U\subseteq V,|U|\geq 2}\left\{\frac{|E(U)|}{|U|-1}\right\}\leq 2k.$$
\end{proof}

\begin{definition}[Degeneracy]
Let $G=(V,E)$ be a graph. The \emph{degeneracy} of $G$ is the smallest integer $d$ such that there exists an ordering of its vertices $(v_1,v_2,\dots,v_n)$ such that each vertex $v_i$ has at most $d$ neighbors among $\{v_{i+1},v_{i+2},\dots,v_{n}\}$.
\end{definition}

It is well-known (see, e.g.,~\cite{10.5555/2534493}, Chapter 1) that the degeneracy of a graph is at most twice its arboricity.

\begin{claim}[A bound on the degeneracy]\label{claim: arboricity and degeneracy}
    In a graph $G=(V,E)$ with arboricity $a$ and degeneracy $d$, we have $d\leq 2a-1$.
\end{claim}

We will assume without loss of generality that there are no isolated vertices in the graph, i.e, $m\geq \frac{n}{2}$. Otherwise, we can remove isolated vertices from the graph on each level of recursion.

All our results are stated for the $\mathrm{ARBITRARY\,\,CRCW\,\,PRAM}$ model.

We represent the input graph by adjacency lists, with a separate processor designated to every vertex and to every edge.
To store an edge-coloring $\varphi$, for every edge $e$ the processor dedicated to $e$ stores $\varphi(e)$.
For every vertex $v$, the processor $p_v$ designated to $v$ stores a color missing at $v$ (or $\perp$, if there is no such color).
For every vertex $v$, we also store two hash tables $\mathrm{Color2Edge}(v)$ and $\mathrm{Edge2Color}(v)$, both of length $\deg(v)$.
Given a color $c$, the table $\mathrm{Color2Edge}(v)$ returns in $O(1)$ time whether there is an edge $e$ incident on $v$ which is $\varphi$-colored by $c$.
Given an edge $e$ that is incident on $v$, the table
$\mathrm{Edge2Color}(v)$ returns in $O(1)$ time the color $\varphi(e)$.

Let $G=(V,E)$ be an undirected graph and let $\varphi$ be some fixed partial proper edge-coloring of $G$ with at least $\Delta+1$ colors, that we will use in the sequel. For a vertex $v$ and a color $\alpha\in\{1,2,...,\Delta+1\}$, we say that the color $\alpha$ is \emph{free at $v$} if there is no edge incident on $v$, which is colored (under $\varphi$) by $\alpha$. Denote $M(v)=\{\alpha \text{ $|$ the color $\alpha$ is free at $v$}\}$. Note that for each $v\in V$ there is always at least one free color at $v$.

\section{$(\Delta+1)$-Edge-Coloring}\label{pram edge coloring}
In this section we describe our algorithm for $(\Delta+1)$-edge-coloring problem and analyse it. 

In Section \ref{sec: Vizing Theorem} we present the basic concepts needed for proving Vizing's theorem constructively. Most notably, we define fans, and present a routine for building a singe fan in parallel. Next, in Section \ref{sec: Parallel Fan-Recoloring}, we show how to build many fans in parallel. 
We then proceed (Section \ref{sec: The Edge-Coloring Algorithm}) to describing our parallel $(\Delta+1)$-edge-coloring algorithm.

\subsection{Manipulating with Fans}\label{sec: Vizing Theorem}
The constructive proof of Vizing theorem~\cite{vizing1964estimate} iteratively colors the edges of the graph, and uses a structure called \emph{fan}, defined below. See Figure \ref{fig:fan} for an example of a fan.

\begin{definition}[Fan]
    Let $G=(V,E)$ be an undirected graph and let $\varphi$ be some fixed partial proper edge-coloring of $G$. Let $v\in V$ be a vertex called the fan center. A \emph{fan} $\langle u_1,...,u_k\rangle $ of $v$ with (designated) missing colors $m(v)=m_{\varphi}(v)$ and $\langle m(u_1),...,m(u_k)\rangle=\langle m_{\varphi}(u_1),...,m_{\varphi}(u_k)\rangle$ is an ordered sequence of vertices that satisfies all the following conditions:
    \begin{enumerate}
        \item[(i)] $\langle u_1,...,u_k\rangle$ is a nonempty sequence of distinct neighboring vertices of $v$.
        \item[(ii)] For each $w\in\{v,u_1,...,u_k\}$, the color $m(w)$ is free at $w$.
        \item[(iii)] The edge $(v, u_1)$ is uncolored, and for any $i\in\{2,3,...,k\}$, the edge $(v, u_i)$ is colored $m(u_{i-1})$.
    \end{enumerate}
    \begin{itemize}
        \item The edges $\{(v,u_i)\,|\,i\in\{1,2,...,k\}\}$ are called the edges of the fan.
        \item A fan $\langle u_1,...,u_k\rangle $ of $v$ with missing colors $m(v)$ and $\langle m(u_1),...,m(u_k)\rangle$ is called a \emph{maximal fan} if it cannot be extended, that is, either $m(u_k)$ is free at $v$, or the incident edge of $v$ that is colored $m(u_k)$ is already in the fan.
        \item We say that a fan $\langle u_1,...,u_k\rangle $ of $v$ with missing colors $m(v)$ and $\langle m(u_1),...,m(u_k)\rangle$ is characterized by a pair of colors denoted by $(\alpha(v), \beta(v))$, for $\alpha(v)=m(v)$ and $\beta(v)=m(u_k)$.
    \end{itemize}
\end{definition}

\begin{figure}
    \centering
    \begin{tikzpicture}
        \node[circle] at ({180}:0.3cm)  {$v$};
        \node[circle] at ({90}:2.33cm)  {$u_1$};
        \node[circle] at ({70}:2.3cm)  {$u_2$};
        \node[circle] at ({50}:2.35cm)  {$u_3$};
        \node[circle] at ({30}:2.38cm)  {$u_4$};
        \node[circle] at ({10}:2.4cm)  {$u_5$};
        \node[circle] at ({350}:2.4cm)  {$u_6$};
        \node[circle] at ({330}:2.4cm)  {$u_7$};
        \node[circle,fill=violet] at (360:0mm) (center) {};
        \node[circle,fill=blue] at ({90}:2cm) (n1) {};
        \node[circle,fill=red] at ({70}:2cm) (n2) {};
        \node[circle,fill=pink] at ({50}:2cm) (n3) {};
        \node[circle,fill=olive] at ({30}:2cm) (n4) {};
        \node[circle,fill=cyan] at ({10}:2cm) (n5) {};
        \node[circle,fill=brown] at ({350}:2cm) (n6) {};
        \node[circle,fill=pink] at ({330}:2cm) (n7) {};
    
        \draw[dotted, line width=0.7mm] (center)--(n1);
        \draw[line width=0.7mm, blue] (center)--(n2);
        \draw[line width=0.7mm, red] (center)--(n3);
        \draw[line width=0.7mm, pink] (center)--(n4);
        \draw[line width=0.7mm, olive] (center)--(n5);
        \draw[line width=0.7mm, cyan] (center)--(n6);
        \draw[line width=0.7mm, brown] (center)--(n7);
    \end{tikzpicture}

    \caption{In all figures, the color of each vertex $u_i\in\langle u_1,u_2,...,u_7\rangle$ represents a free (i.e., missing) color at $u_i$. An uncolored edge is represented by a dotted line. This figure depicts a maximal fan $\langle u_1,u_2,...,u_7\rangle$ centered at $v$ with missing colors $\color{violet}\bullet\color{black}$ and $\langle\color{blue}\bullet\color{black},\color{red}\bullet\color{black},\color{pink}\bullet\color{black},\color{olive}\bullet\color{black},\color{cyan}\bullet\color{black},\color{brown}\bullet\color{black},\color{pink}\bullet\color{black}\rangle$, that is characterized by $(\alpha,\beta)=(\color{violet}\bullet\color{black}, \color{pink}\bullet\color{black})$.}
    \label{fig:fan}
\end{figure}
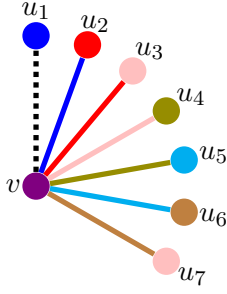

In order to compute a maximal fan with a center $v\in V$, we define an auxiliary directed graph $G^{(v,u)}_{fan}$, where $u$ is a neighbor of $v$ such that the edge $(v,u)$ is uncolored. This graph will have the property that a maximal path in this graph (that is, a simple path that cannot be extended without encountering a vertex that is already on it) starting at the vertex $u\in G^{(v,u)}_{fan}$ corresponds to a maximal fan of $v$ with the uncolored edge $(v,u)$.

\begin{definition}\label{the auxiliary graph G_fan}$\left(\mathrm{The\,\, auxiliary\,\, graph\,\,} G^{(v,u)}_{fan}\right)$\textbf{.}
    Let $G=(V,E)$ be an undirected graph and let $v\in V$ a vertex. Let $\varphi$ be some fixed partial proper edge-coloring of $G$ and let $u\in V$ be a neighbor of $v$ such that the edge $(v,u)$ is uncolored. We define an auxiliary directed graph $G^{(v,u)}_{fan}=\left(V_{fan}^{(v,u)}, E_{fan}^{(v,u)}\right)$ of $v$ over the vertex set $V^{(v,u)}_{fan}=\{u,1,2,...,\Delta+1\}$. The edges of the graph are defined as follows:
    \begin{enumerate}
        \item[(1)] For every $\alpha\in \{1,2,...,\Delta+1\}$, if there exists a neighbor $w\in V$ of $v$ such that the edge $(v, w)$ (in $G$) is $\varphi$-colored $\alpha$, we choose an arbitrary color $\beta\in M(w)$, and define a single directed edge $\langle\alpha, \beta\rangle$.
        \item[(2)] There is a single directed edge $\langle u, \beta_u\rangle$ for an arbitrary color $\beta_u\in M(u)$.
    \end{enumerate}
    Note that $G_{fan}^{(v,u)}$ has maximum out-degree 1.\\
\end{definition}

Before analysing the relationship between $G^{(v,u)}_{fan}$ and a maximal fan of $v$ with an uncolored edge $(v,u)$, we devise an efficient algorithm that given a directed graph $G=(V,E)$, with maximum out-degree at most 1, and a vertex $r\in V$, finds a maximal path in $G$ starting at $r$. We will use this algorithm in the construction of a fan. The description of this algorithm and its analysis appear in Appendix \ref{maximal path}.

\begin{restatable}[Maximal-path algorithm]{lemma}{maxPath}
\label{maximal path algorithm}
    Let $G=(V,E)$ be an $n$-vertex directed graph with maximum out-degree at most 1, and let $r\in V$ a vertex. Procedure \textsc{Maximal-Path} computes a maximal path in $G$ starting at $r$ in $O(\log n)$ time using $O(n)$ processors.
\end{restatable}

In the next lemma we present and analyse an efficient algorithm for constructing a maximal fan of $v$ with an uncolored edge $(v,u)$ using the graph $G^{(v,u)}_{fan}$.

\begin{lemma}[Construction of a maximal fan]\label{Fan construction}
    Let $G=(V,E)$ be a graph with maximum degree $\Delta$ and $v\in V$ be a vertex. Let $\varphi$ be a partial proper edge-coloring of $G$ and let $u_1\in V$ be a neighbor of $v$ such that the edge $(v,u_1)$ is uncolored. A maximal fan centered at $v$ with an uncolored edge $(v,u_1)$ can be computed in $O\left(\log\Delta\right)$ time using $O\left(\deg(v)\right)$ processors.
\end{lemma}

\begin{proof}
    Let $G^{(v,u_1)}_{fan}$ be the graph defined in Definition \ref{the auxiliary graph G_fan}, let $P=\langle u_1,\alpha_1,\alpha_2,...,\alpha_k\rangle$ be a maximal path in the graph $G^{(v,u_1)}_{fan}$ that starts from $u_1$, and let $\alpha\in \{1,2,...,\Delta+1\}$ be a free color at $v$. By the construction of $G^{(v,u_1)}_{fan}$, we know that:
    \begin{itemize}
        \item $\alpha_1$ is free at $u_1$.
        \item For each $i\in \{2,3,...,k\}$, there is a vertex $u_i\in V$ such that the edge $(v,u_i)$ is colored $\alpha_{i-1}$ and $\alpha_{i}$ is free at $u_i$.
    \end{itemize}
    Hence $f=\langle u_1,u_2,...,u_{k}\rangle$ is a fan of $v$ with missing colors $\alpha$ and $\langle \alpha_1,\alpha_2,...,\alpha_{k}\rangle$. This fan is, however, not necessarily a maximal one.\\
    Since $\alpha_k$ is the last vertex in the maximal path $P$ in $G^{(v,u_1)}_{fan}$, then one of the following conditions holds: Either, there is no $\alpha_{k+1}\in \{1,2,...,\Delta+1\}$ such that $\langle\alpha_k,\alpha_{k+1}\rangle\in E_{fan}^{(v,u_1)}$, or there exists $\langle\alpha_k,\alpha_{k+1}\rangle\in E_{fan}^{(v,u_1)}$, such that $\alpha_{k+1}\in \{\alpha_1,\alpha_2,...,\alpha_{k}\}$. We next analyse these two cases.
    \begin{itemize}
        \item[(1)] In the first case $\alpha_{k}$ is free at $v$. (Otherwise, there is a vertex $u_{k+1}$ such that the edge $(v,u_{k+1})$ is colored $\alpha_k$. But then there exists a missing color $\alpha_{k+1}\in M\left(u_{k+1}\right)$ such that $\left\langle\alpha_k,\alpha_{k+1}\right\rangle\in E^{(v,u_1)}_{fan}$.) Hence $\langle u_1,u_2,...,u_{k}\rangle$ is a maximal fan of $v$ with missing colors $\alpha$ and $\langle \alpha_1,\alpha_2,...,\alpha_{k}\rangle$.
        \item[(2)] In the second case there is an edge $(v,u_{k+1})$ that is colored $\alpha_k$, and the color $\alpha_{k+1}$ is free at $u_{k+1}$. In addition, since $\alpha_{k+1}$ already appeared before in the path, the incident edge of $v$ that is colored $\alpha_{k+1}$ is already in the fan. Observe that since $\varphi(v,u_{k+1})=\alpha_k$, and $\alpha_1,\alpha_2,...,\alpha_k$ are all distinct colors $\left(\text{as $P$ is a simple path in $G^{(v,u_1)}_{fan}$}\right)$, it follows that $u_{k+1}\notin\{u_1,u_2,...,u_k\}$. Hence $f'=f\circ(u_{k+1})=\langle u_1,u_2,...,u_{k},u_{k+1}\rangle$ is a maximal fan of $v$ with missing colors $\alpha$ and $\langle \alpha_1,\alpha_2,...,\alpha_{k}, \alpha_{k+1}\rangle$ (see Figure \ref{fig:fan construction}).
    \end{itemize}
    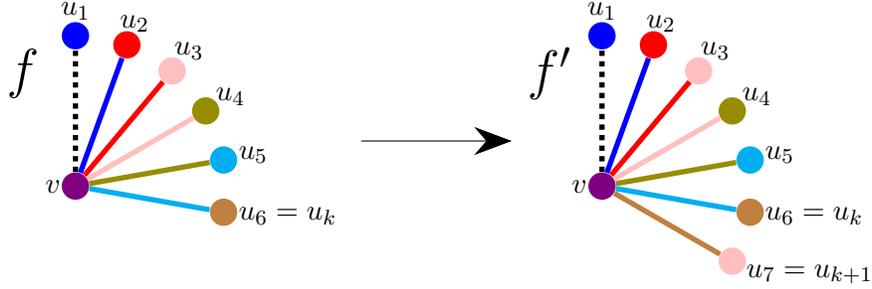
\begin{figure}
    \centering
    \begin{tikzpicture}
        \begin{scope}[xshift=-5cm]
        \node[circle] at ({180}:0.3cm)  {$v$};
        \node[circle] at ({90}:2.33cm)  {$u_1$};
        \node[circle] at ({70}:2.3cm)  {$u_2$};
        \node[circle] at ({50}:2.35cm)  {$u_3$};
        \node[circle] at ({30}:2.38cm)  {$u_4$};
        \node[circle] at ({10}:2.4cm)  {$u_5$};
        \node[circle] at ({352}:2.85cm)  {$u_6=u_k$};
        \node[circle,fill=violet] at (360:0mm) (center) {};
        \node[circle,fill=blue] at ({90}:2cm) (n1) {};
        \node[circle,fill=red] at ({70}:2cm) (n2) {};
        \node[circle,fill=pink] at ({50}:2cm) (n3) {};
        \node[circle,fill=olive] at ({30}:2cm) (n4) {};
        \node[circle,fill=cyan] at ({10}:2cm) (n5) {};
        \node[circle,fill=brown] at ({350}:2cm) (n6) {};
    
        \draw[dotted, line width=0.7mm] (center)--(n1);
        \draw[line width=0.7mm, blue] (center)--(n2);
        \draw[line width=0.7mm, red] (center)--(n3);
        \draw[line width=0.7mm, pink] (center)--(n4);
        \draw[line width=0.7mm, olive] (center)--(n5);
        \draw[line width=0.7mm, cyan] (center)--(n6);

        \end{scope}

        \begin{scope}[xshift=2cm]
        \node[circle] at ({180}:0.3cm)  {$v$};
        \node[circle] at ({90}:2.33cm)  {$u_1$};
        \node[circle] at ({70}:2.3cm)  {$u_2$};
        \node[circle] at ({50}:2.35cm)  {$u_3$};
        \node[circle] at ({30}:2.38cm)  {$u_4$};
        \node[circle] at ({10}:2.4cm)  {$u_5$};
        \node[circle] at ({352}:2.85cm)  {$u_6=u_k$};
        \node[circle] at ({337}:3cm)  {$u_7=u_{k+1}$};
        \node[circle,fill=violet] at (360:0mm) (center) {};
        \node[circle,fill=blue] at ({90}:2cm) (n1) {};
        \node[circle,fill=red] at ({70}:2cm) (n2) {};
        \node[circle,fill=pink] at ({50}:2cm) (n3) {};
        \node[circle,fill=olive] at ({30}:2cm) (n4) {};
        \node[circle,fill=cyan] at ({10}:2cm) (n5) {};
        \node[circle,fill=brown] at ({350}:2cm) (n6) {};
        \node[circle,fill=pink] at ({330}:2cm) (n7) {};
    
        \draw[dotted, line width=0.7mm] (center)--(n1);
        \draw[line width=0.7mm, blue] (center)--(n2);
        \draw[line width=0.7mm, red] (center)--(n3);
        \draw[line width=0.7mm, pink] (center)--(n4);
        \draw[line width=0.7mm, olive] (center)--(n5);
        \draw[line width=0.7mm, cyan] (center)--(n6);
        \draw[line width=0.7mm, brown] (center)--(n7);

        \end{scope}
        
        \draw [-{Stealth[length=5mm]}] (-1.2,0.6) -- (0.8,0.6);
        
        \node at (-5.7,1.5) {\huge $f$};
        \node at (1.3,1.5) {\huge $f'$};

    \end{tikzpicture}

    \caption{The last step of the construction process of a maximal fan. In this example, $\alpha_k=\color{brown}\bullet\color{black}$,  $\alpha_{k+1}=\alpha_3=\color{pink}\bullet\color{black}$ and $(v,u_4)$ is colored $\color{pink}\bullet\color{black}$ and is already in the fan. Also, $(v,u_{k+1})$ is colored $\color{brown}\bullet\color{black}$. This color did not appear in the fan $f$.}
    \label{fig:fan construction}
\end{figure}

    We now analyse the complexity of this construction. 
    \begin{itemize}
        \item For the construction of $G^{(v,u_1)}_{fan}$, for each incident edge $(v,w)$ of $v$, we designate a processor. If the edge is colored $\alpha$, the processor chooses an arbitrary color $\beta\in M(w)$, and adds the directed edge $\langle \alpha,\beta\rangle$ to $G^{(v,u_1)}_{fan}$. In addition, we designate a processor that chooses an arbitrary color $\beta\in M(u_1)$, and add the directed edge $\langle u_1,\beta\rangle$ to $G^{(v,u_1)}_{fan}$. This process requires $O(1)$ time and $O(\deg(v))$ processors.
        \item Observe that the number of vertices in $G^{(v,u_1)}_{fan}$ is $\deg(v)+2\leq \Delta+2$.
        Hence by Lemma \ref{maximal path algorithm}, computing a maximal path in $G^{(v,u_1)}_{fan}$ starting at $u_1$ requires $O\left(\log\Delta\right)$ time using $O(\deg(v))$ processors.
        \item We designate a processor $p(\alpha_i)$ to every color $\alpha_i$ in the computed maximal  path $P=(u_1,\alpha_1,\alpha_2,...,\alpha_k)$. The processor $p(\alpha_1)$ adds the edge $(v,u_1)$ to the fan. Also, the processor $p(\alpha_i)$, for $i\in \{1,2,...,k-1\}$ fetches the vertex $u_{i+1}$ such that $\varphi(v,u_{i+1})=\alpha_i$ (and $\alpha_{i+1}\in M(u_{i+1})$) from the data structure $\mathrm{Color2Edge}(v)$, and adds the edge $(v,u_{i+1})$ to the ($(i+1)$st place of the) fan. If step 2 of the construction was invoked, then the vertex $u_{k+1}$ that was computed on this step is appended to the end of the fan sequence within additional $O(1)$ time. This process requires $O(1)$ time using $O(\deg(v))$ processors.
    \end{itemize}
    We conclude that the construction of a maximal fan centered at $v$ with an uncolored edge $(v,u_1)$ requires $O\left(\log\deg(v)\right)=O\left(\log\Delta\right)$ time using $O\left(\deg (v)\right)$ processors.
\end{proof}

We next describe how fans are used. Recall that a fan contains exactly one uncolored edge. We will use a maximal fan in order to color its uncolored edge. Before we describe this coloring procedure, we present some definitions and notations that we will use later.\\
For a fan $f=\langle u_1,u_2,...,u_k\rangle$ of $v\in V$, we define a \textit{rotation} of this fan by recoloring each edge $(v,u_i)$, for $i\in\{1,2,...,k-1\}$, with the color of the edge $(v,u_{i+1})$ and uncoloring the edge $(v,u_k)$. See Figure \ref{fig:rotation} for an illustration of a rotation. Observe that by the definition of a fan (the color of $(v,u_{i+1})$ must be missing at $u_i$), this process defines a new proper partial edge-coloring of $G$ in which the edge $(v,u_k)$ becomes uncolored, $(v,u_1)$ becomes colored (and all the other colored edges of $f$ get recolored).

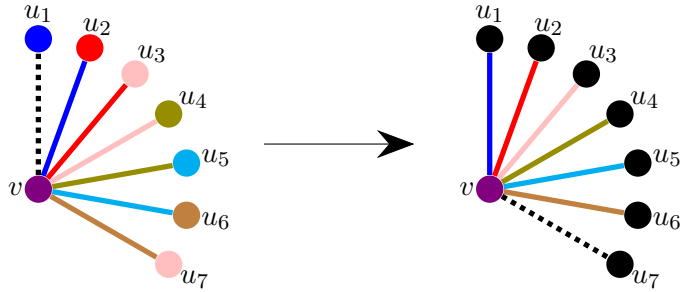
\begin{figure}
    \centering
    \begin{tikzpicture}

        \begin{scope}[xshift=-5cm]
        \node[circle] at ({180}:0.3cm)  {$v$};
        \node[circle] at ({90}:2.33cm)  {$u_1$};
        \node[circle] at ({70}:2.3cm)  {$u_2$};
        \node[circle] at ({50}:2.35cm)  {$u_3$};
        \node[circle] at ({30}:2.38cm)  {$u_4$};
        \node[circle] at ({10}:2.4cm)  {$u_5$};
        \node[circle] at ({350}:2.4cm)  {$u_6$};
        \node[circle] at ({330}:2.4cm)  {$u_7$};
        \node[circle,fill=violet] at (360:0mm) (center) {};
        \node[circle,fill=blue] at ({90}:2cm) (n1) {};
        \node[circle,fill=red] at ({70}:2cm) (n2) {};
        \node[circle,fill=pink] at ({50}:2cm) (n3) {};
        \node[circle,fill=olive] at ({30}:2cm) (n4) {};
        \node[circle,fill=cyan] at ({10}:2cm) (n5) {};
        \node[circle,fill=brown] at ({350}:2cm) (n6) {};
        \node[circle,fill=pink] at ({330}:2cm) (n7) {};
    
        \draw[dotted, line width=0.7mm] (center)--(n1);
        \draw[line width=0.7mm, blue] (center)--(n2);
        \draw[line width=0.7mm, red] (center)--(n3);
        \draw[line width=0.7mm, pink] (center)--(n4);
        \draw[line width=0.7mm, olive] (center)--(n5);
        \draw[line width=0.7mm, cyan] (center)--(n6);
        \draw[line width=0.7mm, brown] (center)--(n7);

        \end{scope}

        \begin{scope}[xshift=1cm]
        \node[circle] at ({180}:0.3cm)  {$v$};
        \node[circle] at ({90}:2.33cm)  {$u_1$};
        \node[circle] at ({70}:2.3cm)  {$u_2$};
        \node[circle] at ({50}:2.35cm)  {$u_3$};
        \node[circle] at ({30}:2.38cm)  {$u_4$};
        \node[circle] at ({10}:2.4cm)  {$u_5$};
        \node[circle] at ({350}:2.4cm)  {$u_6$};
        \node[circle] at ({330}:2.4cm)  {$u_7$};
        \node[circle,fill=violet] at (360:0mm) (center) {};
        \node[circle, fill=black] at ({90}:2cm) (n1) {};
        \node[circle, fill=black] at ({70}:2cm) (n2) {};
        \node[circle, fill=black] at ({50}:2cm) (n3) {};
        \node[circle, fill=black] at ({30}:2cm) (n4) {};
        \node[circle, fill=black] at ({10}:2cm) (n5) {};
        \node[circle, fill=black] at ({350}:2cm) (n6) {};
        \node[circle, fill=black] at ({330}:2cm) (n7) {};
    
        \draw[line width=0.7mm, blue] (center)--(n1);
        \draw[line width=0.7mm, red] (center)--(n2);
        \draw[line width=0.7mm, pink] (center)--(n3);
        \draw[line width=0.7mm, olive] (center)--(n4);
        \draw[line width=0.7mm, cyan] (center)--(n5);
        \draw[line width=0.7mm, brown] (center)--(n6);
        \draw[dotted, line width=0.7mm] (center)--(n7);
        \end{scope}

        \draw [-{Stealth[length=5mm]}] (-2,0.6) -- (0,0.6);
    \end{tikzpicture}

    \caption{Rotation of a fan $\langle u_1,u_2,...,u_7\rangle$ centered at $v$. In all figures, black color indicates a "wild-card", i.e., an undetermined missing color.}
    \label{fig:rotation}
\end{figure}

Let $\alpha,\beta\in \{1,2,...,\Delta+1\}$ be two colors. We define the graph $G_{\alpha,\beta}=(V,E_{\alpha,\beta})$, where $E_{\alpha,\beta}=\{e\in E\,|\,\varphi(e)=\alpha \text{ or } \varphi(e)=\beta\}$. Observe that since $\varphi$ is a proper edge-coloring, then $\Delta\left(G_{\alpha,\beta}\right)\leq 2$ (each vertex has at most one incident edge of each color). Hence $G_{\alpha,\beta}$ consists of a collection of simple paths and even length cycles. We call a path connected component in $G_{\alpha,\beta}$ an \textit{$\alpha\beta$-path} (an $\alpha\beta$-path might be a single vertex, for vertices in which both $\alpha$ and $\beta$ are missing). For an $\alpha\beta$-path $P=(w_1,w_2,...,w_l)$ in $G_{\alpha,\beta}$, we define \textit{exchanging} of this path to be recoloring all $\alpha$-colored edges of $P$ by $\beta$, and vice versa. For a color $\gamma\in\{\alpha,\beta\}$, let $\overline{\gamma}$ denote the color such that $\{\overline{\gamma}\}=\{\alpha,\beta\}\setminus\{\gamma\}$. An $\alpha\beta$-path $P$ is necessarily a maximal one, i.e., $\overline{\varphi(w_1,w_2)}\in M(w_1)$ and $\overline{\varphi(w_{l-1},w_l)}\in M(w_l)$. Exchanging of $P$ defines a new partial proper edge-coloring of $G$, in which $\varphi(w_1,w_2)$ is now free at $w_1$, $\varphi(w_{l-1},w_l)$ is free at $w_l$, and $M(w_i)$ is unchanged for each $i\in\{2,3,...,l-1\}$. Note that after the exchanging, the edge-coloring is still proper. 
\begin{observation}\label{path endpoint}
    For a vertex $v\in V$ such that $\left|\{\alpha, \beta\}\cap M(v)\right|\geq 1$, $v$ is an endpoint of an $\alpha\beta$-path $P_{\alpha,\beta}$ in $G_{\alpha,\beta}$. (If $\alpha,\beta$ are both in $M(v)$, then $P_{\alpha,\beta}=\{v\}$.)
\end{observation} 
For an endpoint $w$ of an $\alpha\beta$-path $P$, we may refer to $P$ as \textit{the $\alpha\beta$-path of $w$}. For a fan $f=\langle u_1,u_2,...,u_k\rangle$ centered at some $v\in V$ and characterized by $(\alpha,\beta)$, we also refer to the $\alpha\beta$-path of $v$ as \textit{the $\alpha\beta$-path of $f$} (see Figure \ref{fig:alpha beta path}).

    \begin{figure}
    \centering
    \begin{tikzpicture}
    \begin{scope}[xshift=-3cm]
        \node[circle] at ({180}:0.3cm)  {$v$};
        \node[circle] at ({90}:2.33cm)  {$u_1$};
        \node[circle] at ({70}:2.3cm)  {$u_2$};
        \node[circle] at ({50}:2.35cm)  {$u_3$};
        \node[circle] at ({25}:2.38cm)  {$u_4$};
        \node[circle] at ({10}:2.4cm)  {$u_5$};
        \node[circle] at ({350}:2.4cm)  {$u_6$};
        \node[circle] at ({330}:2.4cm)  {$u_7$};
        \node[circle,fill=violet] at (360:0mm) (center) {};
        \node[circle,fill=blue] at ({90}:2cm) (n1) {};
        \node[circle,fill=red] at ({70}:2cm) (n2) {};
        \node[circle,fill=pink] at ({50}:2cm) (n3) {};
        \node[circle,fill=olive] at ({30}:2cm) (n4) {};
        \node[circle,fill=cyan] at ({10}:2cm) (n5) {};
        \node[circle,fill=brown] at ({350}:2cm) (n6) {};
        \node[circle,fill=pink] at ({330}:2cm) (n7) {};
    
        \draw[dotted, line width=0.7mm] (center)--(n1);
        \draw[line width=0.7mm, blue] (center)--(n2);
        \draw[line width=0.7mm, red] (center)--(n3);
        \draw[line width=0.7mm, pink] (center)--(n4);
        \draw[line width=0.7mm, olive] (center)--(n5);
        \draw[line width=0.7mm, cyan] (center)--(n6);
        \draw[line width=0.7mm, brown] (center)--(n7);

        \node[circle, fill=black] at (3,2) (n8) {};
        \node[circle, fill=black] at (3.5,0.5) (n9) {};
        \node[circle, fill=black] at (4,2) (n10) {};
        \node[circle, fill=black] at (4.5,0.5) (n11) {};
        \node[circle, fill=black] at (5,2) (n12) {};
        \node[circle, fill=black] at (5.5,0.5) (n13) {};

        \draw[line width=0.7mm, violet] (n4)--(n8);
        \draw[line width=0.7mm, pink] (n8)--(n9);
        \draw[line width=0.7mm, violet] (n9)--(n10);
        \draw[line width=0.7mm, pink] (n10)--(n11);
        \draw[line width=0.7mm, violet] (n11)--(n12);
        \draw[line width=0.7mm, pink] (n12)--(n13);
    \end{scope}

        \node at (-3.8,1) {\Large$f$};
        \node at (1,2.8) {\Large$P$};
    \end{tikzpicture}

    \caption{$P$ is the $\alpha\beta$-path of a maximal fan $f$. The fan is characterized by $(\alpha,\beta)=(\color{violet}\bullet\color{black},\color{pink}\bullet\color{black})$.}
    \label{fig:alpha beta path}
\end{figure}
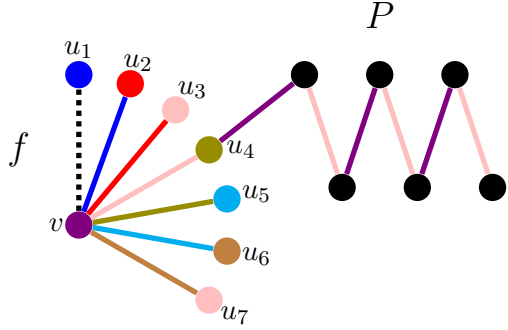

\begin{definition}\label{def: special vertices}
    For a maximal fan $\langle u_1,u_2,...,u_k\rangle$ centered at $v\in V$ and characterized by $(\alpha,\beta)$, if there exists an index $i\in \{2,3,...,k-1\}$ such that $(v,u_i)$ is colored $\beta$, define $x(v)=u_i$, $y(v)=u_k$ and $z(v)=u_{i-1}$. Otherwise, there are no edges colored $\beta$ that are incident on $v$. (Recall that $m(u_k)=\beta$, and thus $\varphi(v,u_k)\neq \beta$. Also, if there were an edge $(v,w)$ colored by $\beta$, it could be added to the fan $f$ as $w=u_{k+1}$, contradicting the maximality of $f$.) In this case we define $x(v)=\emptyset$, $y(v)=u_k$ and $z(v)=\emptyset$.
\end{definition}
 Note that if $x(v),z(v)\neq\emptyset$, then $x(v)$ and $z(v)$ are two consecutive vertices in the fan ($z(v)=u_{i-1}$ and $x(v)=u_i$ for some $i\in\{2,3,...,k-1\}$), and $x(v),y(v)$ and $z(v)$ are three distinct vertices in the fan. See Figure \ref{fig:special vertices} for an illustration.

\begin{figure}
    \centering
    \begin{tikzpicture}

        \begin{scope}[xshift=-4cm]
        \node[circle] at ({180}:0.3cm)  {$v$};
        \node[circle] at ({90}:2.33cm)  {$u_1$};
        \node[circle] at ({70}:2.33cm)  {$u_2$};
        \node[circle] at ({37}:2.8cm)  {$u_3=z(v)$};
        \node[circle] at ({22}:2.9cm)  {$u_4=x(v)$};
        \node[circle] at ({10}:2.4cm)  {$u_5$};
        \node[circle] at ({350}:2.4cm)  {$u_6$};
        \node[circle] at ({338}:2.9cm)  {$u_7=y(v)$};
        \node[circle,fill=violet] at (360:0mm) (center) {};
        \node[circle,fill=blue] at ({90}:2cm) (n1) {};
        \node[circle,fill=red] at ({70}:2cm) (n2) {};
        \node[circle,fill=pink] at ({50}:2cm) (n3) {};
        \node[circle,fill=olive] at ({30}:2cm) (n4) {};
        \node[circle,fill=cyan] at ({10}:2cm) (n5) {};
        \node[circle,fill=brown] at ({350}:2cm) (n6) {};
        \node[circle,fill=pink] at ({330}:2cm) (n7) {};
    
        \draw[dotted, line width=0.7mm] (center)--(n1);
        \draw[line width=0.7mm, blue] (center)--(n2);
        \draw[line width=0.7mm, red] (center)--(n3);
        \draw[line width=0.7mm, pink] (center)--(n4);
        \draw[line width=0.7mm, olive] (center)--(n5);
        \draw[line width=0.7mm, cyan] (center)--(n6);
        \draw[line width=0.7mm, brown] (center)--(n7);
    
        \end{scope}

        \begin{scope}[xshift=2cm]
        \node[circle] at ({180}:0.3cm)  {$v$};
        \node[circle] at ({90}:2.33cm)  {$u_1$};
        \node[circle] at ({70}:2.33cm)  {$u_2$};
        \node[circle] at ({47}:2.39cm)  {$u_3$};
        \node[circle] at ({28}:2.4cm)  {$u_4$};
        \node[circle] at ({10}:2.4cm)  {$u_5$};
        \node[circle] at ({350}:2.4cm)  {$u_6$};
        \node[circle] at ({338}:2.9cm)  {$u_7=y(v)$};
        \node[circle,fill=violet] at (360:0mm) (center) {};
        \node[circle,fill=blue] at ({90}:2cm) (n1) {};
        \node[circle,fill=red] at ({70}:2cm) (n2) {};
        \node[circle,fill=purple] at ({50}:2cm) (n3) {};
        \node[circle,fill=olive] at ({30}:2cm) (n4) {};
        \node[circle,fill=cyan] at ({10}:2cm) (n5) {};
        \node[circle,fill=brown] at ({350}:2cm) (n6) {};
        \node[circle,fill=pink] at ({330}:2cm) (n7) {};
    
        \draw[dotted, line width=0.7mm] (center)--(n1);
        \draw[line width=0.7mm, blue] (center)--(n2);
        \draw[line width=0.7mm, red] (center)--(n3);
        \draw[line width=0.7mm, purple] (center)--(n4);
        \draw[line width=0.7mm, olive] (center)--(n5);
        \draw[line width=0.7mm, cyan] (center)--(n6);
        \draw[line width=0.7mm, brown] (center)--(n7);
    
        \end{scope}

        \node at(5.2,1.5){$x(v)=\emptyset$};
        \node at(5.2,2){$z(v)=\emptyset$};

    \end{tikzpicture}

    \caption{Two examples of a fan $\langle u_1,u_2,...,u_7\rangle$ centered at $v$ and characterized by $(\alpha,\beta)=(\color{violet}\bullet,\color{pink}\bullet\color{black})$ with its special vertices $x(v)$, $y(v)$ and $z(v)$, in the two possible cases.}
    \label{fig:special vertices}
\end{figure}
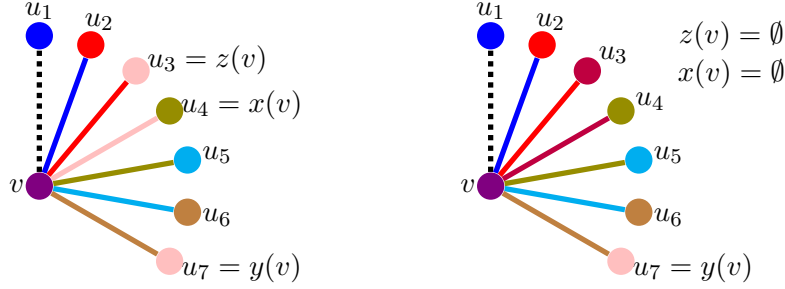

In order to compute $\alpha\beta$-paths, we use a result of Shiloach and Vishkin~\cite{shiloach1982logn} that computes connected components of a graph.

\begin{restatable}[Connected components~\cite{shiloach1982logn}]{lemma}{connComp}
\label{Connected components algorithm}
    Let $G=(V,E)$ be an $n$-vertex $m$-edge graph. There is a deterministic $\mathrm{CRCW\,\,PRAM}$ algorithm that computes connected components of $G$ in $O\left(\log n\right)$ time using $O\left(n+m\right)$ processors.
\end{restatable}

We are now ready to describe how we use maximal fans and $\alpha\beta$-paths in order to color a single uncolored edge.
Let $\langle u_1,...,u_k\rangle$ be a maximal fan of $v\in V$ that is characterized by $(\alpha,\beta)$. We present an algorithm by Misra and Gries~\cite{misra1992constructive} (based on Vizing's constructive proof~\cite{vizing1964estimate}) that defines a new partial proper edge-coloring $\varphi'$ of the graph $G$. The support of $\varphi'$ (i.e., the set of colored edges) contains the support of the original coloring $\varphi$, and also one more edge $(v,u_1)$. The algorithm is described below in Procedure \textsc{Recolor-Fan}. See also Figure \ref{fig:fan recolor} for an illustration of the algorithm.\\
\\$\textsc{Procedure}$ $\textsc{Recolor-Fan}\text{ }(G=(V,E),\varphi, \left(v,\langle u_1,...,u_k\rangle, (\alpha,\beta)\right))$
    \begin{description}
        \item[\textbf{Step 1.}] Exchange the $\alpha\beta$-path $P_v$ of $v$ so that $\beta$ becomes free at $v$. (By Observation \ref{path endpoint}, such a path is well-defined.)
        \item[\textbf{Step 2.}] Let $i\in\{1,2,...,k\}$ 
        be an index such that $\beta$ is the missing color of $u_i$, under both the original coloring $\varphi$, and the coloring $\varphi'$ obtained as a result of the exchanging of $P_v$. (If $z(v)\neq \emptyset$, and $\beta$ is still free at $z(v)$ after step 1, then this is the vertex $z(v)$. Otherwise, it is the vertex $y(v)$.) Then $\langle u_1,...,u_i\rangle$ is a fan of $v$. Rotate the fan $\langle u_1,...,u_i\rangle$, and color the edge $(v,u_i)$ with the color $\beta$.
    \end{description}

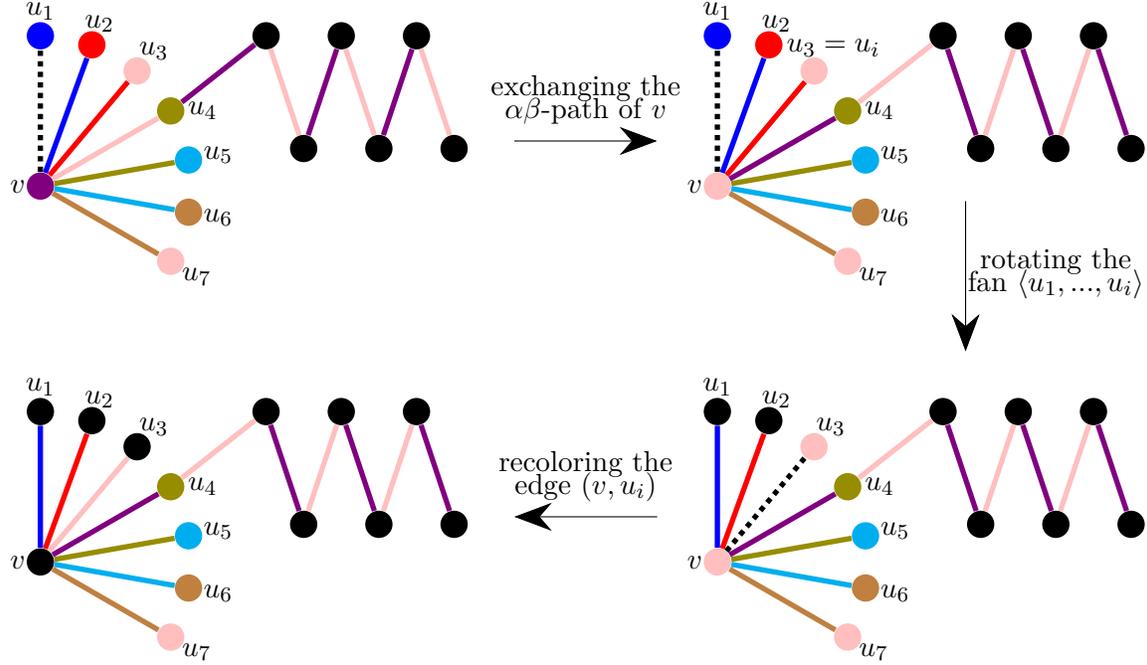
\begin{figure}
    \centering
    \begin{tikzpicture}

        \begin{scope}[xshift=-7cm]
        \node[circle] at ({180}:0.3cm)  {$v$};
        \node[circle] at ({90}:2.33cm)  {$u_1$};
        \node[circle] at ({70}:2.3cm)  {$u_2$};
        \node[circle] at ({50}:2.35cm)  {$u_3$};
        \node[circle] at ({25}:2.38cm)  {$u_4$};
        \node[circle] at ({10}:2.4cm)  {$u_5$};
        \node[circle] at ({350}:2.4cm)  {$u_6$};
        \node[circle] at ({330}:2.4cm)  {$u_7$};
        \node[circle,fill=violet] at (360:0mm) (center) {};
        \node[circle,fill=blue] at ({90}:2cm) (n1) {};
        \node[circle,fill=red] at ({70}:2cm) (n2) {};
        \node[circle,fill=pink] at ({50}:2cm) (n3) {};
        \node[circle,fill=olive] at ({30}:2cm) (n4) {};
        \node[circle,fill=cyan] at ({10}:2cm) (n5) {};
        \node[circle,fill=brown] at ({350}:2cm) (n6) {};
        \node[circle,fill=pink] at ({330}:2cm) (n7) {};
    
        \draw[dotted, line width=0.7mm] (center)--(n1);
        \draw[line width=0.7mm, blue] (center)--(n2);
        \draw[line width=0.7mm, red] (center)--(n3);
        \draw[line width=0.7mm, pink] (center)--(n4);
        \draw[line width=0.7mm, olive] (center)--(n5);
        \draw[line width=0.7mm, cyan] (center)--(n6);
        \draw[line width=0.7mm, brown] (center)--(n7);

        \node[circle, fill=black] at (3,2) (n8) {};
        \node[circle, fill=black] at (3.5,0.5) (n9) {};
        \node[circle, fill=black] at (4,2) (n10) {};
        \node[circle, fill=black] at (4.5,0.5) (n11) {};
        \node[circle, fill=black] at (5,2) (n12) {};
        \node[circle, fill=black] at (5.5,0.5) (n13) {};

        \draw[line width=0.7mm, violet] (n4)--(n8);
        \draw[line width=0.7mm, pink] (n8)--(n9);
        \draw[line width=0.7mm, violet] (n9)--(n10);
        \draw[line width=0.7mm, pink] (n10)--(n11);
        \draw[line width=0.7mm, violet] (n11)--(n12);
        \draw[line width=0.7mm, pink] (n12)--(n13);

        \end{scope}

        \begin{scope}[xshift=2cm]
        \node[circle] at ({180}:0.3cm)  {$v$};
        \node[circle] at ({90}:2.33cm)  {$u_1$};
        \node[circle] at ({70}:2.3cm)  {$u_2$};
        \node[circle] at ({50}:2.4cm)  {$u_3=u_i$};
        \node[circle] at ({25}:2.38cm)  {$u_4$};
        \node[circle] at ({10}:2.4cm)  {$u_5$};
        \node[circle] at ({350}:2.4cm)  {$u_6$};
        \node[circle] at ({330}:2.4cm)  {$u_7$};
        \node[circle,fill=pink] at (360:0mm) (center) {};
        \node[circle,fill=blue] at ({90}:2cm) (n1) {};
        \node[circle,fill=red] at ({70}:2cm) (n2) {};
        \node[circle,fill=pink] at ({50}:2cm) (n3) {};
        \node[circle,fill=olive] at ({30}:2cm) (n4) {};
        \node[circle,fill=cyan] at ({10}:2cm) (n5) {};
        \node[circle,fill=brown] at ({350}:2cm) (n6) {};
        \node[circle,fill=pink] at ({330}:2cm) (n7) {};
    
        \draw[dotted, line width=0.7mm] (center)--(n1);
        \draw[line width=0.7mm, blue] (center)--(n2);
        \draw[line width=0.7mm, red] (center)--(n3);
        \draw[line width=0.7mm, violet] (center)--(n4);
        \draw[line width=0.7mm, olive] (center)--(n5);
        \draw[line width=0.7mm, cyan] (center)--(n6);
        \draw[line width=0.7mm, brown] (center)--(n7);

        \node[circle, fill=black] at (3,2) (n8) {};
        \node[circle, fill=black] at (3.5,0.5) (n9) {};
        \node[circle, fill=black] at (4,2) (n10) {};
        \node[circle, fill=black] at (4.5,0.5) (n11) {};
        \node[circle, fill=black] at (5,2) (n12) {};
        \node[circle, fill=black] at (5.5,0.5) (n13) {};

        \draw[line width=0.7mm, pink] (n4)--(n8);
        \draw[line width=0.7mm, violet] (n8)--(n9);
        \draw[line width=0.7mm, pink] (n9)--(n10);
        \draw[line width=0.7mm, violet] (n10)--(n11);
        \draw[line width=0.7mm, pink] (n11)--(n12);
        \draw[line width=0.7mm, violet] (n12)--(n13);

        \end{scope}

        \begin{scope}[yshift=-5cm]
        \begin{scope}[xshift=2cm]
        \node[circle] at ({180}:0.3cm)  {$v$};
        \node[circle] at ({90}:2.33cm)  {$u_1$};
        \node[circle] at ({70}:2.3cm)  {$u_2$};
        \node[circle] at ({50}:2.35cm)  {$u_3$};
        \node[circle] at ({25}:2.38cm)  {$u_4$};
        \node[circle] at ({10}:2.4cm)  {$u_5$};
        \node[circle] at ({350}:2.4cm)  {$u_6$};
        \node[circle] at ({330}:2.4cm)  {$u_7$};
        \node[circle,fill=pink] at (360:0mm) (center) {};
        \node[circle, fill=black] at ({90}:2cm) (n1) {};
        \node[circle, fill=black] at ({70}:2cm) (n2) {};
        \node[circle, fill=pink] at ({50}:2cm) (n3) {};
        \node[circle, fill=olive] at ({30}:2cm) (n4) {};
        \node[circle, fill=cyan] at ({10}:2cm) (n5) {};
        \node[circle, fill=brown] at ({350}:2cm) (n6) {};
        \node[circle, fill=pink] at ({330}:2cm) (n7) {};

        \node[circle, fill=black] at (3,2) (n8) {};
        \node[circle, fill=black] at (3.5,0.5) (n9) {};
        \node[circle, fill=black] at (4,2) (n10) {};
        \node[circle, fill=black] at (4.5,0.5) (n11) {};
        \node[circle, fill=black] at (5,2) (n12) {};
        \node[circle, fill=black] at (5.5,0.5) (n13) {};

        \draw[line width=0.7mm, pink] (n4)--(n8);
        \draw[line width=0.7mm, violet] (n8)--(n9);
        \draw[line width=0.7mm, pink] (n9)--(n10);
        \draw[line width=0.7mm, violet] (n10)--(n11);
        \draw[line width=0.7mm, pink] (n11)--(n12);
        \draw[line width=0.7mm, violet] (n12)--(n13);

        \draw[line width=0.7mm, blue] (center)--(n1);
        \draw[line width=0.7mm, red] (center)--(n2);
        \draw[line width=0.7mm, violet] (center)--(n4);
        \draw[line width=0.7mm, olive] (center)--(n5);
        \draw[line width=0.7mm, cyan] (center)--(n6);
        \draw[line width=0.7mm, brown] (center)--(n7);
        \draw[dotted, line width=0.7mm] (center)--(n3);
        \end{scope}
        \end{scope}

        \begin{scope}[yshift=-5cm]
        \begin{scope}[xshift=-7cm]
        \node[circle] at ({180}:0.3cm)  {$v$};
        \node[circle] at ({90}:2.33cm)  {$u_1$};
        \node[circle] at ({70}:2.3cm)  {$u_2$};
        \node[circle] at ({50}:2.35cm)  {$u_3$};
        \node[circle] at ({25}:2.38cm)  {$u_4$};
        \node[circle] at ({10}:2.4cm)  {$u_5$};
        \node[circle] at ({350}:2.4cm)  {$u_6$};
        \node[circle] at ({330}:2.4cm)  {$u_7$};
        \node[circle,fill=black] at (360:0mm) (center) {};
        \node[circle, fill=black] at ({90}:2cm) (n1) {};
        \node[circle, fill=black] at ({70}:2cm) (n2) {};
        \node[circle, fill=black] at ({50}:2cm) (n3) {};
        \node[circle, fill=olive] at ({30}:2cm) (n4) {};
        \node[circle, fill=cyan] at ({10}:2cm) (n5) {};
        \node[circle, fill=brown] at ({350}:2cm) (n6) {};
        \node[circle, fill=pink] at ({330}:2cm) (n7) {};

        \node[circle, fill=black] at (3,2) (n8) {};
        \node[circle, fill=black] at (3.5,0.5) (n9) {};
        \node[circle, fill=black] at (4,2) (n10) {};
        \node[circle, fill=black] at (4.5,0.5) (n11) {};
        \node[circle, fill=black] at (5,2) (n12) {};
        \node[circle, fill=black] at (5.5,0.5) (n13) {};

        \draw[line width=0.7mm, pink] (n4)--(n8);
        \draw[line width=0.7mm, violet] (n8)--(n9);
        \draw[line width=0.7mm, pink] (n9)--(n10);
        \draw[line width=0.7mm, violet] (n10)--(n11);
        \draw[line width=0.7mm, pink] (n11)--(n12);
        \draw[line width=0.7mm, violet] (n12)--(n13);

        \draw[line width=0.7mm, blue] (center)--(n1);
        \draw[line width=0.7mm, red] (center)--(n2);
        \draw[line width=0.7mm, violet] (center)--(n4);
        \draw[line width=0.7mm, olive] (center)--(n5);
        \draw[line width=0.7mm, cyan] (center)--(n6);
        \draw[line width=0.7mm, brown] (center)--(n7);
        \draw[line width=0.7mm, pink] (center)--(n3);
        \end{scope}
        \end{scope}

        \node at (0.25,1.3) {exchanging the};
        \node at (0.25,1) {$\alpha\beta$-path of $v$};
        \node at (6.5,-1) {rotating the};
        \node at (6.5,-1.3) {fan $\langle u_1,...,u_i\rangle$};
        \node at (0.25,-3.7) {recoloring the};
        \node at (0.25,-4) {edge $(v,u_i)$};
        \draw [-{Stealth[length=5mm]}] (-0.7,0.6) -- (1.2,0.6);
        \draw [-{Stealth[length=5mm]}] (1.2,-4.4) -- (-0.7,-4.4);
        \draw [-{Stealth[length=5mm]}] (5.3,-0.2) -- (5.3,-2.2);

    \end{tikzpicture}

    \caption{Procedure \textsc{Recolor-Fan} applied on a fan $\langle u_1,u_2,...,u_7\rangle$ centered at $v$ and characterized by $(\alpha,\beta)=(\color{violet}\bullet\color{black}, \color{pink}\bullet\color{black})$.}
    \label{fig:fan recolor}
\end{figure}

We now show that Procedure \textsc{Recolor-Fan} is well-defined (i.e., colors the edges of the fan and the $\alpha\beta$-path properly), and analyse its complexity.
\begin{theorem}[Procedure \textsc{Recolor-Fan}]\label{Procedure Recolor-Fan}
    Let $G$ be an $n$-vertex $m$-edge graph, $\varphi$ be a partial proper edge-coloring of $G$, and $\langle u_1,...,u_k\rangle$ be a maximal fan of $v\in V$ characterized by a pair of colors $(\alpha,\beta)$. Then Procedure \textsc{Recolor-Fan} colors the edges of the fan and its $\alpha\beta$-path using colors from the palette $\{1,2,...,\Delta+1\}$. The resulting partial coloring $\varphi'$ is proper as well. In addition, step 1 of Procedure \textsc{Recolor-Fan} can be applied on various $\alpha\beta$-paths in parallel and it requires $O(\log n)$ time using $O(n)$ processors, and step 2 of Procedure \textsc{Recolor-Fan} requires $O(1)$ time using $O(\deg(v))$ processors.
\end{theorem}

\begin{proof}
    First, observe that after the exchanging of the $\alpha\beta$-path of $v$, the color $\beta$ is free at $v$. Recall that the exchanging of a path affects only free colors of its endpoints.
    Observe that by definition of fan, for each $u_j\in\{u_1,...,u_{k-1}\}\setminus \{z(v)\}$, we have $m(u_j)\neq \beta$. (As only $(v,x(v))$ may be $\varphi$-colored $\beta$, and for $j\in\{1,2,...,k-1\}$, $m(u_j)=\beta$ implies $\varphi(v,u_{j+1})=\beta$.) Also, for each $u_j\in \{ u_1,...,u_{k-1}\}$, we have $m(u_j)\neq \alpha$. (As $\alpha$ is free at $v$, and for $j\in\{1,2,...,k-1\}$, $m(u_j)=\alpha$ implies $\varphi(v,u_{j+1})=\alpha$, contradiction.) Hence exchanging of the $\alpha\beta$-path of $v$ does not affect the missing color $m(u_j)$ of each $u_j\in\{ u_1,...,u_{k-1}\}\setminus \{z(v)\}$. We now consider three cases:
    \begin{itemize}
        \item[(1)] If $z(v)$ and $x(v)$ do not exist (i.e., $x(v)=z(v)=\emptyset$), then there is no incident edge on $v$ that is $\varphi$-colored $\beta$. Hence the $\alpha\beta$-path of $v$ is empty, and its exchanging changes nothing. Hence $\langle u_1,...,u_k\rangle$ is still a fan of $v$, and $\beta$ remains free at $u_k$. Therefore, step 2 of Procedure \textsc{Recolor-Fan} is well-defined and can be performed.
        \item[(2)] Consider the case that $z(v)$ is an endpoint of the $\alpha\beta$-path $P_{\alpha,\beta}$ of $v$. (Recall that $m(z(v))=\beta$ and $\varphi(v,x(v))=\beta$. So if $z(v)$ is an endpoint of the $\alpha\beta$-path of $v$, it means that it has an edge colored $\alpha$ that belongs to the $\alpha\beta$-path incident on it.) Then after the exchanging of the $\alpha\beta$-path of $v$, $\alpha$ will become free at $z(v)$, and the edge $(v,x(v))$ will be colored $\alpha$. Hence $\langle u_1,...,u_k\rangle$ is still a fan of $v$. (Indeed, $z(v)$ and $x(v)$ are consecutive vertices in the fan, $\alpha=m_{\varphi'}(z(v))$ is the $\varphi'$-color of $(v,x(v))$, and the colors of edges $(v,u_j)$ and missing colors $m(u_j)$ of other vertices $u_j$ in the fan are unchanged. Note that $\beta$ is missing at $y(v)$, and thus $y(v)$ cannot appear on $P_{\alpha,\beta}$). Since the $\alpha\beta$-path of $v$ ends at $z(v)$, the color $\beta$ remains free at $u_k\neq z(v)$ (see Figure \ref{fig:alpha beta path collision} for an illustration). Also, $\beta$ is now free at $v$. Hence, step 2 of Procedure \textsc{Recolor-Fan} is well-defined and can be performed.
            \begin{figure}
    \centering
    \begin{tikzpicture}
    \begin{scope}[xshift=-7cm]
        \node[circle] at ({180}:0.3cm)  {$v$};
        \node[circle] at ({90}:2.33cm)  {$u_1$};
        \node[circle] at ({70}:2.3cm)  {$u_2$};
        \node[circle] at ({35}:2.75cm)  {$u_3=z(v)$};
        \node[circle] at ({22}:2.9cm)  {$u_4=x(v)$};
        \node[circle] at ({10}:2.4cm)  {$u_5$};
        \node[circle] at ({350}:2.4cm)  {$u_6$};
        \node[circle] at ({338}:2.9cm)  {$u_7=y(v)$};
        \node[circle,fill=violet] at (360:0mm) (center) {};
        \node[circle,fill=blue] at ({90}:2cm) (n1) {};
        \node[circle,fill=red] at ({70}:2cm) (n2) {};
        \node[circle,fill=pink] at ({50}:2cm) (n3) {};
        \node[circle,fill=olive] at ({30}:2cm) (n4) {};
        \node[circle,fill=cyan] at ({10}:2cm) (n5) {};
        \node[circle,fill=brown] at ({350}:2cm) (n6) {};
        \node[circle,fill=pink] at ({330}:2cm) (n7) {};
    
        \draw[dotted, line width=0.7mm] (center)--(n1);
        \draw[line width=0.7mm, blue] (center)--(n2);
        \draw[line width=0.7mm, red] (center)--(n3);
        \draw[line width=0.7mm, pink] (center)--(n4);
        \draw[line width=0.7mm, olive] (center)--(n5);
        \draw[line width=0.7mm, cyan] (center)--(n6);
        \draw[line width=0.7mm, brown] (center)--(n7);

        \node[circle, fill=black] at (3.5,0.5) (n8) {};
        \node[circle, fill=black] at (5.5,0.7) (n9) {};
        \node[circle, fill=black] at (5,2.5) (n10) {};
        \node[circle, fill=black] at (3,2.7) (n11) {};

        \draw[line width=0.7mm, violet] (n4)--(n8);
        \draw[line width=0.7mm, pink] (n8)--(n9);
        \draw[line width=0.7mm, violet] (n9)--(n10);
        \draw[line width=0.7mm, pink] (n10)--(n11);
        \draw[line width=0.7mm, violet] (n11)--(n3);
    \end{scope}

        \begin{scope}[xshift=2cm]
        \node[circle] at ({180}:0.3cm)  {$v$};
        \node[circle] at ({90}:2.33cm)  {$u_1$};
        \node[circle] at ({70}:2.3cm)  {$u_2$};
        \node[circle] at ({35}:2.75cm)  {$u_3=z(v)$};
        \node[circle] at ({22}:2.9cm)  {$u_4=x(v)$};
        \node[circle] at ({10}:2.4cm)  {$u_5$};
        \node[circle] at ({350}:2.4cm)  {$u_6$};
        \node[circle] at ({338}:2.9cm)  {$u_7=y(v)$};
        \node[circle,fill=pink] at (360:0mm) (center) {};
        \node[circle,fill=blue] at ({90}:2cm) (n1) {};
        \node[circle,fill=red] at ({70}:2cm) (n2) {};
        \node[circle,fill=violet] at ({50}:2cm) (n3) {};
        \node[circle,fill=olive] at ({30}:2cm) (n4) {};
        \node[circle,fill=cyan] at ({10}:2cm) (n5) {};
        \node[circle,fill=brown] at ({350}:2cm) (n6) {};
        \node[circle,fill=pink] at ({330}:2cm) (n7) {};
    
        \draw[dotted, line width=0.7mm] (center)--(n1);
        \draw[line width=0.7mm, blue] (center)--(n2);
        \draw[line width=0.7mm, red] (center)--(n3);
        \draw[line width=0.7mm, violet] (center)--(n4);
        \draw[line width=0.7mm, olive] (center)--(n5);
        \draw[line width=0.7mm, cyan] (center)--(n6);
        \draw[line width=0.7mm, brown] (center)--(n7);

        \node[circle, fill=black] at (3.5,0.5) (n8) {};
        \node[circle, fill=black] at (5.5,0.7) (n9) {};
        \node[circle, fill=black] at (5,2.5) (n10) {};
        \node[circle, fill=black] at (3,2.7) (n11) {};

        \draw[line width=0.7mm, pink] (n4)--(n8);
        \draw[line width=0.7mm, violet] (n8)--(n9);
        \draw[line width=0.7mm, pink] (n9)--(n10);
        \draw[line width=0.7mm, violet] (n10)--(n11);
        \draw[line width=0.7mm, pink] (n11)--(n3);
    \end{scope}

        \draw [-{Stealth[length=5mm]}] (-0.7,0.6) -- (1.2,0.6);

        \node[circle,fill=white] at (-4.85,2.15) {};
        \node[circle,fill=white] at (-6,0.55) {};
        \node at (-4.85,2.15) {$\alpha$};
        \node at (-6,0.55) {$\beta$};

        \node[circle,fill=white] at (4.15,2.15) {};
        \node[circle,fill=white] at (3,0.55) {};
        \node at (4.15,2.15) {$\beta$};
        \node at (3,0.55) {$\alpha$};
    \end{tikzpicture}

    \caption{The $\alpha\beta$-path of the maximal fan ends at $z(v)$. On the right-hand side we depicted the coloring obtained as a result of exchanging this path.}
    \label{fig:alpha beta path collision}
    \end{figure}
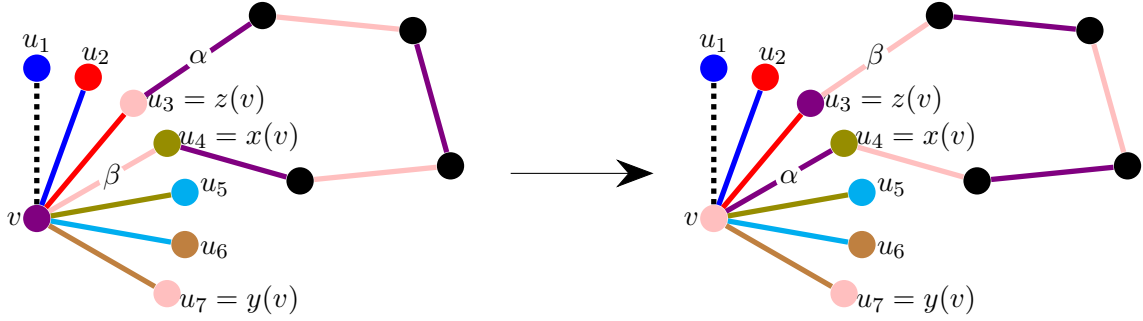

        \item[(3)] Otherwise, after the exchanging of the $\alpha\beta$-path of $v$, the color $\beta$ is still free at $z(v)$ (and the colors of edges $(v,u_j)$ and missing colors of vertices $u_j$ in the beginning of the fan, $\langle u_1,...,u_{i}=z(v)\rangle$, $1\leq j\leq i$, are unchanged). Hence $\langle u_1,...,u_{i}=z(v)\rangle$ is a fan of $v$, and step 2 of Procedure \textsc{Recolor-Fan} is well-defined and can be performed.
    \end{itemize}
    We conclude that after step 1 of Procedure \textsc{Recolor-Fan}, the color $\beta$ is free at both $v$ and $u_i$ (see step 2). Hence after the rotation of the fan $\langle u_1,...,u_i\rangle$, we can properly color the edge $(v,u_i)$ with the color $\beta$. Note that the colors of the edges of the fan and the colors $\alpha$ and $\beta$ are all from the palette $\{1,2,...,\Delta+1\}$. Hence Procedure \textsc{Recolor-Fan} properly recolors the edges of the fan and its $\alpha\beta$-path using colors from the palette $\{1,2,...,\Delta+1\}$. The resulting partial edge-coloring is proper, and the number of colored edges increases by 1.

    We now analyse the complexity of step 1 of Procedure \textsc{Recolor-Fan}. To this end, we first compute the connected components of the graph $G_{\alpha,\beta}$ using Lemma \ref{Connected components algorithm} in $O(\log n)$ time using $O(n)$ processors, and for each edge $e$ in the $\alpha\beta$-path of $v$, we assign a processor, that recolors $e$ with $\beta$ if $e$ is colored $\alpha$, and vice versa. This process requires $O(1)$ time using $O(n)$ processors. Observe that this process can be applied on various $\alpha\beta$-paths (that are characterized by the same pair of colors $(\alpha,\beta)$ or $(\beta,\alpha)$) in parallel with the same complexity. (Note that they are necessarily vertex-disjoint.)
    
    For computing the vertex $u_i$, we assign a processor to each incident edge of $v$. If there is no edge $(v,u_j)$ that is $\varphi'$-colored $\alpha$ (before the exchanging this edge was $\varphi$-colored $\beta$, i.e., $u_j=x(v)$), then we set $u_i=u_k=y(v)$. Otherwise, if such an index $j$ exists, and $u_{j-1}=z(v)$ is the other endpoint (other than $v$) of the $\alpha\beta$-path of $v$, then we set $u_i=u_k=y(v)$, and otherwise $u_i=u_{j-1}=z(v)$. In either case, the vertex $z(v)$ can at this point be determined in $O(1)$ time. Next, for the rotation of the fan, we assign a processor to each $l\in\{1,2,...,i-1\}$ that recolors $(v,u_l)$ with the color of $(v,u_{l+1})$. And finally, we color $(v,u_i)$ with the color $\beta$. Hence, step 2 of Procedure \textsc{Recolor-Fan} requires $O(1)$ time using $O(\deg(v))$ processors.
\end{proof}

\subsection{Parallel Fan-Recoloring}\label{sec: Parallel Fan-Recoloring}

As we have seen in the previous section, in the constructive proof of Vizing theorem, one can properly color edges one after another by $\Delta+1$ colors by repetitive applications of Procedure \textsc{Recolor-Fan}. In this section we parallelize executions of Procedure \textsc{Recolor-Fan} on a large collection of fans, i.e., we would like to find a large collection of fans that Procedure \textsc{Recolor-Fan} can recolor in parallel, so that different recolorings do not interfere with one another.\\

We say that a fan $\langle u_1,...,u_k\rangle$ of $v\in V$ and a fan $\langle u^*_1,...,u^*_l\rangle$ of $v^*\in V$ are \textit{disjoint fans} if $\{v,u_1,...,u_k\}\cap \{v^*,u^*_1,...,u^*_l\}=\emptyset$. Otherwise, we say that these fans are \textit{intersecting}.\\
As a first part of finding a large collection of fans that Procedure \textsc{Recolor-Fan} can be applied on in parallel, we will find a large collection of pairwise disjoint fans. To this end, let $F$ be a fixed subset of the set of the uncolored edges in the input graph $G$. We define an auxiliary graph (which we call \emph{fan-graph}) $G^{(F)}$ over the set $F$. This graph will have the property that two edges $\left(\text{vertices in $G^{(F)}$}\right)$ are connected if they might be a part of two intersecting fans. First, we define the \textit{distance} between a pair of edges $e_1=(v_1,u_1)\in E$ and $e_2=(v_2,u_2)\in E$ in $G$ by 
\begin{equation}\label{Equation: edges distance}
    \text{dist}_G(e_1,e_2)=\min\left\{\text{dist}_G(u_1, u_2),\text{dist}_G(u_1, v_2),\text{dist}_G(v_1, u_2),\text{dist}_G(v_1, v_2)\right\}.
\end{equation}

\begin{definition}\label{def: G^F}$\left(\mathrm{The\,\, fan\text{-}graph\,\,} G^{(F)}\right)$\textbf{.}
    Let $G=(V,E)$ be a graph, $\varphi$ be a partial proper edge-coloring of $G$, and let $F$ be a fixed subset of the set of uncolored edges in $G$.
    We define an auxiliary graph $G^{(F)}=\left(F, E^{(F)}\right)$, where for $e_1\neq e_2\in F$, an edge $(e_1, e_2)$ is in $E^{(F)}$ if $\mathrm{dist}_G(e_1, e_2)\leq 2$.
\end{definition}

As we argue next, if the distance between two uncolored edges is more than 2, they cannot be a part of intersecting fans. See Figure \ref{fig:intersection fans} for an example of two intersecting fans in which the distance between the uncolored edges of these fans is 2.

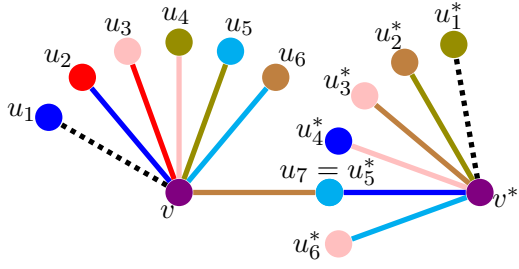
\begin{figure}
    \centering
    \begin{tikzpicture}

        \begin{scope}[xshift=-1cm, rotate=60]
        \node[circle] at ({180}:0.3cm)  {$v$};
        \node[circle] at ({93}:2.35cm)  {$u_1$};
        \node[circle] at ({72}:2.4cm)  {$u_2$};
        \node[circle] at ({50}:2.35cm)  {$u_3$};
        \node[circle] at ({31}:2.35cm)  {$u_4$};
        \node[circle] at ({10}:2.35cm)  {$u_5$};
        \node[circle] at ({350}:2.35cm)  {$u_6$};
        \node[circle] at ({308}:2cm)  {$u_7=u^*_5$};
        \node[circle,fill=violet] at (360:0mm) (center) {};
        \node[circle,fill=blue] at ({90}:2cm) (n1) {};
        \node[circle,fill=red] at ({70}:2cm) (n2) {};
        \node[circle,fill=pink] at ({50}:2cm) (n3) {};
        \node[circle,fill=olive] at ({30}:2cm) (n4) {};
        \node[circle,fill=cyan] at ({10}:2cm) (n5) {};
        \node[circle,fill=brown] at ({350}:2cm) (n6) {};
        \node[circle,fill=cyan] at ({300}:2cm) (n7) {};
    
        \draw[dotted, line width=0.7mm] (center)--(n1);
        \draw[line width=0.7mm, blue] (center)--(n2);
        \draw[line width=0.7mm, red] (center)--(n3);
        \draw[line width=0.7mm, pink] (center)--(n4);
        \draw[line width=0.7mm, olive] (center)--(n5);
        \draw[line width=0.7mm, cyan] (center)--(n6);
        \draw[line width=0.7mm, brown] (center)--(n7);

        \end{scope}

        \begin{scope}[xshift =3cm, yscale=1,xscale=-1, rotate=10]
        \node[circle] at ({180}:0.35cm)  {$v^*$};
        \node[circle] at ({70}:2.39cm)  {$u^*_1$};
        \node[circle] at ({50}:2.42cm)  {$u^*_2$};
        \node[circle] at ({28}:2.4cm)  {$u^*_3$};
        \node[circle] at ({10}:2.4cm)  {$u^*_4$};
        \node[circle] at ({333}:2.4cm)  {$u^*_6$};
        \node[circle,fill=violet] at (360:0mm) (center) {};
        \node[circle,fill=olive] at ({70}:2cm) (n2) {};
        \node[circle,fill=brown] at ({50}:2cm) (n3) {};
        \node[circle,fill=pink] at ({30}:2cm) (n4) {};
        \node[circle,fill=blue] at ({10}:2cm) (n5) {};
        \node[circle,fill=cyan] at ({350}:2cm) (n6) {};
        \node[circle,fill=pink] at ({330}:2cm) (n7) {};
    
        \draw[dotted, line width=0.7mm] (center)--(n2);
        \draw[line width=0.7mm, olive] (center)--(n3);
        \draw[line width=0.7mm, brown] (center)--(n4);
        \draw[line width=0.7mm, pink] (center)--(n5);
        \draw[line width=0.7mm, blue] (center)--(n6);
        \draw[line width=0.7mm, cyan] (center)--(n7);

        \end{scope}
        
    \end{tikzpicture}

    \caption{Two intersecting fans, $\langle u_1,u_2,...,u_7\rangle$ centered at $v$, and $\langle u^*_1,u^*_2,...,u^*_6\rangle$ centered at $v^*$, in which the distance between the uncolored edges of the fans is 2.}
    \label{fig:intersection fans}
\end{figure}

\begin{lemma}\label{G_3 property}$\left(\mathrm{A \,\,property\,\, of\,\,} G^{(F)}\right)$\textbf{.}
    Let $G=(V,E)$ be a graph and let $\varphi$ be a partial proper edge-coloring of $G$. Let $F$ be a fixed subset of the set of uncolored edges in $G$, and let $G^{(F)}=\left(F,E^{(F)}\right)$ be the fan-graph, defined above. For each pair of distinct vertices (edges in $G$) $e,e^*\in F$ such that $(e,e^*)\notin E^{(F)}$, any fan with an uncolored edge $e$ and any fan with an uncolored edge $e^*$ are disjoint.
\end{lemma}
The proof is simple. For completeness we provide it in Appendix \ref{App: Some Proofs from Section 3}.\\

In our algorithm we will focus on a special case, where the set of the uncolored edges $F$ is a \textit{matching} in $G$.

\begin{definition}[Matching]
    A \emph{matching} $M$ in a graph $G=(V,E)$ is a set of edges such that any vertex in the graph $G$ is incident to at most one edge in $M$.
\end{definition}
From now on, we assume that the fixed subset $F$ of uncolored edges is a matching in $G$.
Now we bound the number of edges, the maximum degree and the arboricity of $G^{(F)}$.

\begin{claim}\label{Maximum degree of G^F}$\left(\mathrm{Bounds\,\,on\,\,the\,\,number\,\,of\,\,edges,\,\,maximum\,\, degree\,\,and\,\,arboricity\,\,of\,\,} G^{(F)}\right)$\textbf{.}
    Let $G=(V,E)$ be an $m$-edge graph with maximum degree $\Delta$ and arboricity $a$, and let $\varphi$ be a partial proper edge-coloring of $G$. Let $F$ be a fixed subset of the set of uncolored edges in $G$ and assume that $F$ is a matching in $G$, and let $G^{(F)}=\left(F,E^{(F)}\right)$ be the fan-graph of $G$ with respect to $F$. 
    The number of edges in $G^{(F)}$ is at most $4m\Delta$, its maximum degree is bounded by $2(\Delta-1)^2$ and its arboricity is at most $16\Delta a$.
\end{claim}

\begin{proof}
    For a vertex $z$, denote by $h(z)$ the only edge in $F$ (if exists) that is incident on $z$. Observe that for each edge $e=(u,v)$ and a neighbor $w$ of $u$ or $v$, there are at most two pairs of edges $(e_1,e_2)$ in $F$, $(e_1,e_2)\in E^{(F)}$, such that $e\cap e_1\neq \emptyset$ and $w\in e_2$. Indeed, such edges must be $(h(u),h(w))$ or $(h(v),h(w))$. On the other hand, every edge in $G^{(F)}$ can be represented in this way. Therefore, the number of edges in $G^{(F)}$ is at most $4m\Delta$.
    
    We now bound the maximum degree of $G^{(F)}$. Let $e\in F$.
    Since $F$ is a matching, then for each vertex $v\in V$ that is at distance 2 of an endpoint of $e$, there is at most one edge in $F$ that contains this vertex. In addition, each edge $e'\in F$ such that $dist_G(e,e')\leq 2$ must contain such a vertex. Since there are at most $2(\Delta-1)^2$ such edges $e'$, we get that $\deg_{G^{(F)}}(e)\leq 2(\Delta-1)^2$.\footnote{An edge $e'=(v,u)\in F$ which is incident on a neighbor $v$ of an endpoint of $e$ is also incident on a vertex $u$ whose distance from $e$ is 2.}

    Next, we bound the arboricity of $G^{(F)}$. Let $(v_1,v_2,...,v_n)$ be a degeneracy ordering of $G$, i.e., by Claim \ref{claim: arboricity and degeneracy}, for each $i\in\{1,2,...,n\}$, the vertex $v_i$ has at most $2a$ \emph{right neighbors}, that is, neighbors $v_j$, with $j>i$. We will show that this ordering induces an orientation of $G^{(F)}$ in which each vertex has out-degree at most $8\Delta a$. By Claim \ref{claim: arboricity and orientation} it then follows that the arboricity of $G^{(F)}$ is at most $16\Delta a$.
    Consider an edge $(e,e')\in E^{(F)}$. We orient this edge towards $e'$ if there exist endpoints $v_i$ of $e$ and $v_j$ of $e'$ such that $\text{dist}(v_i, v_j) = 2$ and $i < j$. Note that for each $(e,e')\in E^{(F)}$ there exists at least one such pair of endpoints at distance two. Also, note that an edge might be oriented in both directions. (In this case the edge can be oriented arbitrarily.)
    We now bound the outdegree of this orientation. First, observe that for any vertex $v_i$, the number of vertices $v_j$ with $j>i$ and $\text{dist}(v_i,v_j)=2$, is at most $4\Delta a$. Specifically, for each such $v_j$, there exists a vertex $v_k$ such that either:
    \begin{itemize}
        \item $v_k$ is a right neighbor of $v_i$ and $v_j$ is a neighbor of $v_k$, or
        \item $v_k$ is a neighbor of $v_i$ and $v_j$ is a right neighbor of $v_k$.
    \end{itemize}
    In each case, there are at most $2\Delta a$ such vertices $v_j$, leading to a total of at most $4\Delta a$ such vertices.
    Now consider an edge $e=(v_i,v_j)\in F$. As established, each of its endpoints has at most $4\Delta a$ vertices $w$ that are to the right of it and at distance two from it. Moreover, every outgoing edge of $e$ is oriented towards an edge of the form $h(w)$, for such $w$. Therefore, the outdegree of each vertex in this orientation is at most $8\Delta a$.
\end{proof}

Next, we analyse the complexity of constructing $G^{(F)}$. 
\begin{lemma}\label{G_3 construction}$\left(\mathrm{Construction\,\, of\,\,} G^{(F)}\right)$\textbf{.}
    Let $G=(V,E)$ be an $n$-vertex $m$-edge graph with maximum degree $\Delta$ and let $\varphi$ be a partial proper edge-coloring of $G$. Let $F$ be a fixed subset of the set of uncolored edges in $G$, and assume that $F$ is a matching in $G$. The construction of the graph $G^{(F)}$, with some edges possibly appearing more than once, requires $O(1)$ time using $O(m\cdot\Delta)$ processors. Deletion of possible duplications of the edges requires additional $O(\log n)$ time and $O(m\cdot\Delta)$ processors.
\end{lemma}

\begin{proof}
    Since $F$ is a matching, each vertex $v$ will store at most one edge $h(v)\in F$, such that $v$ is one of its endpoints (if exists). We designate a processor $p_e$ to each edge $e\in E$ and a processor $p_{e,w}$, to each edge-vertex pair $(e,w)$, with an edge $e=(v,u)$ and a neighbor $w$ of an endpoint of $e$, i.e., a neighbor of either $v$ or $u$. This latter processor adds the edge $(h(u),h(w))$ to $G^{(F)}$, if $w$ is a neighbor of $v$ and both $h(u)$ and $h(w)$ are defined, and symmetrically, this processor adds the edge $(h(v),h(w))$ to $G^{(F)}$, if $w$ is a neighbor of $u$ and both $h(v)$ and $h(w)$ are defined.
    Hence the construction of $G^{(F)}$ can be implemented in $O(1)$ time using $O(m\cdot\Delta)$ processors.\\
    
    Note that in the construction above it can happen that multiple processors add the same edge $(e,e')$ in parallel to $G^{(F)}$. However, by Claim \ref{Maximum degree of G^F}, the overall size of the constructed multigraph $G^{(F)}$ is $O(m\cdot \Delta)$. Redundancies can now be removed via sorting within additional $O(\log n)$ time using $O(m\cdot \Delta)$ processors.

\end{proof}

By Lemma \ref{G_3 property}, in order to find a large collection of pairwise disjoint fans, it is enough to find a "large" \emph{independent set} of $G^{(F)}$.

\begin{definition}[Independent set]\label{def: independent set}
    For a graph $G=(V,E)$ with maximum degree $\Delta$ and arboricity $a$, an \emph{independent set} $I\subseteq V$ is a set of vertices in $G$ such that for each $v,u\in I$, $(v,u)\notin E$. Let $\lambda(\cdot,\cdot)$ be a polynomial function of two variables, i.e., it depends polynomially both on the first and the second variable. An independent set is called \emph{$\lambda(\Delta,a)$-large} (or $\lambda$-large, if $\Delta$ and $a$ are clear from the context) if it has size $\Omega\left(\frac{n}{\lambda(\Delta,a)}\right)$.
\end{definition}

For our edge-coloring algorithm, we consider numerous algorithms for computing an independent set. These algorithms are based on different vertex-coloring algorithms, which are presented in Appendix \ref{Vertex-Coloring Algorithm}. 
These results are summarized in the next theorem.

\begin{restatable}[Computing an independent set]{theorem}{independentSetAlg}\label{Large independent set alg}
    Let $G=(V,E)$ be an $n$-vertex $m$-edge graph with maximum degree $\Delta$ and arboricity $a$. 
    \begin{enumerate}
        \item[(1)] (\cite{goldberg1989constructing}) A $\left(\Delta+1\right)$-large independent set, can be computed in $O\left(\log^3 n\right)$ time using $O\left(\frac{m}{\log n}\right)$ processors.
        \item[(2)] (Using the vertex-coloring algorithm from Theorem \ref{procedure vertex-coloring properties}) A $\left(\Delta+1\right)$-large independent set, can be computed in $O\left(\log n+\Delta\cdot\log^2\Delta\right)$ time using $O(m)$ processors.
        \item[(3)] (Using the vertex-coloring algorithm from Theorem \ref{Adaptation of BE08}) An $O\left(a\right)$-large independent set, can be computed in $O\left((a\cdot\log a+\log\Delta)\cdot\log n\right)$ time using $O(m)$ processors.
        \item[(4)] (Using the vertex-coloring algorithm from Theorem \ref{Adaptation of BE11 arboricity}) For any constant $\delta>0$, a $\left(a^{1+o(1)}\right)$-large independent set, can be computed in $O\left(\log^{2+\delta} a\cdot\log n+\log\Delta\cdot\log n\right)$ time using $O\left(m\cdot\frac{\log^{\delta}a\cdot\log n}{\log(a\cdot\log n)}\right)$ processors.
        \item[(5)] (Using the vertex-coloring algorithm from Theorem \ref{Adaptation of Bar16}) An $O\left(\Delta\right)$-large independent set, can be computed in $O\left(\sqrt{\Delta}\cdot\left(\log\Delta+\frac{\log n}{\log\Delta\cdot\log(\Delta\cdot\log n)}\right)+\log n\right)$ time using $O\left(m\cdot\left(\sqrt{\Delta}\cdot\log\Delta+\frac{\Delta^{1/4}\cdot\log n}{\log(\Delta\cdot\log n)}\right)\right)$ processors.
    \end{enumerate}
\end{restatable}

Note that the time complexity of all of the results above is a function of $n$, $\Delta$ and $a$. Also, the number of processors these algorithms require is a function of $n$ and $\Delta$ multiplied by $m$. Hence, to simplify the analysis of our edge-coloring algorithm, we present a notation for these terms. For an $n$-vertex $m$-edge graph $G=(V,E)$ with maximum degree $\Delta$, arboricity $a$, and a parameter $\lambda$ polynomial in $\Delta$ and $a$, we denote the time that is required for computing a $\lambda$-large independent set of $G$ by $IST_{\lambda}\left(n,\Delta,a\right)$, and denote the number of processors required for this process by $m\cdot ISP_{\lambda}\left(n,\Delta,a\right)$. 
Note that for the algorithms from Theorem \ref{Large independent set alg}, we have $IST_{\lambda}\left(O(n),O(\Delta),O(a)\right)=O\left(IST_{\lambda}\left(n,\Delta,a\right)\right)$ and $ISP_{\lambda}\left(O(n),O(\Delta),O(a)\right)=O(ISP_{\lambda}\left(n,\Delta,a\right))$, and that $IST_{\lambda}$ and $ISP_{\lambda}$ are monotonic increasing. Also, note that $IST_{\lambda}\left(n,\Delta,a\right)=\Omega(\log n)$. 

Suppose that we have constructed a $\lambda$-large independent set $I=I^{(F)}$ of the fan-graph $G^{(F)}$, for some polynomial function $\lambda$. We now show how we can restrict this collection of disjoint fans into a collection of fans that can be recolored in parallel. (Recoloring two disjoint fans in parallel may still interfere with one another.) For each edge $e\in I$, we choose arbitrarily an endpoint $v_e$ of $e$, and define $I^*$ to be the set of the selected endpoints. Next, for each edge $e\in I$, and its endpoint $v_e\in I^*$, we construct a fan $f(v_e)=f_e$, centered at $v_e$, with the uncolored edge $e$. By Lemma \ref{G_3 property}, we obtain a collection $C_F=\{f_e\,|\,e\in I\}$ of pairwise disjoint fans. As was mentioned above, we would like to find a collection of fans such that Procedure \textsc{Recolor-Fan} can be applied on all of them in parallel, so that different invocations do not interfere with one another. In the first step of Procedure \textsc{Recolor-Fan}, 
given a fan $\langle u_1,u_2,...,u_k\rangle$ of $v=v_e$ characterized by $(\alpha,\beta)$, we exchange the $\alpha\beta$-path of $v$. In order to avoid intersections between such $\alpha\beta$-paths of different vertices in $I^*$, we focus on vertices in $I^*$ whose fans are characterized by a fixed pair of colors $(\alpha,\beta)$ or $(\beta,\alpha)$, for $\alpha,\beta\in \{1,2,...,\Delta+1\}$. For a pair of colors $\gamma,\delta\in \{1,2,...,\Delta+1\}$, denote 
\begin{equation}\label{IAlpBet}
    I_{\gamma,\delta}=\{v \text{ $|$ } v\in I^*\text{, and the maximal fan $f(v)$ of $v$ is characterized by $(\gamma, \delta)$ or $(\delta,\gamma)$}\}    
\end{equation} 
Let $\alpha,\beta\in\{1,2,...,\Delta+1\}$, such that $|I_{\alpha,\beta}|=\max_{\delta,\gamma\in \{1,2,...,\Delta+1\}}|I_{\delta,\gamma}|$. Denote by $f\left(I_{\alpha,\beta}\right)$ the collection of fans of vertices from $I_{\alpha,\beta}$.
We summarize the properties of this collection of fans in the next lemma:
\begin{lemma}[Disjoint fans and paths]
    The fans in $f\left(I_{\alpha,\beta}\right)$ are pairwise disjoint. In addition, consider two $\alpha\beta$-paths $P$ and $P'$ of two distinct vertices $v\in I_{\alpha,\beta}$ and $v'\in I_{\alpha,\beta}$, respectively. These two paths are either vertex-disjoint or $P=P'$.
\end{lemma}

The only interference that might still appear when recoloring fans in $f\left(I_{\alpha,\beta}\right)$ in parallel are intersections between an $\alpha\beta$-path of one vertex in $I_{\alpha,\beta}$ with a fan of other vertex in $I_{\alpha,\beta}$. To avoid this situation, we define another auxiliary graph $G_{\text{conflict}}=(I_{\alpha,\beta},E_{\text{conflict}})$, in which two distinct vertices $v,v^*\in I_{\alpha,\beta}$ are connected if their fans and $\alpha\beta$-paths might interfere with one another. Namely, $(v,v^*)$ is in $E_{\text{conflict}}$ if an endpoint of the $\alpha\beta$-path of $v$ is a vertex $u^*_i$ in the fan $\langle u^*_1,...,u^*_k\rangle$ of $v^*$ and $\beta(v^*)$ \footnote{$\beta(v^*)$ is the missing color of $u_k^*$. Note that this fan might be characterized by ($\beta,\alpha$), and then $\beta(v^*)=\alpha$. Otherwise $\beta(v^*)=\beta$.} is the (designated) missing color of $u^*_i$ in this fan. We will see later that indeed, only in this case, the recoloring of the fans of $v$ and $v^*$ might interfere with one another. See Figure \ref{fig:disturbing fans} for an example of two such fans.

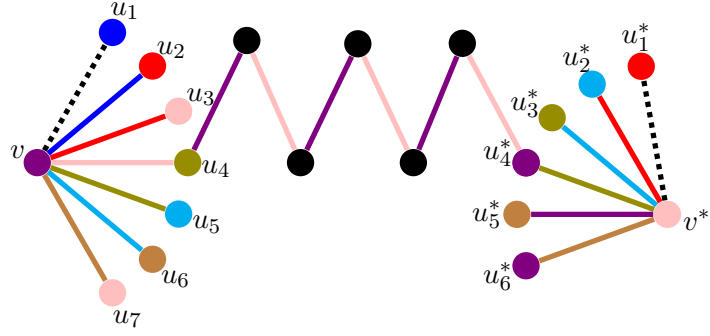
\begin{figure}
    \centering
    \begin{tikzpicture}

        \begin{scope}[xshift=-4cm, yshift=1cm, rotate=330]
        \node[circle] at ({180}:0.3cm)  {$v$};
        \node[circle] at ({90}:2.33cm)  {$u_1$};
        \node[circle] at ({70}:2.33cm)  {$u_2$};
        \node[circle] at ({52}:2.35cm)  {$u_3$};
        \node[circle] at ({28}:2.38cm)  {$u_4$};
        \node[circle] at ({10}:2.4cm)  {$u_5$};
        \node[circle] at ({350}:2.4cm)  {$u_6$};
        \node[circle] at ({330}:2.4cm)  {$u_7$};
        \node[circle,fill=violet] at (360:0mm) (center) {};
        \node[circle,fill=blue] at ({90}:2cm) (n1) {};
        \node[circle,fill=red] at ({70}:2cm) (n2) {};
        \node[circle,fill=pink] at ({50}:2cm) (n3) {};
        \node[circle,fill=olive] at ({30}:2cm) (n4) {};
        \node[circle,fill=cyan] at ({10}:2cm) (n5) {};
        \node[circle,fill=brown] at ({350}:2cm) (n6) {};
        \node[circle,fill=pink] at ({330}:2cm) (n7) {};
    
        \draw[dotted, line width=0.7mm] (center)--(n1);
        \draw[line width=0.7mm, blue] (center)--(n2);
        \draw[line width=0.7mm, red] (center)--(n3);
        \draw[line width=0.7mm, pink] (center)--(n4);
        \draw[line width=0.7mm, olive] (center)--(n5);
        \draw[line width=0.7mm, cyan] (center)--(n6);
        \draw[line width=0.7mm, brown] (center)--(n7);

        \node[circle, fill=black] at (1.6,2.8) (n8) {};
        \node[circle, fill=black] at (30:3.5cm) (n9) {};
        \node[circle, fill=black] at (2.9,3.5) (n10) {};
        \node[circle, fill=black] at (30:5cm) (n11) {};
        \node[circle, fill=black] at (4.1,4.2) (n12) {};
        \node[circle, fill=violet] at (30:6.5cm) (n13) {};

        \draw[line width=0.7mm, violet] (n4)--(n8);
        \draw[line width=0.7mm, pink] (n8)--(n9);
        \draw[line width=0.7mm, violet] (n9)--(n10);
        \draw[line width=0.7mm, pink] (n10)--(n11);
        \draw[line width=0.7mm, violet] (n11)--(n12);
        \draw[line width=0.7mm, pink] (n12)--(n13);
        \end{scope}

        \begin{scope}[xshift =4.377cm, yshift=0.312cm, yscale=1,xscale=-1, rotate=10]
        \node[circle] at ({180}:0.4cm)  {$v^*$};
        \node[circle] at ({70}:2.4cm)  {$u^*_1$};
        \node[circle] at ({50}:2.4cm)  {$u^*_2$};
        \node[circle] at ({28}:2.4cm)  {$u^*_3$};
        \node[circle] at ({10}:2.4cm)  {$u^*_4$};
        \node[circle] at ({350}:2.4cm)  {$u^*_5$};
        \node[circle] at ({330}:2.4cm)  {$u^*_6$};
        \node[circle,fill=pink] at (360:0mm) (center) {};
        \node[circle,fill=red] at ({70}:2cm) (n2) {};
        \node[circle,fill=cyan] at ({50}:2cm) (n3) {};
        \node[circle,fill=olive] at ({30}:2cm) (n4) {};
        \node[circle,fill=violet] at ({10}:2cm) (n5) {};
        \node[circle,fill=brown] at ({350}:2cm) (n6) {};
        \node[circle,fill=violet] at ({330}:2cm) (n7) {};
    
        \draw[dotted, line width=0.7mm] (center)--(n2);
        \draw[line width=0.7mm, red] (center)--(n3);
        \draw[line width=0.7mm, cyan] (center)--(n4);
        \draw[line width=0.7mm, olive] (center)--(n5);
        \draw[line width=0.7mm, violet] (center)--(n6);
        \draw[line width=0.7mm, brown] (center)--(n7);

        \end{scope}
    \end{tikzpicture}

    \caption{Two interfering fans - $f(v)=\langle u_1,u_2,...,u_7\rangle$ centered at $v$ and characterized by $(\alpha,\beta)=(\color{violet}\bullet\color{black},\color{pink}\bullet\color{black})$ and $f(v^*)=\langle u^*_1,u^*_2,...,u^*_6\rangle$ centered at $v^*$ and characterized by $(\beta,\alpha)=(\color{pink}\bullet\color{black},\color{violet}\bullet\color{black})$ - an endpoint of the $(\color{violet}\bullet\color{pink}\bullet\color{black})$-path of $v$ is a vertex $u^*_i=u^*_4$ in the fan of $v^*$ with missing color $\beta=\color{violet}\bullet$.}
    \label{fig:disturbing fans}
\end{figure}

We now formally define the graph $G_{\text{conflict}}$. 

\begin{definition}[The auxiliary graph $G_{\text{conflict}}$]\label{The auxiliary graph G_conflict}
    Let $G=(V,E)$ be a graph and $\alpha,\beta\in\{1,2,...,\Delta+1\}$ be a pair of colors. Let $I_{\alpha,\beta}$ be as defined above (see Equation (\ref{IAlpBet})) and let $G_{\alpha,\beta}$ be as defined above ($G_{\alpha,\beta}=(V,E_{\alpha,\beta})$, where $E_{\alpha,\beta}=\{e\in E\,|\,\varphi(e)=\alpha \text{ or } \varphi(e)=\beta\}$). We define the graph $G_{\mathrm{conflict}}=(I_{\alpha,\beta}, E_{\mathrm{conflict}})$ over the vertices in $I_{\alpha,\beta}$, where for distinct vertices $v,u\in I_{\alpha,\beta}$, the edge $(u,v)$ is in $E_{\mathrm{conflict}}$ if there exists a path connected component in $G_{\alpha,\beta}$ such that its endpoints are either $v$ and $y(u)$ or $v$ and $z(u)$, or symmetrically $u$ and $y(v)$ or $u$ and $z(v)$. (Recall that the special vertices $x(v)$, $y(v)$ and $z(v)$ are defined in Definition \ref{def: special vertices}.)
\end{definition}

We first show that indeed, for two vertices $v,u\in I_{\alpha,\beta}$ that are not connected in $G_{\text{conflict}}$, the fans of $v$ and $u$ can be recolored by Procedure \textsc{Recolor-Fan} in parallel.

\begin{lemma}[A property of $G_{\text{conflict}}$]\label{G_conflict property}
    Let $G=(V,E)$ be a graph and let $\varphi$ be a partial proper edge-coloring of $G$. Let $I_{\alpha,\beta}$ and $G_{\mathrm{conflict}}$ be as defined above in Equation (\ref{IAlpBet}) and Definition \ref{The auxiliary graph G_conflict}, respectively. For each pair of vertices $v,v^*\in I_{\alpha,\beta}$ such that $(v,v^*)\notin E_{\mathrm{conflict}}$, the fans of $v$ and $v^*$ can be recolored in parallel.
\end{lemma}

\begin{proof}
    Let $\langle u_1,...,u_k\rangle$ be the fan of $v$, and let $\langle u^*_1,...,u^*_l\rangle$ be the fan of $v^*$. Since the $\alpha\beta$-paths of $v$ and $v^*$ are either vertex disjoint or constitute the same path, the path exchanging of the $\alpha\beta$-paths of $v$ and $v^*$ can be done in parallel.
    We show now that the rotation of the fan of $v$ (and $v^*$) can be performed properly after these exchangings (for $v^*$ the proof is symmetric).
    Assume, without loss of generality, that the fan of $v$ is characterized by $(\alpha,\beta)$ (and not $(\beta,\alpha)$). First, observe that after the exchanging of the $\alpha\beta$-path of $v$, the color $\beta$ is free at $v$. Recall that the exchanging of a path affects only the free colors of its endpoints.
    Observe that by the definition of a fan, for each $u_j\in\{u_1,...,u_{k-1}\}\setminus \{z(v)\}$, we have $m(u_j)\neq \beta$ (as only $(v,x(v))$ is $\varphi$-colored $\beta$) and for each $u_j\in \{ u_1,...,u_{k-1}\}$, we have $m(u_j)\neq \alpha$. (Note that $\alpha$ is free at $v$. On the other hand, $m(u_j)=\alpha$ implies $\varphi(v,u_{j+1})=\alpha$, contradiction.) Hence, the exchanging of the $\alpha\beta$-paths of $v$ and $v^*$ does not affect the (designated) missing colors of the vertices in $\{u_1,...,u_{k-1}\}\setminus \{z(v)\}$. We now consider three cases:
    \begin{itemize}
        \item[(1)] If $z(v)$ and $x(v)$ do not exist (see Figure \ref{fig:special vertices}), then there is no incident edge on $v$ that is $\varphi$-colored $\beta$. Hence, by the definition of a fan, the color $\beta$ is not the missing color of vertices in $\{u_1,...,u_{k-1}\}$, and the $\alpha\beta$-path of $v$ is empty. Therefore, the exchanging of the $\alpha\beta$-paths of $v$ does not affect the missing colors of the vertices in $\{u_1,...,u_{k-1}\}$. In addition, the $\alpha\beta$-path of $v^*$ cannot end at $u_k=y(v)$, as $(v,v^*)\notin E_{\text{conflict}}$. 
        Hence, after its exchanging, $\beta$ is still missing at $u_k$. So we get that $\langle u_1,...,u_k\rangle$ is still a fan of $v$, and $\beta$ remains free at $u_k=y(v)$. Hence, step 2 of Procedure \textsc{Recolor-Fan} (the rotations) is well-defined and can be performed.
        \item[(2)] Consider the case that $z(v)$ is an endpoint of the $\alpha\beta$-path of $v$. (The other endpoint is $v$ itself. See Figure \ref{fig:alpha beta path collision}.) Then after the exchanging of the $\alpha\beta$-path of $v$, the color $\alpha$ is free at $z(v)$, and the edge $(v,x(v))$ is colored $\alpha$ (and the colors and missing colors of the rest of the fan are unchanged). Hence $\langle u_1,...,u_k\rangle$ is still a fan of $v$ ($z(v)$ and $x(v)$ are consecutive vertices in the fan). Since $(v,v^*)\notin E_{\mathrm{conflict}}$, the $\alpha\beta$-path of $v^*$ cannot end at $u_k=y(v)$ or $z(v)$. Hence exchanging the $\alpha\beta
        $-path of $v^*$ does not affect the missing colors of $u_k$ and $z(v)$. Observe also that if the $\alpha\beta$-paths of $v$ and $v^*$ are equal to one another then $v^*=z(v)$. On the other hand, $v$ and $v^*$ are two distinct vertices in $I_{\alpha,\beta}\subseteq I^*$, i.e., they are endpoints of two edges $e$ and $e^*$, respectively, $e,e^*\in I^{(F)}$. Since $I^{(F)}$ is an independent set in $G^{(F)}$, it means that $dist_G(e,e^*)>2$, but $v$ and $v^*=z(v)$ are neighbors in $G$, contradiction. Hence, the $\alpha\beta$-path of $v^*$ and the $\alpha\beta$-path of $v$ are disjoint. Then the colors of the edges of the fan of $v$ are also unchanged when exchanging the $\alpha\beta$-path of $v^*$ (the $\alpha\beta$-path of $v^*$ cannot contain the edge $(v,x(v))$ or the vertex $z(v)$). Hence $\langle u_1,...,u_k\rangle$ is still a fan of $v$, and step 2 of Procedure \textsc{Recolor-Fan} is well-defined and can be performed.
        \item[(3)] Otherwise, since $(v,v^*)\notin E_{\mathrm{conflict}}$, then $z(v)$ is not an endpoint of the $\alpha\beta$-path of $v^*$. (In addition, it is not an endpoint of the $\alpha\beta$-path of $v$.) After the exchanging of the $\alpha\beta$-paths of $v$ and $v^*$, the color $\beta$ is still the missing color of $z(v)$, and the colors and missing colors of the sub-fan, $\langle u_1,...,z(v)\rangle$, are unchanged). Hence $\langle u_1,...,z(v)\rangle$ is also a fan of $v$, and step 2 of Procedure \textsc{Recolor-Fan} is well-defined and can be performed.
    \end{itemize}
\end{proof}

We now analyse the maximum degree of $G_{\mathrm{conflict}}$.

\begin{claim}[Maximum degree of $G_{\mathrm{conflict}}$]\label{Maximum degree of G_conflict}
    Let $G=(V,E)$ be a graph and let $\varphi$ be a partial proper edge-coloring of $G$. Let $G_{\mathrm{conflict}}$ be the graph defined in Definition \ref{The auxiliary graph G_conflict}. The maximum degree of $G_{\mathrm{conflict}}$ is at most 3.
\end{claim}

\begin{proof}
    Let $v\in I_{\alpha,\beta}$ a vertex in $G_{\mathrm{conflict}}$. By the definition of $G_{\mathrm{conflict}}$, $v$ can be only connected to the other endpoint of the $\alpha\beta$-path of $y(v)$, the other endpoint of the $\alpha\beta$-path of $z(v)$, and to the center of the fan containing the other endpoint of the $\alpha\beta$-path of $v$. Hence $\Delta(G_{\mathrm{conflict}})\leq 3$.
\end{proof}

Next, we analyse the complexity of the construction of $G_{\mathrm{conflict}}$.

\begin{lemma}[Construction of $G_{\text{conflict}}$]\label{G_conflice construction}
    Let $G=(V,E)$ be an $n$-vertex graph with maximum degree $\Delta$, and let $\varphi$ be a partial proper edge-coloring of $G$. Let $I_{\alpha,\beta}$ and $G_{\mathrm{conflict}}$ be as defined above (see Equation (\ref{IAlpBet}) and Definition \ref{The auxiliary graph G_conflict}, respectively). The construction of the graph $G_{\mathrm{conflict}}$ requires $O(\log n)$ time using $O(n)$ processors. 
\end{lemma}

\begin{proof}
    The construction of $G_{\text{conflict}}$ can be done in three steps: 
    \begin{itemize}
        \item First, we compute connected components of the graph $G_{\alpha,\beta}=(V,E_{\alpha,\beta})$, that was defined above ($G_{\alpha,\beta}=(V,E_{\alpha,\beta})$, where $E_{\alpha,\beta}=\{e\in E\,|\,\varphi(e)=\alpha \text{ or } \varphi(e)=\beta\}$), using the algorithm from Lemma \ref{Connected components algorithm}~\cite{shiloach1982logn} in $O(\log n)$ time using $O\left(n\right)$ processors. (Note that there are at most $n$ edges in $G_{\alpha,\beta}$.) Let $l\leq n$ be the resulting number of connected components.
        \item Next, we define an auxiliary array $L$ of size $l$, whose entries are pairs. We then assign a processor to each vertex $v\in V$. Recall that the connected components of $G_{\alpha,\beta}$ are either simple cycles or simple paths. For each $v\in V$, if $v$ is an endpoint of a path connected component indexed by $i$, the processor adds $v$ to $L[i]$. This step requires $O(1)$ time and $O(n)$ processors.
        \item Finally, we assign a processor to each connected component $i\in\{1,2,...,l\}$ of $G_{\alpha,\beta}$. Observe that for a non-empty path connected component $i$ of $G_{\alpha,\beta}$, the entry $L[i]$ will contain the two endpoints $v_i$ and $u_i$ of this path. If $v_i\in I_{\alpha,\beta}$ and $u_i=y(u_i^*)$ or $u_i=z(u_i^*)$ for some $u_i^*\in I_{\alpha,\beta}$ then the processor defines the edge $(v_i,u^*_i)$ in $G_{\text{conflict}}$ (every vertex of the form $x(v)$, $y(v)$, or $z(v)$ will store its fan center $v$). This step requires $O(1)$ time and $O(n)$ processors.
    \end{itemize}
    Hence the construction of $G_{\text{conflict}}$ requires $O(\log n)$ time using $O\left(n\right)$ processors.
\end{proof}

By Lemma \ref{G_conflict property}, in order to find a collection of fans that can be recolored in parallel, it is sufficient to find an independent set of $G_{\mathrm{conflict}}$. Next, we present Procedure \textsc{Parallel-Fans} that computes a collection of fans that can be recolored in parallel.\\
\\$\textsc{Procedure}$ $\textsc{Parallel-Fans}\text{ }(G=(V,E),\varphi)$
    \begin{description}
        \item[\textbf{Step 1.}] Let $F$ be a fixed subset of the set of uncolored edges in $G$. Construct the graph $G^{(F)}$ (see Definition \ref{def: G^F}) using the algorithm from Lemma \ref{G_3 construction}.
        \item[\textbf{Step 2.}] Compute a $\lambda$-large independent set $I^{(F)}$ of $G^{(F)}$ using one of the algorithms from Theorem \ref{Large independent set alg}.
        \item[\textbf{Step 3.}] For each $e\in I^{(F)}$ in parallel, choose an arbitrarily endpoint $v$ of $e$ and compute a maximal fan of $v$ with the uncolored edge $e$, using Lemma \ref{Fan construction}. Denote by $I^*$ the set of the selected endpoints.
        \item[\textbf{Step 4.}] Let $\alpha,\beta\in\{1,2,...,\Delta+1\}$ be such that $|I_{\alpha,\beta}|=\max_{\delta,\gamma\in \{1,2,...,\Delta+1\}}|I_{\delta,\gamma}|$. Construct the graph $G_{\text{conflict}}=(I_{\alpha,\beta},E_{\mathrm{conflict}})$ using the algorithm from Lemma \ref{G_conflice construction}.
        \item[\textbf{Step 5.}] Compute an $\left(\Delta\left(G_{(\text{conflict})}\right)+1\right)$-large independent set $I_{\text{fan}}$ of $G_{\text{conflict}}$ using the algorithm from Theorem \ref{Large independent set alg}(2).
        \item[\textbf{return}] the fans of the vertices in $I_{\text{fan}}$.
    \end{description}

We next analyse the complexity of Procedure \textsc{Parallel-Fans}.
It is parametrized by the parameter $\lambda$, which determines how large is the independent set that the procedure computes on step 2. We denote Procedure \textsc{Parallel-Fans} with a parameter $\lambda$ by Procedure \textsc{Parallel-Fans$_{\lambda}$}. We omit the subscript $\lambda$ when the statement at hand is independent of the parameter.\\

We start with analyzing the computation of $I_{\alpha,\beta}$ (on step 4 of this procedure). See Equation (\ref{IAlpBet}) for the definition of $I_{\alpha,\beta}$.

\begin{claim}[Computation of $I_{\alpha,\beta}$]\label{computation of i_ab}
    Let $G=(V,E)$ be an $n$-vertex graph with maximum degree $\Delta$ and let $\varphi$ be a partial proper edge-coloring of $G$. Let $F$ be a fixed subset of the set of uncolored edges in $G$. Let $I^*$ be the collection of fan centers as defined in Procedure \textsc{Parallel-Fans}. The computation of $I_{\alpha,\beta}$ requires $O(\log n)$ time using $O\left(n\right)$ processors.
\end{claim}

\begin{proof}
    First, we describe a simple routine that requires $O\left(n+\Delta^2\right)$ processors. We assign a processor to each vertex $v\in I^*$ and assign a processor to each pair of colors $\alpha,\beta\in\{1,2,...,\Delta+1\}$. Using the processors of the vertices of $I^*$, we sort their respective fans according to their colors. Denote by $\mathcal{A}$ the resulting sorted array. Then, for each pair of colors $\alpha,\beta\in\{1,2,...,\Delta+1\}$, its processor will compute the size of $I_{\alpha,\beta}$ by computing the range of this pair of colors in the array $\mathcal{A}$. Finally, we return the largest such a set. Overall, we have used $O\left(n+\Delta^2\right)$ processors and $O(\log n)$ time.
    
    In fact, this task can also be implemented within the same time using $O(n)$ processors. For each pair of colors $\alpha,\beta$, let $v_{\alpha,\beta}\in I^*$ be the vertex whose fan $f=f(v_{\alpha,\beta})$ is characterized by $\alpha,\beta$, and such that $f$ is the first among all fans characterized by $\alpha,\beta$ in the array $\mathcal{A}$. Now, the processors $v_{\alpha,\beta}$, for all pair of colors $\alpha,\beta$, for which there is a fan characterized by this pair in $\mathcal{A}$, find the range of their respective pairs of colors in $\mathcal{A}$, and also find the maximal range. As a result, the overall number of processors required for this computation is $O(n)$.
\end{proof}

Now we proceed to analysing the complexity of Procedure \textsc{Parallel-Fans}.

\begin{lemma}[Complexity of Procedure \textsc{Parallel-Fans}]\label{Procedure Parallel-Fans complexity}
    Let $G=(V,E)$ be an $n$-vertex $m$-edge graph with maximum degree $\Delta$ and arboricity $a$, and let $\varphi$ be a partial proper edge-coloring of $G$. Let $F$ be a fixed subset of the set of uncolored edges in $G$ and assume that $F$ is a matching in $G$. Procedure \textsc{Parallel-Fans$_{\lambda}$} requires $O\left(IST_{\lambda}\left(n,\Delta^2,a\cdot\Delta\right)\right)$ time using $O\left(m\cdot\Delta\cdot ISP_{\lambda}\left(n,\Delta^2,a\cdot\Delta\right)+m\cdot\Delta\right)$ processors.
\end{lemma}

\begin{proof}
    \begin{itemize}
        \item By Lemma \ref{G_3 construction}, the construction of $G^{(F)}$ in step 1 requires $O(\log n)$ time using $O(m\cdot\Delta)$ processors.    
        \item Observe that since $F$ is a matching, there are at most $n$ vertices in $G^{(F)}$. By Claim \ref{Maximum degree of G^F}, there are $m^{(F)}=O(m\cdot\Delta)$ edges in $G^{(F)}$, its maximum degree is $\Delta^{(F)}=O(\Delta^2)$, and its arboricity is $a^{(F)}=O(a\cdot\Delta)$. Hence, the computation of a $\lambda$-large independent set $I^{(F)}$ of $G^{(F)}$ in step 2 requires $O\left(IST_{\lambda}\left(n,\Delta^2,a\cdot\Delta\right)\right)$ time using $O\left(m^{(F)}\cdot ISP_{\lambda}\left(n,\Delta^2,a\cdot\Delta\right)\right)=O\left(m\cdot\Delta\cdot ISP_{\lambda}\left(n,\Delta^2,a\cdot\Delta\right)\right)$ processors.
        \item By Lemma \ref{Fan construction}, for all $(v,u)\in I^{(F)}$ and $v\in I^*$ in parallel, the construction of a maximal fan of $v$ with the uncolored edge $(v,u)$ requires $O(\log\Delta)$ time using $\sum_{v\in I^*}O\left(\deg(v)\right)=O\left(m\right)$ processors altogether.
        \item By Claim \ref{computation of i_ab}, the computation of $I_{\alpha,\beta}$ requires $O(\log n)$ time using $O\left(n\right)$ processors.
        \item By Lemma \ref{G_conflice construction}, the construction of $G_{\mathrm{conflict}}$ requires $O(\log n)$ time using $O(n)$ processors.
        \item Recall that by Claim \ref{Maximum degree of G_conflict}, $\Delta(G_{\mathrm{conflict}})\leq 3$.
        Hence by Theorem \ref{Large independent set alg}(2), the computation of a $(\Delta(G_{\mathrm{conflict}})+1)$-large independent set of $G_{\mathrm{conflict}}$ in step 5 requires $O(\log n)$ time using $O\left(n\right)$ processors.
    \end{itemize}
    To summarize, Procedure \textsc{Parallel-Fans} requires $O\left(IST_{\lambda}\left(n,\Delta^2,a\cdot\Delta\right)+\log n\right)=O\left(IST_{\lambda}\left(n,\Delta^2,a\cdot\Delta\right)\right)$ time using $O\left(m\cdot\Delta\cdot ISP_{\lambda}\left(n,\Delta^2,a\cdot\Delta\right)+m\cdot\Delta\right)$ processors. (Note that $ISP_{\lambda}\left(n,\Delta^2,a\cdot\Delta\right)$ might be less than 1. See for example Theorem~\ref{Large independent set alg}(1))
\end{proof}

Finally, we analyse the size of the collection of fans that is returned by Procedure \textsc{Parallel-Fans}.

\begin{lemma}[Properties of Procedure \textsc{Parallel-Fans}]
    Let $G=(V,E)$ be an $n$-vertex graph with maximum degree $\Delta$ and arboricity $a$, and let $\varphi$ be a partial proper edge-coloring of $G$. Let $F$ be a fixed subset of the set of uncolored edges in $G$ and assume that $F$ is a matching in $G$. Procedure \textsc{Parallel-Fans$_{\lambda}$} outputs a collection of $\Omega\left(\frac{|F|}{\lambda\left(\Delta^2,a\cdot\Delta\right)\cdot\Delta^2}\right)$ fans that can be recolored in parallel.
\end{lemma}

\begin{proof}
    Let $\Delta^{(F)}=\Delta\left(G^{(F)}\right)$. The size of the $\lambda$-large independent set $I^{(F)}$ of $G^{(F)}$ that is computed in step 2 of the algorithm is at least $\Omega\left(\frac{|F|}{\lambda\left(\Delta^{(F)},a^{(F)}\right)}\right)$, i.e., $$\left|I^*\right|=\left|I^{(F)}\right|=\Omega\left(\frac{|F|}{\lambda\left(\Delta^{(F)},a^{(F)}\right)}\right).$$
    (Recall that $I^*$ contains one endpoint from each edge $e\in I^{(F)}$.) Observe that there are $\binom{\Delta+1}{2}+(\Delta+1)=\frac{\Delta^2+3\Delta+2}{2}$ pairs of (not necessarily distinct) missing colors $\gamma,\delta\in \{1,2,...,\Delta+1\}$. Note also that $I^*=\cup_{\alpha,\beta\in\{1,2,...,\Delta\}}I_{\alpha,\beta}$, and that for distinct pairs $(\alpha,\beta)\neq(\alpha',\beta')$ of colors, the respective subsets $I_{\alpha,\beta}$ and $I_{\alpha',\beta'}$ are disjoint. (See Equation (\ref{IAlpBet}).) Since $I_{\alpha,\beta}$ is the largest among these sets, we have $$|I_{\alpha,\beta}|\geq\frac{|I^*|}{\frac{\Delta^2+3\Delta+2}{2}}\geq\frac{|F|}{c\cdot\lambda\left(\Delta^{(F)},a^{(F)}\right)\cdot \Delta^2},$$
    for a constant $c>0$. Since the maximum degree of $G_{\text{conflict}}$ is at most 3, by Theorem \ref{Large independent set alg}(2), the size of the independent set $I_{\text{fan}}$ of $G_{\text{conflict}}$ that is computed in step 5 of the algorithm is at least $\left|I_{\text{fan}}\right|\geq\frac{\left|I_{\alpha,\beta}\right|}{4}\geq\frac{|F|}{4\cdot c\cdot\lambda\left(\Delta^{(F)},a^{(F)}\right)\cdot\Delta^2}$.
    By Claim \ref{Maximum degree of G^F}, we have $\Delta^{(F)}\leq 2\Delta^2-4\Delta+2$, and $a^{(F)}=O(a\cdot\Delta)$. Hence, Procedure \textsc{Parallel-Fans} outputs a collection of at least $\Omega\left(\frac{|F|}{\lambda\left(\Delta^{2},a\cdot\Delta\right)\cdot\Delta^2}\right)$ fans that can be recolored in parallel. (Since $\lambda$ is a polynomial function, we have $\lambda(O(\Delta^2),O(a\cdot\Delta))=O\left(\lambda\left(\Delta^2,a\cdot\Delta\right)\right)$.)
\end{proof}

We summarize the properties of Procedure \textsc{Parallel-Fans} in the following theorem.

\begin{theorem}[Procedure \textsc{Parallel-Fans}]\label{procedure parallel fans}
    Let $G=(V,E)$ be an $n$-vertex $m$-edge graph with maximum degree $\Delta$ and arboricity $a$, and let $\varphi$ be a partial proper edge-coloring of $G$. Let $F$ be a fixed subset of the set of uncolored edges in $G$ and assume that $F$ is a matching of $G$. Procedure \textsc{Parallel-Fans$_{\lambda}$} outputs a collection of at least $\Omega\left(\frac{|F|}{\lambda\left(\Delta^{2},a\cdot\Delta\right)\cdot\Delta^2}\right)$ fans that can be recolored in parallel. It does so in $O\left(IST_{\lambda}\left(n,\Delta^2,a\cdot\Delta\right)\right)$ time using $O\left(m\cdot\Delta\cdot ISP_{\lambda}\left(n,\Delta^2,a\cdot\Delta\right)+m\cdot\Delta\right)$ processors.
\end{theorem}

\subsection{The Edge-Coloring Algorithm}\label{sec: The Edge-Coloring Algorithm}

In this section we present our parallel edge-coloring algorithm.
We start by describing Procedure \textsc{Reduce-Color}, that given a graph $G=(V,E)$ with maximum degree $\Delta$, and a proper edge-coloring $\varphi$ of $G$ that uses $k>\Delta+1$ colors, returns a proper edge-coloring of $G$ that uses $k-1$ colors. The algorithm uncolors a color class in the graph, and iteratively recolors a large subset of the uncolored edges in parallel, using Procedures \textsc{Parallel-Fans} and \textsc{Recolor-Fan}.\\
\\$\textsc{Procedure}$ $\textsc{Reduce-Color}\text{ }(G=(V,E),\varphi)$\Comment{$|\varphi|=k$}
    \begin{description}
        \item[\textbf{Step 1.}] Uncolor the edges in $G$ that are colored $k$. Denote the set of uncolored edges by $F$, and let $\varphi$ be the resulting partial proper edge-coloring.
        \item \textbf{Repeat}
        \item[\textbf{Step 2.}] Compute a collection of fans that can be recolored in parallel using Procedure \textsc{Parallel-Fans} with input $(G,\varphi)$, and for each such fan in parallel, apply Procedure \textsc{Recolor-Fan} to recolor the uncolored edge of this fan. Let $H$ denote the set of edges of $F$ colored on this step.
        \item[\textbf{Step 3.}] $F= F\setminus H$
        \item\textbf{Until} $F=\emptyset$
    \end{description}

We analyse Procedure \textsc{Reduce-Color} and then proceed to the main edge-coloring algorithm. 

The complexity of Procedure \textsc{Reduce-Color} depends on the parameter $\lambda$ of Procedure \textsc{Parallel-Fans}, invoked on step 2 of the former procedure. We denote Procedure \textsc{Reduce-Color} that is invoked with a parameter $\lambda$ by \textsc{Reduce-Color$_{\lambda}$}. We may omit the parameter $\lambda$ whenever the statement at hand does not depend on it.

\begin{theorem}[Properties of Procedure \textsc{Reduce-Color}]\label{Procedure Reduce-Color properties}
    Let $G=(V,E)$ be an $n$-vertex $m$-edge graph with maximum degree $\Delta$ and arboricity $a$. Let $k>\Delta+1$ be an integer and let $\varphi$ be a proper $k$-edge-coloring of $G$. Procedure \textsc{Reduce-Color$_{\lambda}$} computes a proper $(k-1)$-edge-coloring of $G$ in $O\left(\lambda\left(\Delta^2,a\cdot\Delta\right)\cdot\Delta^2\cdot \log n\cdot IST_{\lambda}\left(n,\Delta^2,a\cdot\Delta\right)\right)$ time using $O\left(m\cdot\Delta\cdot ISP_{\lambda}\left(n,\Delta^2,a\cdot\Delta\right)+m\cdot\Delta\right)$ processors.
\end{theorem}

\begin{proof}
    Let $F$ the set of uncolored edges as described in step 1. Observe that $F$ is a matching in $G$ of size at most $\frac{n}{2}$. By Theorem \ref{procedure parallel fans}, in each iteration of step 2, we compute a collection of $\Omega\left( \frac{|F|}{\lambda\left(\Delta^{2},a\cdot\Delta\right)\cdot\Delta^2}\right)$ fans that can be recolored in parallel. This computation requires $O\left(IST_{\lambda}\left(n,\Delta^2,a\cdot\Delta\right)\right)$ time using $O\left(m\cdot\Delta\cdot ISP_{\lambda}\left(n,\Delta^2,a\cdot\Delta\right)+m\cdot\Delta\right)$ processors. Hence, by Theorem \ref{Procedure Recolor-Fan}, in each iteration we recolor $\Omega\left(\frac{|F|}{\lambda\left(\Delta^{2},a\cdot\Delta\right)\cdot\Delta^2}\right)$ edges with colors from $\{1,2,...,\Delta+1\}\subseteq\{1,2,...,k-1\}$ in $O(\log n)$ time, using $O(m)$ processors. Let $c>0$ be a sufficiently small constant such that in each iteration of the main loop of Procedure \textsc{Reduce-Color} we color at least $c\cdot\frac{|F|}{\lambda\left(\Delta^{2},a\cdot\Delta\right)\cdot\Delta^2}$ edges.\\
    For each integer $i$, after $i$ iterations of the main loop of Procedure \textsc{Reduce-Color} we are left with at most $|F|\cdot\left(1-\frac{c}{\lambda\left(\Delta^{2},a\cdot\Delta\right)\cdot\Delta^2}\right)^i$ uncolored edges. Hence after $\frac{\lambda\left(\Delta^{2},a\cdot\Delta\right)\cdot\Delta^2\cdot\ln n}{c}$ iterations, we are left with at most $|F|\cdot\left(1-\frac{c}{\lambda\left(\Delta^{2},a\cdot\Delta\right)\cdot\Delta^2}\right)^{\lambda\left(\Delta^{2},a\cdot\Delta\right)\cdot\Delta^2\cdot\ln n}\leq |F|\cdot\frac{1}{n}<1$ uncolored edges. Hence after $O\left(\lambda\left(\Delta^2,a\cdot\Delta\right)\cdot\Delta^2\cdot\log n\right)$ iterations of the main loop the algorithm terminates. We conclude that Procedure \textsc{Reduce-Color} computes a proper $(k-1)$-edge-coloring of $G$ in $O(\lambda\left(\Delta^2,a\cdot\Delta\right)\cdot\Delta^2\cdot \log n\cdot$ $IST_{\lambda}\left(n,\Delta^2,a\cdot\Delta\right))$ time using $O\left(m\cdot\Delta\cdot ISP_{\lambda}\left(n,\Delta^2,a\cdot\Delta\right)+m\cdot\Delta\right)$ processors.
\end{proof}

We are now ready to present the main $(\Delta+1)$-edge-coloring algorithm. We start by explaining the idea of the algorithm for the simple case where the input graph $G=(V,E)$ is an Eulerian graph, and has a maximum degree $\Delta=2^h$, for a positive integer $h$. The algorithm employs a divide-and-conquer approach. We partition the graph $G$ into $\frac{\Delta}{2}$ edge-disjoint subgraphs $G^{(0)}_1,G^{(0)}_2,...,G^{(0)}_{\frac{\Delta}{2}}$ with maximum degree 2 each, rapidly solve the problem for each of them, and merge these solutions into a solution for the input graph $G$. See Figure \ref{fig: procedure edge coloring} for an illustration. We now describe each of these parts:
\begin{description}
    \item[\textbf{Partition:}] We compute a directed Eulerian cycle in $G$, and form a bipartite graph $G_B=\left(V^{(in)}\cup V^{(out)},E_B\right)$, where $V^{(in)}=\left\{v^{(in)}|v\in V\right\}$, $V^{(out)}=\left\{v^{(out)}|v\in V\right\}$, and there is an edge $\left(v^{(out)},u^{(in)}\right)\in E_B$ if the directed edge $\langle v,u\rangle$ is in the Eulerian cycle. Observe that the maximum degree of $G_B$ is $\frac{\Delta}{2}$, and we can edge-color $G_B$ using $\frac{\Delta}{2}$ colors, using the algorithm of~\cite{lev1981fast} (see Lemma \ref{bipartite coloring} below). Define the edge-disjoint subgraphs $G^{(0)}_1,G^{(0)}_2,...,G^{(0)}_{\frac{\Delta}{2}}$ according to the color classes of the edge-coloring. Namely, for each $i\in\left\{1,2,...,\frac{\Delta}{2}\right\}$, we define $G^{(0)}_i=\left(V,E^{(0)}_i\right)$, where $E^{(0)}_i=\left\{(v,u)\in E\,|\, \left(v^{(out)},u^{(in)}\right)\, \text{is colored $i$ in $G_B$}\right\}$.
    \item[\textbf{Color:}] Observe that the maximum degree of each such subgraph is 2. Hence we can 3-edge-color each of these subgraphs using a simple edge-coloring algorithm. (See the algorithm from Theorem \ref{3 edge-coloring}.)
    \item[\textbf{Merge:}] This part is composed of $h=\log\left(\frac{\Delta}{2}\right)$ iterations. In each iteration $i\geq 1$, we are given the subgraphs $G^{(i-1)}_1,G^{(i-1)}_2,...,G^{(i-1)}_{\frac{\Delta}{2^{i}}}$, that are colored using $\Delta\left(G^{(i-1)}_j\right)+1$ colors, $j\in\left\{1,2,...,\frac{\Delta}{2^i}\right\}$, respectively. We merge pairs of subgraphs $G^{(i)}_j=G^{(i-1)}_{2j-1}\cup G^{(i-1)}_{2j}$, using disjoint palettes, and as we will see, we will receive a $\left(\Delta\left(G^{(i)}_j\right)+2\right)$-edge-coloring of $G^{(i)}_j$. Next, we reduce a color from each of these colorings in parallel using Procedure \textsc{Reduce-Color} to obtain a $\left(\Delta\left(G^{(i)}_j\right)+1\right)$-edge-coloring of each $G^{(i)}_j$. After these $h$ iterations, we receive a $\left(\Delta\left(G^{(h)}_1\right)+1\right)=(\Delta(G)+1)$-edge-coloring of $G^{(h)}_1=G$.
\end{description}

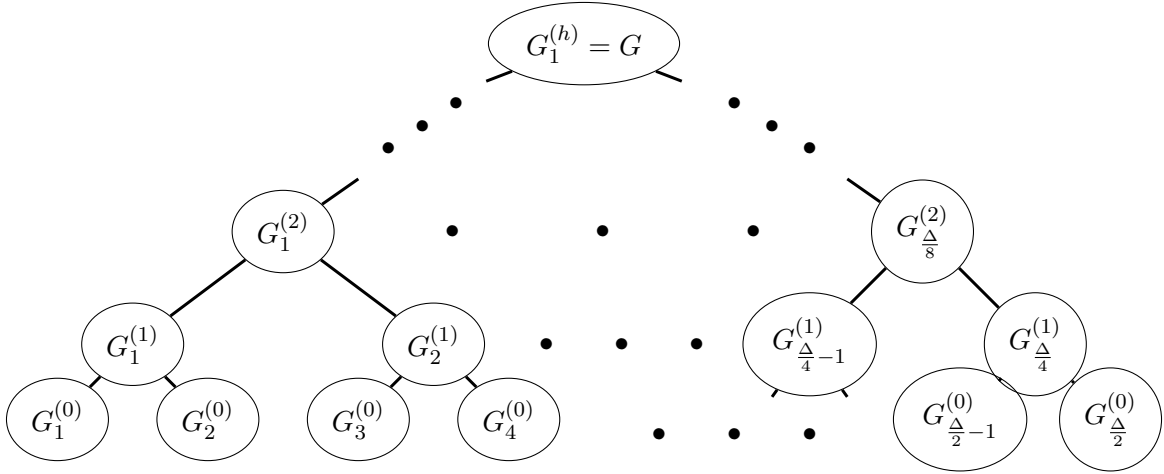
\begin{figure}
    \centering
    \begin{tikzpicture}
        \node[draw, ellipse] (11) at (0,0){$G^{(0)}_1$};
        \node[draw, ellipse] (12) at (2,0){$G^{(0)}_2$};
        \node[draw, ellipse] (13) at (4,0){$G^{(0)}_3$};
        \node[draw, ellipse] (14) at (6,0){$G^{(0)}_4$};
        \node[draw, ellipse] (1n1) at (12,0){$G^{(0)}_{\frac{\Delta}{2}-1}$};
        \node[draw, ellipse] (1n) at (14,0){$G^{(0)}_{\frac{\Delta}{2}}$};
        
        \node[draw, ellipse] (21) at (1,1){$G^{(1)}_1$};
        \node[draw, ellipse] (22) at (5,1){$G^{(1)}_2$};
        \node[draw, ellipse] (2n1) at (10,1){$G^{(1)}_{\frac{\Delta}{4}-1}$};
        \node[draw, ellipse] (2n) at (13,1){$G^{(1)}_{\frac{\Delta}{4}}$};
        
        \node[draw, ellipse] (31) at (3,2.5){$G^{(2)}_1$};
        \node[draw, ellipse] (3n) at (11.5,2.5){$G^{(2)}_{\frac{\Delta}{8}}$};
        
        \node[draw, ellipse] (nn) at (7,5){$G^{(h)}_1=G$};

        \node at (4.4,3.6){$\bullet$};
        \node at (4.85,3.9){$\bullet$};
        \node at (5.3,4.2){$\bullet$};
        
        \node at (10,3.6){$\bullet$};
        \node at (9.5,3.9){$\bullet$};
        \node at (9,4.2){$\bullet$};
        
        \node at (5.25,2.5){$\bullet$};
        \node at (7.25,2.5){$\bullet$};
        \node at (9.25,2.5){$\bullet$};
        
        \node at (6.5,1){$\bullet$};
        \node at (7.5,1){$\bullet$};
        \node at (8.5,1){$\bullet$};
        
        \node at (8,-0.2){$\bullet$};
        \node at (9,-0.2){$\bullet$};
        \node at (10,-0.2){$\bullet$};
        
        \draw[line width=0.4mm] (11)--(21);
        \draw[line width=0.4mm] (12)--(21);
        \draw[line width=0.4mm] (13)--(22);
        \draw[line width=0.4mm] (14)--(22);
        \draw[line width=0.4mm] (1n1)--(2n);
        \draw[line width=0.4mm] (1n)--(2n);
        \draw[line width=0.4mm] (9.5,0.3)--(2n1);
        \draw[line width=0.4mm] (10.5,0.3)--(2n1);
        \draw[line width=0.4mm] (21)--(31);
        \draw[line width=0.4mm] (22)--(31);
        \draw[line width=0.4mm] (2n1)--(3n);
        \draw[line width=0.4mm] (2n)--(3n);
        \draw[line width=0.4mm] (31)--(4,3.2);
        \draw[line width=0.4mm] (3n)--(10.5,3.2);
        \draw[line width=0.4mm] (5.7,4.5)--(nn);
        \draw[line width=0.4mm] (8.3,4.5)--(nn);

    \end{tikzpicture}

    \caption{The merging process of Procedure \textsc{Edge-Coloring}.}
    \label{fig: procedure edge coloring}
\end{figure}

Before presenting Procedure \textsc{Edge-Coloring}, we recall a result due to Atallah and Vishkin~\cite{atallah1984finding} for computing Eulerian cycle in an Eulerian graph, state a result of Lev, Pippenger and Valiant~\cite{lev1981fast} for $\Delta$-edge-coloring bipartite graphs, and present an algorithm for $(\Delta+1)$-edge-coloring graphs with maximum degree at most 2 (the full description of this algorithm is provided in Appendix \ref{App: coloring}).

\begin{restatable}[Computing an Eulerian cycle~\cite{atallah1984finding}]{lemma}{eulerianCyc}\label{eulerian cycle alg}
    Let $G=(V,E)$ be an $n$-vertex $m$-edge Eulerian graph. There is a deterministic $\mathrm{CRCW\,\, PRAM}$ algorithm that finds an Eulerian cycle in $G$ in $O\left(\log n\right)$ time using $O\left(n+m\right)$ processors.
\end{restatable}

\begin{lemma}[A $\Delta$-edge-coloring algorithm for bipartite graphs~\cite{lev1981fast}]\label{bipartite coloring}
    Let $G=(V,E)$ be an $n$-vertex $m$-edge bipartite graph with maximum degree $\Delta$. There is a deterministic $\mathrm{EREW\,\, PRAM}$ algorithm that computes a $\Delta$-edge-coloring of $G$ in $O\left(\log^2 n\cdot\log\Delta\right)$ time using $O\left(n+m\right)$ processors.
\end{lemma}

\begin{restatable}[$(\Delta+1)$-edge-coloring graphs with maximum degree at most 2]{theorem}{threeEdgeColoring}
\label{3 edge-coloring}
    Let $G=(V,E)$ be an $n$-vertex $m$-edge graph with maximum degree $\Delta\leq 2$. There is an algorithm that computes a $(\Delta+1)$-edge-coloring of $G$ in $O\left(\log n\right)$ time using $O\left(m\right)$ processors.
\end{restatable}

We are now ready to present our parallel $(\Delta+1)$-edge-coloring algorithm. A detailed description of the algorithm is presented in Procedure \textsc{Edge-Coloring}.\\
\\$\textsc{Procedure}$ $\textsc{Edge-Coloring}$ $\left(G=(V,E)\right)$
\begin{description}

    \item{\textbf{Step 1.}} If the graph $G$ is not an Eulerian graph, we add a dummy node $x$ to $V$ and connect it to all the odd-degree vertices in $G$. Denote this graph by $G^*=(V^*,E^*)$, and observe that $G^*$ is an Eulerian graph. Next, we construct a bipartite auxiliary graph $G_B=\left(V^{(in)}\cup V^{(out)}, E_B\right)$ for $V^{(in)}=\left\{v^{(in)}\text{ $|$ } v\in V\right\}$ and $V^{(out)}=\left\{v^{(out)}\text{ $|$ } v\in V\right\}$ as follows: We compute a directed Eulerian cycle in $G^*$ using Lemma \ref{eulerian cycle alg} (due to~\cite{atallah1984finding}), and for each directed edge $\langle v,u\rangle$ in this cycle (for $u,v\neq x$) we define an edge $\left(v^{(out)},u^{(in)}\right)$ in $G_B$. Observe that the maximum degree of $G_B$ is $\left\lceil\frac{\Delta}{2}\right\rceil$.
    \item{\textbf{Step 2.}} Edge-color $G_B$ with $\left\lceil\frac{\Delta}{2}\right\rceil$ colors using the algorithm from Lemma \ref{bipartite coloring}, due to~\cite{lev1981fast}.
    \item{\textbf{Step 3.}} Let $h=\left\lceil\log\left\lceil\frac{\Delta}{2}\right\rceil\right\rceil=\left\lceil\log\frac{\Delta}{2}\right\rceil$, and $p=2^{h}$. (The equality $\left\lceil\log\left\lceil\frac{\Delta}{2}\right\rceil\right\rceil=\left\lceil\log\frac{\Delta}{2}\right\rceil$ holds for any integer $\Delta\geq 2$.)
    \item{\textbf{Step 4.}} For $i\in \left\{1,2,...,\left\lceil\frac{\Delta}{2}\right\rceil\right\}$, let $M_i$ be the set of edges colored $i$. Construct a subgraph $G^{(0)}_i=\left(V, E^{(0)}_i\right)$, where $(u,v)\in E^{(0)}_i$ for each edge $\left(u^{(out)}, v^{(in)}\right)\in M_i$. Observe that since $M_i$ is a matching in $G_B$, then $\Delta\left(G^{(0)}_i\right)\leq 2$ (for each vertex $v$, there is at most one edge incident on $v^{(in)}$ and at most one edge incident on $v^{(out)}$).\\
    For $i\in \left\{\left\lceil\frac{\Delta}{2}\right\rceil+1,...,p\right\}$, let $G^{(0)}_i$ be a dummy empty graph. (We define these graphs to simplify notation.)
    \item{\textbf{Step 5.}} \textbf{For} each $1\leq i\leq p$ in parallel \textbf{do}\\
    \makebox[1.15cm]{} Compute a $\left(\Delta\left(G^{(0)}_i\right)+1\right)$-edge-coloring $\varphi^{(0)}_i$ of $G^{(0)}_i$ (see Theorem \ref{3 edge-coloring}).\\ 
    \makebox[0.7cm]{} \textbf{EndFor}
    \item{\textbf{Step 6.}} \textbf{For} each $0\leq k\leq h-1$ \textbf{do} \\
        \makebox[1.3cm]{} \textbf{For} each $1\leq i\leq \frac{p}{2^{k+1}}$ in parallel \textbf{do}\Comment{$p$ is defined on step 3.}\\
            \makebox[1.8cm]{} (1) Define $G_i^{(k+1)}=G_{2i-1}^{(k)}\,\cup\, G_{2i}^{(k)}$, and for each $e\in G_i^{(k+1)}$, define an\\
            \makebox[1.8cm]{} edge-coloring $\psi$ of $G^{(k+1)}_i$ by $\psi^{(k+1)}_i(e)=\begin{cases}
                \varphi^{(k)}_{2i-1}(e),& e\in G_{2i-1}^{(k)}\\
                \varphi^{(k)}_{2i}(e)+\left|\varphi_{2i-1}^{(k)}\right|, & e\in G_{2i}^{(k)}
                \end{cases}$\\ \makebox[1.8cm]{} \Comment{$\left|\varphi_{2i-1}^{(k)}\right|$ is the size of the palette of the coloring $\varphi_{2i-1}^{(k)}$}\\
            \makebox[1.8cm]{} (2) Using Procedure $\textsc{Reduce-Color}$, recolor the edges $E^{(k+1)}_i$ of $G_i^{(k+1)}$.\\ \makebox[1.8cm]{} Denote the resulting edge-coloring by $\varphi^{(k+1)}_i$.\\
        \makebox[1.35cm]{} \textbf{Endfor}\\
    \makebox[0.8cm]{} \textbf{Endfor}

    \item{\textbf{Step 7.}} \textbf{If} $\left|\varphi^{(h)}_1\right|=\Delta+2$ \textbf{then}\Comment{A special case where we need to reduce one more\\ \makebox[7.2cm]{}color}\\
    \makebox[1.35cm]{} $\varphi=\textsc{Reduce-Color}\left(G, \varphi^{(h)}_1\right)$\\
    \textbf{else} $\varphi=\varphi^{(h)}_1$
    
    \item{\textbf{return}} $\varphi$ 
\end{description}

In the next lemma we analyse the structure of the graphs $G^{(k)}_i$ throughout the recursion.

\begin{lemma}[Recursive decomposition]\label{subgraphs partition}
    Let $G=(V,E)$ be an input graph for Procedure \textsc{Edge-Coloring} with maximum degree $\Delta$. Then for each $k\in \left\{0,1,..., h\right\}$ and $i\in\left\{1,2,...,\frac{p}{2^k}\right\}$ ($h=\left\lceil\log\frac{\Delta}{2}\right\rceil$ and $p=2^h$), we have $$G^{(k)}_i=\bigcup_{j=1}^{2^{k}} G^{(0)}_{(i-1)\cdot2^{k}+j}.$$ $\left(\text{$G^{(k)}_i$ is equal to union of $2^{k}$ subgraphs that are defined in step 4.}\right)$
\end{lemma}

\begin{proof}
    We prove the lemma by induction on $k$.\\
    The induction base is $k=0$. For every $i\in\left\{1,2,...,p\right\}$, indeed we have $$\bigcup_{j=1}^{1} G^{(0)}_{(i-1)+j}=G^{(0)}_i.$$
    For the induction step, consider some index $k>0$. By the definition of $G^{(k)}_{i}$, and by the induction hypothesis,
    $$G^{(k)}_{i}=G^{(k-1)}_{2i-1}\,\cup\, G^{(k-1)}_{2i}=\left(\bigcup_{j=1}^{2^{k-1}} G^{(0)}_{(2i-2)\cdot2^{k-1}+j}\right)\cup\left(\bigcup_{j=1}^{2^{k-1}} G^{(0)}_{(2i-1)\cdot2^{k-1}+j}\right)=\bigcup_{j=1}^{2^{k}} G^{(0)}_{(i-1)\cdot2^{k}+j}.$$  
\end{proof}

In the next lemma, we analyse the size of the palettes of the edge-colorings $\varphi^{(k)}_i$, of the subgraphs $G^{(k)}_i$, respectively, throughout the recursion. Recall that $h=\left\lceil\log\frac{\Delta}{2}\right\rceil$ and $p=2^h.$

\begin{lemma}[Sizes of palettes in recursive subgraphs]\label{subcolorings properties}
    Let $G=(V,E)$ be an input graph for Procedure \textsc{Edge-Coloring} and let $k\in \left\{0,1,...,h\right\}$ and $i\in\left\{1,2,...,\frac{p}{2^k}\right\}$ be a pair of indexes. Then $\varphi^{(k)}_i$ is a proper $\left(\left(\sum_{j=1}^{2^{k}} \Delta\left(G^{(0)}_{(i-1)\cdot2^{k}+j}\right)\right)+1\right)$-edge-coloring of $G^{(k)}_i$.
\end{lemma}

\begin{proof}
    First, recall that by Lemma \ref{subgraphs partition}, $G^{(k)}_i=\bigcup_{j=1}^{2^{k}} G^{(0)}_{(i-1)\cdot 2^{k}+j}$.
    We prove the lemma by induction on $k$.\\
    The induction base is $k=0$. In step 5, for every $i\in\{1,2,...,p\}$, we indeed define $\varphi^{(0)}_i$ as a proper $\left(\Delta\left(G^{(0)}_i\right)+1\right)$-edge-coloring of $G^{(0)}_i$.\\
    For the induction step, consider $k>0$. We defined $$\psi^{(k)}_i(e)=\begin{cases}
                \varphi^{(k-1)}_{2i-1}(e),& e\in G_{2i-1}^{(k-1)}\\
                \varphi^{(k-1)}_{2i}(e)+\left|\varphi_{2i-1}^{(k-1)}\right|, & e\in G_{2i}^{(k-1)}
                \end{cases}.$$\\
    By the induction hypothesis, 
    \begin{itemize}
        \item $\varphi^{(k-1)}_{2i-1}$ is a proper $\left(\left(\sum_{j=1}^{2^{k-1}} \Delta\left(G^{(0)}_{(2i-2)\cdot2^{k-1}+j}\right)\right)+1\right)$-edge-coloring of $G^{(k-1)}_{2i-1}$.
        \item $\varphi^{(k-1)}_{2i}$ is a proper $\left(\left(\sum_{j=1}^{2^{k-1}} \Delta\left(G^{(0)}_{(2i-1)\cdot2^{k-1}+j}\right)\right)+1\right)$-edge-coloring of $G^{(k-1)}_{2i}$.
    \end{itemize}
    Hence $\psi^{(k)}_{i}$ is a proper $\left(\sum_{j=1}^{2^{k}} \Delta\left(G^{(0)}_{(i-1)\cdot2^{k}+j}\right)+2\right)$-edge-coloring of $G^{(k)}_{i}$.
    Since $G^{(k)}_i=\bigcup_{j=1}^{2^{k}} G^{(0)}_{(i-1)\cdot2^{k}+j}$, we have $$\Delta\left(G^{(k)}_i\right)\leq \sum_{j=1}^{2^{k}} \Delta\left(G^{(0)}_{(i-1)\cdot2^{k}+j}\right).$$ Hence, by Theorem \ref{Procedure Reduce-Color properties}, $\varphi^{(k)}_i$, that is computed in step 6 of Procedure \textsc{Edge-Coloring}, is a proper $\left(\sum_{j=1}^{2^{k}} \Delta\left(G^{(0)}_{(i-1)\cdot2^{k}+j}\right)+1\right)$-edge-coloring of $G^{(k)}_i$.
\end{proof}

We are now ready to show that ultimately, Procedure \textsc{Edge-Coloring} employs just $\Delta+1$ colors.

\begin{lemma}[Proper edge-coloring]
    Let $G=(V,E)$ be a graph with maximum degree $\Delta$.
    Procedure \textsc{Edge-Coloring} computes a proper $(\Delta+1)$-edge-coloring of $G$.
\end{lemma}

\begin{proof}
    By Lemma \ref{subcolorings properties}, $\varphi^{(h)}_1$ is a proper $\left(\left(\sum_{j=1}^{2^{h}} \Delta\left(G^{(0)}_j\right)\right)+1\right)$-edge-coloring of $G$.\\
    We next show that $$\Delta\leq\sum_{j=1}^{2^{h}} \Delta\left(G^{(0)}_j\right)\leq \Delta+1.$$ 
    Observe that for each $i\in\left\{1,2,...,\left\lceil\frac{\Delta}{2}\right\rceil\right\}$, since the edges in $G^{(0)}_i$ are defined by a color class of a proper edge-coloring of $G_B$, then each $v\in V$ has at most 2 incident edges in $G^{(0)}_i$ (at most one incident edge of $v^{(in)}$ and at most one incident edge of $v^{(out)}$). Hence $\Delta\left(G^{(0)}_i\right)\leq 2$. For each $i\in\left\{\left\lceil\frac{\Delta}{2}\right\rceil+1,...,p\right\}$, the graph $G^{(0)}_i$ is an empty graph. Hence $\Delta\left(G^{(0)}_i\right)=0$.\\
    Let $v\in V$ be a vertex such that $\deg_G(v)=\Delta$. Then we get that, $$\Delta=\sum_{j=1}^{2^h}\deg_{G^{(0)}_j}(v)\leq\sum_{j=1}^{2^h} \Delta\left(G^{(0)}_j\right)=\sum_{j=1}^{\left\lceil\frac{\Delta}{2}\right\rceil} \Delta\left(G^{(0)}_j\right)\leq \sum_{j=1}^{\left\lceil\frac{\Delta}{2}\right\rceil} 2\leq\Delta+1.$$
    Hence $\varphi^{(h)}_1$ is either a proper $\left(\Delta+1\right)$-edge-coloring of $G$,
    or a proper $\left(\Delta+2\right)$-edge-coloring of $G$. Therefore, by Theorem \ref{Procedure Reduce-Color properties}, the coloring $\varphi$, defined in step 7 of the algorithm, is a proper $(\Delta+1)$-edge-coloring of $G$.
\end{proof}

Next, we analyse the time and work complexities of our algorithm.

Procedure \textsc{Edge-Coloring} invokes Procedure \textsc{Reduce-Color} (on steps 6(2) and 7). The latter procedure depends on the parameter $\lambda$. We denote Procedure \textsc{Edge-Coloring}, invoked with the parameter $\lambda$, by \textsc{Edge-Coloring$_{\lambda}$}. The index is omitted when it does not affect the statement at hand.

\begin{lemma}[Complexity of Procedure \textsc{Edge-Coloring}]
        Let $G=(V,E)$ be a graph with maximum degree $\Delta$ and arboricity $a$. Procedure \textsc{Edge-Coloring$_{\lambda}$} requires
        $$O\left(\lambda\left(\Delta^2,a\cdot\Delta\right)\cdot\Delta^2\cdot \log n\cdot IST_{\lambda}\left(n,\Delta^2,a\cdot\Delta\right)\right)$$ time using $O\left(m\cdot\Delta\cdot ISP_{\lambda}\left(n,\Delta^2,a\cdot\Delta\right)+m\cdot\Delta\right)$ processors.
    \end{lemma}

\begin{proof}
    We start by analysing the complexity of steps 1-5:
    \begin{itemize}
        \item The graph $G^*$ can be constructed by assigning a processor $p_v$ to each vertex $v\in V$. The processor $p_v$ adds to the graph the edge $(v,x)$ if the degree of $v$ is odd in $O(1)$ time. This step uses $O(n)$ processors.
        \item By Lemma \ref{eulerian cycle alg}, the computation of the Eulerian cycle $C$ in $G^*$ requires $O(\log n)$ time using $O(n+m)$ processors.
        \item For the construction of $G_B$, we assign a processor to each directed edge $\langle u,v\rangle$ in the Eulerian cycle $C$ that defines the edge $\left(u^{(out)},v^{(in)}\right)$ in $G_B$. This process requires $O(1)$ time using $O(m)$ processors.
        \item By Lemma \ref{bipartite coloring}, the computation of the edge-coloring $\varphi_B$ of $G_B$ requires $O(\log^2 n\cdot\log\Delta)$ time using $O(n+m)$ processors.
        \item The construction of the subgraphs $G^{(0)}_i$, for each $i\in\{1,2,...,p\}$, requires $O(1)$ time using $O(m)$ processors. This is done by assigning a processor $p_{\left(u^{(out)},v^{(in)}\right)}$ to each edge $\left(u^{(out)},v^{(in)}\right)$ of $G_B$. The processor defines the edge $(u,v)\in E^{(0)}_i$, where $i$ is the $\varphi_B$-color of the edge $\left(u^{(out)},v^{(in)}\right)$.
        \item Since $\Delta\left(G^{(0)}_i\right)\leq 2$ for each subgraph $G_i$, by Theorem \ref{3 edge-coloring}, the computations of $(\Delta+1)$-edge-coloring of all the $O(\Delta)$ graphs $G^{(0)}_i$ in parallel requires $O(\log n)$ time using $O(n\cdot\Delta)$ processors. (By a more careful analysis, one can use $O(n+m)$ processors for this step. However, the bound $O(n\cdot \Delta)$ is sufficient for our purposes here.)
    \end{itemize}
    We now analyse the computational complexity of step 6:
    \begin{itemize} 
        \item For each $k\in\{0,1,...,h-1\}$ and $i\in\left\{1,2,...,\frac{p}{2^{k+1}}\right\}$, the construction of all the graphs $G^{(k+1)}_i$ and the colorings $\psi^{(k+1)}_i$ requires $O(1)$ time and $O(m)$ processors altogether. This is done by assigning a processor $p_e$ to each edge $e\in E$. This processor adds the edge $e$ to $G^{(k+1)}_i$ if $e\in G^{(k)}_{2i-1}$ or $e\in G^{(k)}_{2i}$, and computes its new color accordingly.
        \item Let $k\in\{0,1,...,h-1\}$. Observe that for each $i\in\left\{1,2,...,\frac{p}{2^{k+1}}\right\}$, by Lemma \ref{subgraphs partition}, we have $\Delta\left(G^{(k+1)}_i\right)\leq \sum_{j=1}^{2^{k+1}} \Delta\left(G^{(0)}_{(i-1)\cdot2^{k+1}+j}\right)\leq 2^{k+1}\cdot2=2^{k+2}$. Denote by $m^{(k+1)}_i$ the number of edges of the graph $G^{(k+1)}_i$. Observe that the subgraphs $G^{(k+1)}_i$ are pairwise edge-disjoint, i.e., $\sum_{i=1}^{\frac{p}{2^{k+1}}}m^{(k+1)}_i=m$. Hence, for all $i\in\left\{1,2,...,\frac{p}{2^{k+1}}\right\}$, by Theorem \ref{Procedure Reduce-Color properties}, the parallel executions of Procedure \textsc{Reduce-Color$_{\lambda}$} on all the graphs $G^{(k+1)}_i$ with the coloring $\psi^{(k+1)}_i$ require 
        $$O\left(\lambda\left(\left(2^{k+2}\right)^2,a\cdot 2^{k+2}\right)\cdot\left(2^{k+2}\right)^2\cdot \log n\cdot IST_{\lambda}\left(n,\left(2^{k+2}\right)^2,a\cdot 2^{k+2}\right)\right)$$
        time using 
        \begin{align*}
            \sum_{i=1}^{\frac{p}{2^{k+1}}}&O\left(m_i^{(k+1)}\cdot 2^{k+2}\cdot ISP_{\lambda}\left(n,\left(2^{k+2}\right)^2,a\cdot 2^{k+2}\right)+m_i^{(k+1)}\cdot 2^{k+2}\right)=\\
            &O\left(m\cdot 2^{k}\cdot ISP_{\lambda}\left(n,2^{2k},a\cdot 2^{k+2}\right)+m\cdot 2^{k}\right)
        \end{align*} processors altogether. Hence, step 6 requires 
        \begin{equation}\label{eq: processors}
        \begin{aligned}            \sum_{k=0}^{h-1}&\,O\left(\lambda\left(\left(2^{k+2}\right)^2,a\cdot 2^{k+2}\right)\cdot\left(2^{k+2}\right)^2\cdot \log n\cdot IST_{\lambda}\left(n,\left(2^{k+2}\right)^2,a\cdot 2^{k+2}\right)\right)=\\
        =&O\left(\lambda\left(\Delta^2,a\cdot\Delta\right)\cdot \log n\cdot IST_{\lambda}\left(n,\Delta^2,a\cdot\Delta\right)\right)\cdot\sum_{k=0}^{h-1}\,\left(2^{k+2}\right)^2=\\
        =&\,O\left(\lambda\left(\Delta^2,a\cdot\Delta\right)\cdot\Delta^2\cdot \log n\cdot IST_{\lambda}\left(n,\Delta^2,a\cdot\Delta\right)\right)
        \end{aligned}
        \end{equation}
         
        time using 
        \begin{align*}
            \max_{0\leq k\leq h-1}&\left\{O\left(m\cdot 2^{k}\cdot ISP_{\lambda}\left(n,2^{2k},a\cdot 2^{k}\right)+m\cdot 2^k\right)\right\}=\\
            &O\left(m\cdot\Delta\cdot ISP_{\lambda}\left(n,\Delta^2,a\cdot\Delta\right)+m\cdot\Delta\right)
        \end{align*}
        processors. (Equation (\ref{eq: processors}) holds for any polynomial function $\lambda(\cdot)$.)
        \item By Theorem \ref{Procedure Reduce-Color properties}, step 7 also requires $O\left(\lambda\left(\Delta^2,a\cdot\Delta\right)\cdot\Delta^2\cdot \log n\cdot IST_{\lambda}\left(n,\Delta^2,a\cdot\Delta\right)\right)$ time using $O\left(m\cdot\Delta\cdot ISP_{\lambda}\left(n,\Delta^2,a\cdot\Delta\right)+m\cdot\Delta\right)$ processors.
    \end{itemize}

    To summarize, Procedure \textsc{Edge-Coloring} requires $$O\left(\lambda\left(\Delta^2,a\cdot\Delta\right)\cdot\Delta^2\cdot \log n\cdot IST_{\lambda}\left(n,\Delta^2,a\cdot\Delta\right)\right)$$ time using $O\left(m\cdot\Delta\cdot ISP_{\lambda}\left(n,\Delta^2,a\cdot\Delta\right)+m\cdot\Delta\right)$ processors.
\end{proof}

Our analysis implies the following corollary:

\begin{restatable}[Properties of Procedure \textsc{Edge-Coloring}]{corollary}{ProcedureEdgeColoringproperties}\label{Cor: Procedure Edge Coloring properties}
    Let $G=(V,E)$ be a graph with maximum degree $\Delta$ and arboricity $a$.
    Procedure \textsc{Edge-Coloring} computes a proper $(\Delta+1)$-edge-coloring of $G$ in $O\left(\lambda\left(\Delta^2,a\cdot\Delta\right)\cdot\Delta^2\cdot \log n\cdot IST_{\lambda}\left(n,\Delta^2,a\cdot\Delta\right)\right)$ time using $O\left(m\cdot\Delta\cdot ISP_{\lambda}\left(n,\Delta^2,a\cdot\Delta\right)+m\cdot\Delta\right)$ processors.
\end{restatable}

In the next theorem, we summarize the complexities of our edge-coloring algorithm using each of the algorithms for computing large independent sets from Theorem \ref{Large independent set alg}.
\begin{theorem}[Properties of Procedure \textsc{Edge-Coloring}]\label{Procedure Edge-Coloring properties}
    Let $G=(V,E)$ be an $n$-vertex $m$-edge graph with maximum degree $\Delta$ and arboricity $a$. Procedure \textsc{Edge-Coloring} computes a proper $(\Delta+1)$-edge-coloring of $G$
    \begin{enumerate}
        \item[(1)] (Using Theorem \ref{Large independent set alg}(1)) in $O\left(\Delta^4\cdot\log^4 n\right)$ time using $O\left(m\cdot\Delta\right)$ processors.
        \item[(2)] (Using Theorem \ref{Large independent set alg}(2)) in $O\left(\Delta^4\cdot \log^2 n+\Delta^6\cdot\log^2\Delta\cdot\log n\right)$ time using $O\left(m\cdot\Delta\right)$ processors.
        \item[(3)] (Using Theorem \ref{Large independent set alg}(3)) in $O\left(a^2\cdot\Delta^4\cdot\log\Delta\cdot\log^2 n\right)$ time using $O\left(m\cdot\Delta\right)$ processors.
        \item[(4)] (Using Theorem \ref{Large independent set alg}(4)) in $O\left(\Delta^{3+o(1)}\cdot a^{1+o(1)}\cdot\log^2 n\right)$ time using $O\left(m\cdot\Delta\cdot\frac{\log^{\delta}\Delta\cdot\log n}{\log(\Delta\cdot\log n)}\right)$ processors, for any constant $\delta>0$.  
        \item[(5)] (Using Theorem \ref{Large independent set alg}(5)) in $O\left(\Delta^{5}\cdot\log n\cdot\left(\log\Delta+\frac{\log n}{\log\Delta\cdot\log(\Delta\cdot\log n)}\right)+\Delta^4\cdot\log^2 n\right)=O\left(\Delta^5\cdot\log^2 n\right)$ time using \\$O\left(m\cdot\left(\Delta\cdot\log\Delta+\frac{\sqrt{\Delta}\cdot\log n}{\log(\Delta\cdot\log n)}\right)\right)$ processors.
    \end{enumerate}
\end{theorem}

\begin{proof}
    \begin{enumerate}
        \item[(1)] In the algorithm of Theorem \ref{Large independent set alg}(1), we have 
        \begin{align*}
            \lambda\left(\Delta^2,a\cdot\Delta\right)&=O\left(\Delta^2\right)\\
            IST_{\lambda}\left(n,\Delta^2, a\cdot\Delta\right)&=O\left(\log^3 n\right)\\
            ISP_{\lambda}\left(n,\Delta^2, a\cdot\Delta\right)&=O\left(\frac{1}{\log n}\right)
        \end{align*}
        Hence, by Corollary \ref{Cor: Procedure Edge Coloring properties}, Procedure \textsc{Edge-Coloring} requires $$O\left(\Delta^2\cdot\Delta^2\cdot\log n\cdot \log^3 n\right)=O\left(\Delta^4\cdot\log^4 n\right)$$ 
        time using $O\left(m\cdot\Delta\right)$ processors.

        \item[(2)] In the algorithm of Theorem \ref{Large independent set alg}(2), we have 
        \begin{align*}
            \lambda\left(\Delta^2,a\cdot\Delta\right)&=O\left(\Delta^2\right)\\
            IST_{\lambda}\left(n,\Delta^2, a\cdot\Delta\right)&=O\left(\log n+\Delta^2\cdot\log^2\Delta\right)\\
            ISP_{\lambda}\left(n,\Delta^2, a\cdot\Delta\right)&=O\left(1\right)
        \end{align*}
        Hence, by Corollary \ref{Cor: Procedure Edge Coloring properties}, Procedure \textsc{Edge-Coloring} requires $$O\left(\Delta^2\cdot\Delta^2\cdot\log n\cdot\left(\log n+\Delta^2\cdot\log^2\Delta\right)\right)=O\left(\Delta^4\cdot \log^2 n+\Delta^6\cdot\log^2\Delta\cdot\log n\right)$$ 
        time using $O(m\cdot\Delta)$ processors.

        \item[(3)] In the algorithm of Theorem \ref{Large independent set alg}(3), we have 
        \begin{align*}
            \lambda\left(\Delta^2,a\cdot\Delta\right)&=O\left(a\cdot\Delta\right)\\
            IST_{\lambda}\left(n,\Delta^2, a\cdot\Delta\right)&=O\left(a\cdot\Delta\cdot\log\Delta\cdot\log n\right)\\
            ISP_{\lambda}\left(n,\Delta^2, a\cdot\Delta\right)&=O\left(1\right)
        \end{align*}
        Hence, by Corollary \ref{Cor: Procedure Edge Coloring properties}, Procedure \textsc{Edge-Coloring} requires $$O\left(a\cdot\Delta\cdot\Delta^2\cdot\log n\cdot\left(a\cdot\Delta\cdot\log\Delta\cdot\log n\right)\right)=O\left(a^2\cdot\Delta^4\cdot\log\Delta\cdot\log^2 n\right)$$ 
        time using $O(m\cdot\Delta)$ processors.

        \item[(4)] In the algorithm of Theorem \ref{Large independent set alg}(4), we have 
        \begin{align*}
            \lambda\left(\Delta^2,a\cdot\Delta\right)&=O\left(\left(a\cdot\Delta\right)^{1+o(1)}\right)\\
            IST_{\lambda}\left(n,\Delta^2, a\cdot\Delta\right)&=O\left(\log^{1+\delta}\Delta\cdot\log n\right)\\
            ISP_{\lambda}\left(n,\Delta^2, a\cdot\Delta\right)&=O\left(\frac{\log^{\delta}\Delta\cdot\log n}{\log(\Delta\cdot\log n)}\right)
        \end{align*}
        Hence, by Corollary \ref{Cor: Procedure Edge Coloring properties}, Procedure \textsc{Edge-Coloring} requires $$O\left(\left(a\cdot\Delta\right)^{1+o(1)}\cdot\Delta^2\cdot\log n\cdot\left(\log^{1+\delta}\Delta\cdot\log n\right)\right)=O\left(\Delta^{3+o(1)}\cdot a^{1+o(1)}\cdot\log^2 n\right)$$ 
        time using $O\left(m\cdot\Delta\cdot\left(\frac{\log^{\delta}\Delta\cdot\log n}{\log(\Delta\cdot\log n)}\right)\right)$ processors.
        \item[(5)] In the algorithm of Theorem \ref{Large independent set alg}(5), we have 
        \begin{align*}
            \lambda\left(\Delta^2,a\cdot\Delta\right)&=O\left(\Delta^{2}\right)\\
            IST_{\lambda}\left(n,\Delta^2, a\cdot\Delta\right)&=O\left(\Delta\cdot\left(\log\Delta+\frac{\log n}{\log\Delta\cdot\log(\Delta\cdot\log n)}\right)+\log n\right)\\
            ISP_{\lambda}\left(n,\Delta^2, a\cdot\Delta\right)&=O\left(\Delta\cdot\log\Delta+\frac{\sqrt{\Delta}\cdot\log n}{\log(\Delta\cdot\log n)}\right)
        \end{align*}
        Hence, by Corollary \ref{Cor: Procedure Edge Coloring properties}, Procedure \textsc{Edge-Coloring} requires 
        \begin{align*}
            &O\left(\Delta^2\cdot\Delta^2\cdot\log n\cdot\left(\Delta\cdot\left(\log\Delta+\frac{\log n}{\log\Delta\cdot\log(\Delta\cdot\log n)}\right)+\log n\right)\right)=\\
            &O\left(\Delta^{5}\cdot\log n\cdot\left(\log\Delta+\frac{\log n}{\log\Delta\cdot\log(\Delta\cdot\log n)}\right)+\Delta^4\cdot\log^2 n\right)=O\left(\Delta^5\cdot\log^2 n\right)\footnotemark
        \end{align*}\footnotetext{The first term in the left-hand side expression typically dominates the second one, except for a narrow range $\omega(1)=\Delta=o\left(\log\log n\right)$.}
        time using $O\left(m\cdot\left(\Delta\cdot\log\Delta+\frac{\sqrt{\Delta}\cdot\log n}{\log(\Delta\cdot\log n)}\right)\right)$ processors.
    \end{enumerate}
\end{proof}

\newpage

\section{$(1+\varepsilon)\Delta$-edge-coloring}\label{Our alg}

In this section we present our $(1+\varepsilon)\Delta$-edge-coloring algorithm, and analyse it.
Our algorithm uses a degree-splitting method. We start by presenting the \emph{Degree-Splitting} problem.\\

\textbf{The Degree-Splitting Problem.}
The undirected degree-splitting problem seeks to partition the graph edges $E$ into two parts so that the partition looks almost balanced around each vertex. Concretely, we need to color each edge red or blue such that for each node, the difference between the number of red and blue edges incident on it is at most some small discrepancy threshold value $\kappa$. In other words, we want an assignment $q:E \rightarrow \{1, -1\}$ such that for each node $v\in V$, we have $$\left|\sum_{e\in E_v} q(e)\right|\leq \kappa,$$
where $E_v$ denotes the set of edges incident on $v$.\\
For every $v\in V$, we refer to the sum $\sum_{e\in E_v} q(e)$ as the \emph{balance} of $v$.\\

Next, we present an algorithm that computes a degree-splitting with discrepancy at most 2. The algorithm requires $O(\log n)$ time using $O(m)$ processors (see also~\cite{israeli1986improved} for a related algorithm).
Note that in Section~\ref{pram edge coloring} we also used a degree-splitting routine, in which the graph was partitioned into approximately $\Delta/2$ subgraphs, each having maximum degree at most~$2$. In contrast, here we partition the graph into just two subgraphs, while ensuring that the degree of each vertex is approximately halved between them.

Next, we state two results that we will use in our algorithm.
The first result is the algorithm due to Atallah and Vishkin for computing Eulerian cycle in Eulerian graphs in $O(\log n)$ time using $O(m)$ processors~\cite{atallah1984finding} (see Lemma \ref{eulerian cycle alg}). The second result is an algorithm for 2-edge-coloring a path of length $n$. The algorithm requires $O(\log n)$ time using $O(n)$ processors. The full description of the latter algorithm (called Procedure \textsc{Alternating-Coloring}) is given in Appendix \ref{App: coloring}.

\begin{restatable}[Alternating path coloring]{lemma}{altPath}
\label{path coloring}
    Let $P=(V,E)=(v_0,e_1,v_1,e_2,v_2,...,v_{m-1},e_m,v_m)$ be an $n$-vertex path or an even length cycle with $m$ edges. Procedure \textsc{Alternating-Coloring} edge-colors $P$ using two alternating colors in $O(\log n)$ time using $O(m)$ processors.
\end{restatable}
For an odd length cycle, essentially the same procedure computes a 2-edge-coloring in which all colors are alternating, except for two certain consecutive edges that are colored by the same color.

Now we present the degree-splitting algorithm. The idea of the algorithm comes from a special case of a result due to Beck and Fiala~\cite{beck1981integer}, who showed that any hypergraph of rank $t$ (each
hyperedge has at most $t$ vertices) admits a 2-edge-coloring with discrepancy at most $2t-2$. The degree-splitting algorithm is also used in previous works~\cite{israeli1986improved, ghaffari2020improved} in $\mathrm{PRAM}$ and distributed settings.
The idea of the algorithm for the case of simple graphs ($t=2$) is as follows:\\
For a connected input graph $G=(V,E)$ with an even number of edges, we add a dummy node to $G$, and connect it to all the odd-degree vertices of the graph to obtain a graph $G'$. In this way all the degrees of $G'$ are even, and it is possible to compute an Eulerian cycle in this graph efficiently. Observe that this cycle is of even length. So we can color the edges of this cycle with two alternating colors and then by deleting the dummy node and its incident edges we will receive the desired coloring of the edges of $G$ with discrepancy at most 1. 
In the case that the input graph has an odd number of edges, we will analyze this algorithm more carefully and achieve a coloring of the edges of $G$ with discrepancy at most 2.

The description of this algorithm is given in Theorem \ref{Procedure Degree-Splitting}. See also Figure \ref{fig:deg-splitting} for an illustration of the degree-splitting procedure.

\begin{theorem}[Procedure \textsc{Degree-Splitting}]\label{Procedure Degree-Splitting}
    Let $G=(V,E)$ be a graph with $n$ vertices and $m$ edges. A degree-splitting of $G$ with discrepancy $\kappa=2$ can be computed in $O(\log n)$ time using $O(m)$ processors.
\end{theorem}

\begin{proof}
    Without loss of generality, we assume that the graph $G$ is connected. Otherwise, we will apply the algorithm on each connected component separately.\\
    Define a graph $G'=(V',E')$, where $V'=V\cup\{x\}$, for a dummy node $x\notin V$, and $E'=E\cup\{(x,v)\;|\;v\in V, \text{ and $\deg_G(v)$ is odd}\}$. Observe that since we increased by 1 only the degrees of the odd-degree vertices, then for any $v\in V$, its degree in $G'$, $\deg_{G'}(v)$, is even. In addition, since $\sum_{v\in V}\deg_G(v)=2\cdot|E|$, the number of odd-degree vertices in $G$ must be even. Hence the degree of $x$ (in $G'$) is even as well. Hence, all the degrees in $G'$ are even, and we can compute an Eulerian cycle in $G'$ using the algorithm from Lemma \ref{eulerian cycle alg}~\cite{atallah1984finding} in $O(\log n)$ time using $O(m)$ processors. Next, starting from the dummy node $x$ (if $x$ is isolated in $G'$ we will start from an arbitrary other vertex), we color the edges of the cycle with two alternating colors, red and blue, except maybe the first and the last edge that we color, that might have the same color (if the cycle is of odd length) using Corollary \ref{cycles} (see Appendix \ref{App: coloring}) in $O(\log n)$ time using $O(m)$ processors. In that way, the balance of each vertex in $G'$ is either 2 (if we started coloring the cycle from this vertex and the cycle is of odd length), or 0 (otherwise).
    After removing $x$ (and the edges it is incident on) from $G'$, the discrepancy of its neighbors will become 1, and the discrepancy of the other vertices will not change. So we get a degree-splitting of $G$ with discrepancy at most 2 in $O(\log n)$ time using $O(m)$ processors.
\end{proof}

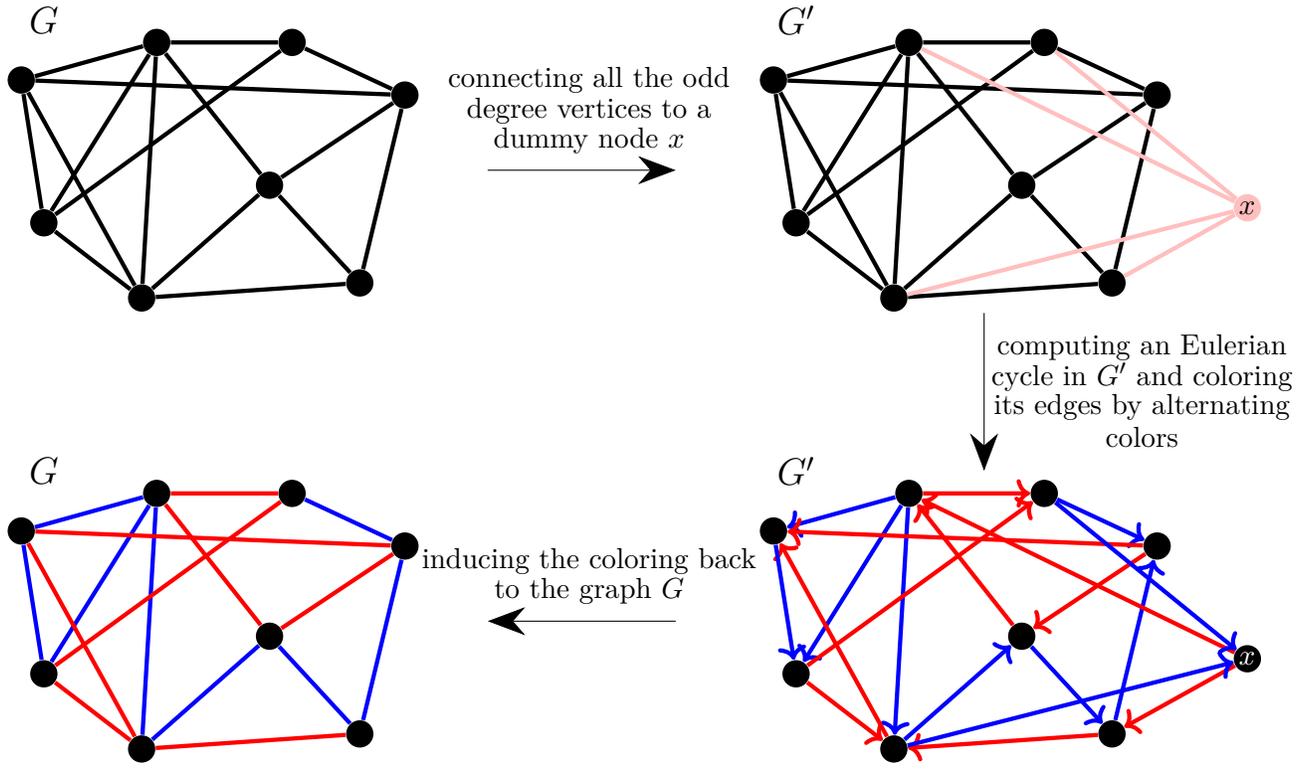
\begin{figure}
    \begin{center}
    \addtolength{\leftskip} {-2cm} 
    \addtolength{\rightskip}{-2cm}
    \begin{tikzpicture}
    \begin{scope}[xshift=-5cm]
    
        \node[circle, fill=black] at (0.3,1.3) (n1) {};
        \node[circle, fill=black] at (-1,2.3) (n2) {};
        \node[circle, fill=black] at (2,2.8) (n3) {};
        \node[circle, fill=black] at (0.5,4.7) (n4) {};
        \node[circle, fill=black] at (3.8,4) (n5) {};
        \node[circle, fill=black] at (-1.3,4.2) (n6) {};
        \node[circle, fill=black] at (2.3,4.7) (n7) {};
        \node[circle, fill=black] at (3.2,1.5) (n8) {};

        \path [-](n1) edge[ultra thick] (n2);
        \path [-](n1) edge[ultra thick] (n4);
        \path [-](n1) edge[ultra thick] (n3);
        \path [-](n2) edge[ultra thick] (n4);
        \path [-](n3) edge[ultra thick] (n5);
        \path [-](n4) edge[ultra thick] (n7);
        \path [-](n4) edge[ultra thick] (n6);
        \path [-](n5) edge[ultra thick] (n7);
        \path [-](n3) edge[ultra thick] (n4);
        \path [-](n1) edge[ultra thick] (n6);
        \path [-](n3) edge[ultra thick] (n8);
        \path [-](n5) edge[ultra thick] (n8);
        \path [-](n1) edge[ultra thick] (n8);
        \path [-](n2) edge[ultra thick] (n7);
        \path [-](n2) edge[ultra thick] (n6);
        \path [-](n5) edge[ultra thick] (n6);
    
    \end{scope}

    \begin{scope}[xshift=5cm]
    
        \node[circle, fill=black] at (0.3,1.3) (n1) {};
        \node[circle, fill=black] at (-1,2.3) (n2) {};
        \node[circle, fill=black] at (2,2.8) (n3) {};
        \node[circle, fill=black] at (0.5,4.7) (n4) {};
        \node[circle, fill=black] at (3.8,4) (n5) {};
        \node[circle, fill=black] at (-1.3,4.2) (n6) {};
        \node[circle, fill=black] at (2.3,4.7) (n7) {};
        \node[circle, fill=black] at (3.2,1.5) (n8) {};
        \node[circle, fill=pink] at (5,2.5) (n9) {};

        \path [-](n1) edge[ultra thick] (n2);
        \path [-](n1) edge[ultra thick] (n4);
        \path [-](n1) edge[ultra thick] (n3);
        \path [-](n2) edge[ultra thick] (n4);
        \path [-](n3) edge[ultra thick] (n5);
        \path [-](n4) edge[ultra thick] (n7);
        \path [-](n4) edge[ultra thick] (n6);
        \path [-](n5) edge[ultra thick] (n7);
        \path [-](n3) edge[ultra thick] (n4);
        \path [-](n1) edge[ultra thick] (n6);
        \path [-](n3) edge[ultra thick] (n8);
        \path [-](n5) edge[ultra thick] (n8);
        \path [-](n1) edge[ultra thick] (n8);
        \path [-](n2) edge[ultra thick] (n7);
        \path [-](n2) edge[ultra thick] (n6);
        \path [-](n5) edge[ultra thick] (n6);
        \path [-](n1) edge[ultra thick, pink] (n9);
        \path [-](n4) edge[ultra thick, pink] (n9);
        \path [-](n7) edge[ultra thick, pink] (n9);
        \path [-](n8) edge[ultra thick, pink] (n9);
    
    \end{scope}
    \begin{scope}[xshift=5cm, yshift=-6cm]
    
        \node[circle, fill=black] at (0.3,1.3) (n1) {};
        \node[circle, fill=black] at (-1,2.3) (n2) {};
        \node[circle, fill=black] at (2,2.8) (n3) {};
        \node[circle, fill=black] at (0.5,4.7) (n4) {};
        \node[circle, fill=black] at (3.8,4) (n5) {};
        \node[circle, fill=black] at (-1.3,4.2) (n6) {};
        \node[circle, fill=black] at (2.3,4.7) (n7) {};
        \node[circle, fill=black] at (3.2,1.5) (n8) {};
        \node[circle, fill=black] at (5,2.5) (n9) {};

        \path [->](n2) edge[ultra thick, red] (n1);
        \path [->](n4) edge[ultra thick, blue] (n1);
        \path [->](n1) edge[ultra thick, blue] (n3);
        \path [->](n4) edge[ultra thick, blue] (n2);
        \path [->](n5) edge[ultra thick, red] (n3);
        \path [->](n4) edge[ultra thick, red] (n7);
        \path [->](n4) edge[ultra thick, blue] (n6);
        \path [->](n7) edge[ultra thick, blue] (n5);
        \path [->](n3) edge[ultra thick, red] (n4);
        \path [->](n1) edge[ultra thick, red] (n6);
        \path [->](n3) edge[ultra thick, blue] (n8);
        \path [->](n8) edge[ultra thick, blue] (n5);
        \path [->](n8) edge[ultra thick, red] (n1);
        \path [->](n2) edge[ultra thick, red] (n7);
        \path [->](n6) edge[ultra thick, blue] (n2);
        \path [->](n5) edge[ultra thick, red] (n6);
        \path [->](n9) edge[ultra thick, red] (n4);
        \path [->](n9) edge[ultra thick, red] (n8);
        \path [->](n1) edge[ultra thick, blue] (n9);
        \path [->](n7) edge[ultra thick, blue] (n9);
    
    \end{scope}

    \begin{scope}[xshift=-5cm, yshift=-6cm]
    
        \node[circle, fill=black] at (0.3,1.3) (n1) {};
        \node[circle, fill=black] at (-1,2.3) (n2) {};
        \node[circle, fill=black] at (2,2.8) (n3) {};
        \node[circle, fill=black] at (0.5,4.7) (n4) {};
        \node[circle, fill=black] at (3.8,4) (n5) {};
        \node[circle, fill=black] at (-1.3,4.2) (n6) {};
        \node[circle, fill=black] at (2.3,4.7) (n7) {};
        \node[circle, fill=black] at (3.2,1.5) (n8) {};

        \path [-](n2) edge[ultra thick, red] (n1);
        \path [-](n4) edge[ultra thick, blue] (n1);
        \path [-](n1) edge[ultra thick, blue] (n3);
        \path [-](n4) edge[ultra thick, blue] (n2);
        \path [-](n5) edge[ultra thick, red] (n3);
        \path [-](n4) edge[ultra thick, red] (n7);
        \path [-](n4) edge[ultra thick, blue] (n6);
        \path [-](n7) edge[ultra thick, blue] (n5);
        \path [-](n3) edge[ultra thick, red] (n4);
        \path [-](n1) edge[ultra thick, red] (n6);
        \path [-](n3) edge[ultra thick, blue] (n8);
        \path [-](n8) edge[ultra thick, blue] (n5);
        \path [-](n8) edge[ultra thick, red] (n1);
        \path [-](n2) edge[ultra thick, red] (n7);
        \path [-](n6) edge[ultra thick, blue] (n2);
        \path [-](n5) edge[ultra thick, red] (n6);
    
    \end{scope}

    \draw [-{Stealth[length=5mm]}] (-0.1,3) -- (2.4,3);
    \node at (1.25,4.2){connecting all the odd};
    \node at (1.25,3.8){degree vertices to a};
    \node at (1.25,3.4){dummy node $x$};
    \draw [-{Stealth[length=5mm]}] (2.4,-3) -- (-0.1,-3);
    \node at (1.25,-2.2){inducing the coloring back};
    \node at (1.25,-2.6){to the graph $G$};
    \draw [-{Stealth[length=5mm]}] (6.5,1.1) -- (6.5,-1);
    \node at (8.6,0.65){computing an Eulerian};
    \node at (8.6,0.25){cycle in $G'$ and coloring};
    \node at (8.6,-0.15){its edges by alternating};
    \node at (8.6,-0.55){colors};
    \node at (-6,5){\Large $G$};
    \node at (4,5){\Large $G'$};
    \node at (4,-1){\Large $G'$};
    \node at (-6,-1){\Large $G$};
    \node at (10,2.5){$x$};
    \node at (10,-3.5){\color{white}$x$};
    
    \end{tikzpicture}
    \end{center}
    \caption{Procedure \textsc{Degree-Splitting}}
    \label{fig:deg-splitting}
\end{figure}

For $i\in\{-1,1\}$, denote by $\deg_i(v)$ the number of edges incident on $v$ that are colored $i$ in the degree-splitting.
As $|\deg_{-1}(v)-\deg_1(v)|\leq 2$ for any $v\in V$, it follows that
\begin{equation}\label{eq4.1}
    \deg_i(v)\leq \frac{\deg(v)}{2}+1.
\end{equation}
for every $i\in\{-1,1\}$.\\

Using this degree-splitting algorithm we now describe our $(1+\varepsilon)\Delta$-edge-coloring algorithm. The algorithm receives as an input a graph $G=(V,E)$, and a non-negative integer parameter $h$. It first computes a degree-splitting of the input graph $G=(V,E)$, then defines two subgraphs $G_1$ and $G_2$ of $G$, each on the same vertex set $V$, and the set of edges of each of them is defined by the edges that are colored with the same color in the degree-splitting algorithm from the first step.
Then, in parallel, we compute recursively colorings of each of the subgraphs of $G$, and merge these colorings using disjoint palettes. The parameter $h$ determines the depth of the recursion. At the base case of the algorithm ($h=0$) it computes a proper $(\Delta+1)$-edge-coloring using an algorithm from Theorem \ref{Procedure Edge-Coloring properties}.\\
The pseudo-code of Procedure \textsc{Approx-Edge-Coloring} is described in Algorithm \ref{Procedure Edge-Coloring}. See also Figure \ref{fig:edge-coloring} for an illustration of Procedure \textsc{Approx-Edge-Coloring}.

\begin{algorithm}
  \caption{Procedure \textsc{Aprrox-Edge-Coloring} ($G=(V,E),h$)\label{Procedure Edge-Coloring}}
  \begin{algorithmic} [1] %the [1] makes the numbering%
  \If{$h=0$}
    \LongState{compute a $(\Delta+1)$-edge-coloring $\varphi$ of $G$ using an algorithm from Theorem \ref{Procedure Edge-Coloring properties}}
    \Else
    \Let{$G_1=(V,E_1),G_2=(V,E_2)$}{\textsc{Degree-Splitting}($G$)\Comment{see Theorem \ref{Procedure Degree-Splitting}}}
    \For{each $i\in \{1,2\}$ in parallel}
    \LongState{compute a coloring $\varphi_i$ of $G_i$ using Procedure \textsc{Aprrox-Edge-Coloring}($G_i,h-1$)}
    \EndFor
    \State define coloring $\varphi$ of $G$ by $\varphi(e)=
    \begin{cases}
    \varphi_1(e),& e\in E_1\\
    |\varphi_1|+\varphi_2(e), & e\in E_2 
    \end{cases}$ \myindent{5.5}\Comment{$|\varphi_1|$ is the size of the palette of the coloring $\varphi_1$}
    \EndIf
    \State \textbf{return} $\varphi$
  \end{algorithmic}
\end{algorithm}

\begin{figure}
    \begin{center}
    \addtolength{\leftskip} {-2cm} 
    \addtolength{\rightskip}{-2cm}
    \begin{tikzpicture}

        \begin{scope}[xshift=-7.5cm]            
            \node[draw, ellipse, minimum size = 1.3cm] (nn) at (7,6){$G$};

        \end{scope}
        
        \begin{scope}[xshift=-2.5cm]
            \node[draw, ellipse, minimum size = 0.9cm] (31) at (6,4.5){$G_1$};
            \node[draw, ellipse, minimum size = 0.9cm] (3n) at (8,4.5){$G_2$};
            
            \node[draw, ellipse, minimum size = 1.3cm] (nn) at (7,6){$G$};

            \path [->](nn) edge[ultra thick] (31);
            \path [->](nn) edge[ultra thick] (3n);

        \end{scope}

        \begin{scope}[xshift=2.5cm]
            \node[draw, ellipse, minimum size = 0.9cm] (31) at (6,4.5){$G_1$};
            \node[draw, ellipse, minimum size = 0.9cm] (3n) at (8,4.5){$G_2$};
            
            \node[draw, ellipse, minimum size = 1.3cm] (nn) at (7,6){$G$};
    
            \path [-](nn) edge[ultra thick] (31);
            \path [-](nn) edge[ultra thick] (3n);

            \node at (5.8,4.25) {\color{red}$\bullet$};
            \node at (6,4.25) {\color{green}$\bullet$};
            \node at (6.2,4.25) {\color{blue}$\bullet$};

            \node at (7.8,4.25) {\color{cyan}$\bullet$};
            \node at (8,4.25) {\color{magenta}$\bullet$};
            \node at (8.2,4.25) {\color{gray}$\bullet$};

        \end{scope}

        \begin{scope}[xshift=7.5cm]
            \node[draw, ellipse, minimum size = 0.9cm] (31) at (6,4.5){$G_1$};
            \node[draw, ellipse, minimum size = 0.9cm] (3n) at (8,4.5){$G_2$};
            
            \node[draw, ellipse, minimum size = 1.3cm] (nn) at (7,6){$G$};
    
            \path [->](31) edge[ultra thick] (nn);
            \path [->](3n) edge[ultra thick] (nn);

            \node at (5.8,4.25) {\color{red}$\bullet$};
            \node at (6,4.25) {\color{green}$\bullet$};
            \node at (6.2,4.25) {\color{blue}$\bullet$};

            \node at (7.8,4.25) {\color{cyan}$\bullet$};
            \node at (8,4.25) {\color{magenta}$\bullet$};
            \node at (8.2,4.25) {\color{gray}$\bullet$};
            
            \node at (6.7,5.75) {\color{red}$\bullet$};
            \node at (7,5.75) {\color{green}$\bullet$};
            \node at (7.3,5.75) {\color{blue}$\bullet$};

            \node at (6.7,5.58) {\color{cyan}$\bullet$};
            \node at (7,5.58) {\color{magenta}$\bullet$};
            \node at (7.3,5.58) {\color{gray}$\bullet$};

        \end{scope}

    \draw [-{Stealth[length=5mm]}] (1,6) -- (3,6);
    \node at (2,7.2){Partition the graph};
    \node at (2,6.8){using Procedure };
    \node at (2,6.4){\textsc{Degree-Splitting}};
    \draw [-{Stealth[length=5mm]}] (5.9,6) -- (7.9,6);
    \node at (6.95,7.2){Recursively color the};
    \node at (6.95,6.8){subgraphs $G_1$ and $G_2$};
    \node at (6.95,6.4){with disjoint palettes};
    \draw [-{Stealth[length=5mm]}] (10.8,6) -- (12.8,6);
    \node at (11.9,6.4){Merge the colorings};

    \end{tikzpicture}
    \end{center}
    \caption{Procedure \textsc{Aprrox-Edge-Coloring}}
    \label{fig:edge-coloring}
\end{figure}
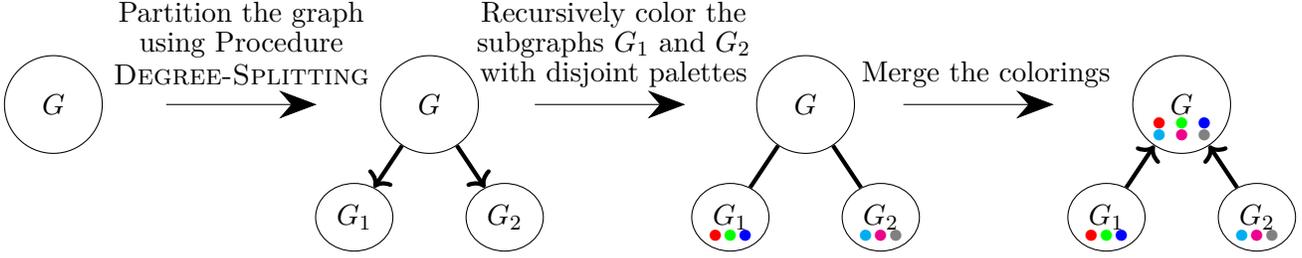

\subsection{Analysis of the Algorithm}
In this section we analyse Procedure \textsc{Aprrox-Edge-Coloring}.
We start by showing that the coloring $\varphi$ computed by Procedure \textsc{Aprrox-Edge-Coloring} is a proper edge-coloring. We then bound the number of colors that $\varphi$ employs. We will also bound the time and the number of processors used by the algorithm, and present a simple trade-off between the number of colors that the algorithm uses and its time and work complexities.\\

In~\cite{liang1995fast}, Liang presented a similar algorithm for $(1+\varepsilon)\Delta$-edge-coloring that requires $O(\varepsilon^{-4.5}\cdot\log^3(\varepsilon^{-1})\cdot$$\log n+\varepsilon^{-4}\cdot\log^4 n)$ time and $O\left(n^2+n\cdot\varepsilon^{-3}\right)$ processors. We provide a simpler analysis for the number of colors used by the algorithm, and derive a better bound for both the running time and the number of processors that comes from our improved $(\Delta+1)$-edge-coloring algorithms (see Section \ref{pram edge coloring}).\\

We begin by presenting some notation that we will use for the analysis. Let $G=(V,E)$ be an input for Procedure \textsc{Aprrox-Edge-Coloring}. We will use the notation $G^{(i)}$ for a graph that is computed after $i$ levels of recursion. We denote its number of edges by $m^{(i)}$ and its maximum degree by $\Delta^{(i)}$ (for example $G^{(0)}=G$, $m^{(0)}=m$, and $\Delta^{(0)}=\Delta$).\\

In our analysis, one of the main parameters that affects the performance is the maximum degree of the input 
graph. Therefore, we start the analysis by bounding the maximum degree of the input graph in every level of the recursion.

\begin{claim}[A bound on the maximum degree of $G^{(i)}$]\label{deg bound}
    Let $G=(V,E)$ and $h\leq \log\Delta$ be an input for Procedure \textsc{Aprrox-Edge-Coloring}, and let $0\leq i\leq h$. The maximum degree $\Delta^{(i)}$ of a graph $G^{(i)}$ that was computed after $i$ levels of recursion satisfies $\Delta^{(i)}\leq\frac{\Delta}{2^i}+2$.
\end{claim}

\begin{proof}
    We prove the claim by induction on $i$:\\
    For $i=0$, we have $G^{(0)}=G$, and indeed the maximum degree of $G$ satisfies $\Delta\leq\Delta+2$.\\
    For $i>0$, let $\Delta^{(i)}$ be the maximum degree of a graph $G^{(i)}$ on which the procedure is invoked on the $i$'th level of the recursion and let $\Delta^{(i-1)}$ be the maximum degree of the graph $G^{(i-1)}$, which we splitted in the previous level of recursion in order to create the graph $G^{(i)}$. By the induction hypothesis and by Equation (\ref{eq4.1}) we have
    $$\Delta^{(i)}\leq\frac{\Delta^{(i-1)}}{2}+1\leq\frac{\frac{\Delta}{2^{i-1}}+2}{2}+1=\frac{\Delta}{2^i}+2.$$
\end{proof}

We now show that the coloring $\varphi$ produced by Procedure \textsc{Aprrox-Edge-Coloring} is a proper edge-coloring, and we also analyse the number of colors it uses.

\begin{lemma}[Analysis of the coloring $\varphi$]\label{proper edge-coloring}
    Let $G=(V,E)$ be a graph with maximum degree $\Delta$.
    Procedure \textsc{Aprrox-Edge-Coloring} invoked on $G$ with parameter $h\leq\log\Delta$ computes a proper $(\Delta+3\cdot 2^h)$-edge-coloring $\varphi$ of $G$.
\end{lemma}

\begin{proof}
    We start by showing that the coloring $\varphi$ is proper by induction on $h$:\\
    For $h=0$, by Theorem \ref{Procedure Edge-Coloring properties} the coloring $\varphi$ is a proper $(\Delta+1)$-edge-coloring.\\
    For $h>0$, let $e_1\neq e_2\in E$. By the induction hypothesis, the colorings $\varphi_1$ and $\varphi_2$ that were defined on line 6 of Algorithm \ref{Procedure Edge-Coloring} are proper edge-colorings.\\
    If $e_1,e_2\in E_1$, then $\varphi(e_1)=\varphi_1(e_1)\neq\varphi_1(e_2)=\varphi(e_2)$.\\
    If $e_1,e_2\in E_2$, then $\varphi(e_1)=|\varphi_1|+\varphi_2(e_1)\neq|\varphi_1|+\varphi_2(e_2)=\varphi(e_2)$.\\
    Otherwise, without loss of generality, $e_1\in E_1$, $e_2\in E_2$, and $$\varphi(e_1)=\varphi_1(e_1)\leq|\varphi_1|<|\varphi_1|+\varphi_2(e_2)=\varphi(e_2).$$
    Hence $\varphi$ is a proper edge-coloring.\\
    Next, we analyse the number of colors used by Procedure \textsc{Aprrox-Edge-Coloring}.\\
    Recall that $G^{(i)}$ is a graph computed after $i$ levels of recursion, and its maximum degree is at most $\Delta^{(i)}$. 
    Denote by $f(i)$ the number of colors used by the algorithm on the graph $G^{(i)}$. We prove by induction on $i$ that $$f(i)\leq 2^{h-i}\cdot\left(\frac{\Delta}{2^h}+3\right).$$
    For $i=h$, by the assumption the algorithm uses $f(h)=\Delta^{(h)}+1$ colors, and by Claim \ref{deg bound}, $$f(h)=\Delta^{(h)}+1\leq\left(\frac{\Delta}{2^h}+2\right)+1=\frac{\Delta}{2^h}+3.$$
    For $i<h$, we merge two disjoint palettes of size $f(i+1)$ each. Hence, by the induction hypothesis we have $$f(i)= 2\cdot f(i+1)\leq 2\cdot2^{h-i-1}\cdot\left(\frac{\Delta}{2^h}+3\right)=2^{h-i}\cdot\left(\frac{\Delta}{2^h}+3\right).$$
    Therefore, the number of colors employed by the algorithm on the graph $G=G^{(0)}$ is at most $f(0)\leq 2^h\cdot\left(\frac{\Delta}{2^h}+3\right)=\Delta+3\cdot 2^h$.
\end{proof}

We next analyse the complexity of our algorithm. 
Note that the time complexities of all of the results in Theorem \ref{Procedure Edge-Coloring properties} are a function of $n$ and $\Delta$ (and $a$, that can be bounded by $\Delta$). Also, the number of processors that these algorithms require is a function of $n$ and $\Delta$ (and $a$, that can be bounded by $\Delta$), multiplied by $m$. Hence, to simplify the analysis of our $(1+\varepsilon)\Delta$-edge-coloring algorithm, we present a notation for these terms. For an $n$-vertex $m$-edge graph $G=(V,E)$ with maximum degree $\Delta$, we denote the time that is required for computing a $(\Delta+1)$-edge-coloring of $G$ by $ECT\left(n,\Delta\right)$, and denote the number of processors required for this process by $m\cdot ECP\left(n,\Delta\right)$. 
Note that for the algorithms from Theorem \ref{Procedure Edge-Coloring properties}, we have $ECT\left(O(n),O(\Delta)\right)=O\left(ECT\left(n,\Delta\right)\right)$ and $ECP\left(O(n),O(\Delta)\right)=O(ECP\left(n,\Delta\right))$, and that $ECT$ and $ECP$ are monotonic increasing. Also, note that $ECT\left(n,\Delta\right)=\Omega(\log\Delta\cdot\log n)$ and $ECP\left(n,\Delta\right)=\Omega(1)$.

\begin{lemma}[A bound on the time and number of processors]\label{PRAM time and processors bound}
    Let $G=(V,E)$ be an $n$-vertex $m$-edge graph with maximum degree $\Delta$. Procedure \textsc{Aprrox-Edge-Coloring} invoked on $G$ with parameter $h\leq\log\Delta$ requires $O\left(ECT\left(n,\frac{\Delta}{2^h}\right)\right)$ time,
    and it uses $O\left(m\cdot ECP\left(n,\frac{\Delta}{2^h}\right)\right)$ processors.
\end{lemma}

\begin{proof}
    First, observe that in each level $0\leq i< h$ of recursion we execute Procedure \textsc{Degree-Splitting} in parallel on all the $t_i\leq 2^i$ subgraphs in that level of recursion. Let $m^{(i)}_1,m^{(i)}_2,...,m^{(i)}_{t_i}$ be the number of edges in these subgraphs. Since these subgraphs are edge-disjoint, we have $\sum_{j=1}^{t_i}m^{(i)}_j=m$. Hence by Theorem \ref{Procedure Degree-Splitting}, each level of recursion except for the base case requires $O\left(\log n\right)$ time using $\sum_{j=1}^{t_i} O\left(m^{(i)}_j\right)=O(m)$ processors.
    By Theorem \ref{Procedure Edge-Coloring properties} and Claim \ref{deg bound}, the base case requires $O\left(ECT\left(n,\Delta^{(h)}\right)\right)=O\left(ECT\left(n,\frac{\Delta}{2^h}\right)\right)$
    time using $\sum_{i=1}^{t_h}O\left(m^{(h)}_i\cdot ECP\left(n,\Delta^{(h)}\right)\right)=O\left(m\cdot ECP\left(n,\frac{\Delta}{2^h}\right)\right)$ processors.
    Since there are $h$ levels of recursion (except the base case) and $h\leq \log\Delta$, we conclude that the total time of the algorithm is $O\left(ECT\left(n,\frac{\Delta}{2^h}\right)+h\cdot\log n\right)=O\left(ECT\left(n,\frac{\Delta}{2^h}\right)\right)$,
    and it uses $O\left(m\cdot ECP\left(n,\frac{\Delta}{2^h}\right)\right)$ processors.
\end{proof}

We are now ready to summarize the main properties of Procedure \textsc{Aprrox-Edge-Coloring}.
Lemma \ref{proper edge-coloring} and Lemma \ref{PRAM time and processors bound} imply the following corollary:

\begin{corollary}[Properties of Procedure \textsc{Approx-Edge-Coloring}]\label{pram properties}
    Let $G=(V,E)$ be a graph with maximum degree $\Delta$.
    Procedure \textsc{Aprrox-Edge-Coloring} invoked on $G$ with parameter $h\leq\log
    \Delta$ computes a proper $(\Delta+3\cdot 2^h)$-edge-coloring of $G$ in
    $O\left(ECT\left(n,\frac{\Delta}{2^h}\right)\right)$ time using $O\left(m\cdot ECP\left(n,\frac{\Delta}{2^h}\right)\right)$ processors.
\end{corollary}

To get simpler and more useful trade-offs, we set $h\approx\log\left(\varepsilon\cdot\Delta\right)$.

\begin{theorem}[A generic trade-off]\label{trade-off}
    Let $G=(V,E)$ be an $n$-vertex $m$-edge graph with maximum degree $\Delta$, and let $\frac{1}{\Delta}\leq\varepsilon<1$. Procedure \textsc{Aprrox-Edge-Coloring} invoked on $G$ with parameter $h=\max\left\{\left\lfloor\log\left(\frac{\varepsilon\cdot\Delta}{3}\right)\right\rfloor,0\right\}$ computes a proper $(1+\varepsilon)\Delta$-edge-coloring of $G$ in 
    $O\left(ECT\left(n,\varepsilon^{-1}\right)\right)$ time using $O\left(m\cdot ECP\left(n,\varepsilon^{-1}\right)\right)$ processors.
\end{theorem}

In the next theorem, we summarize possible trade-offs of our edge-coloring algorithm that can be obtained via various $(\Delta+1)$-edge-coloring algorithms from Theorem \ref{Procedure Edge-Coloring properties}.
\begin{theorem}[Specific trade-offs]\label{(1+eps)Delta Trade-off}
    Let $G=(V,E)$ be an $n$-vertex $m$-edge graph with maximum degree $\Delta$, and let $\frac{1}{\Delta}\leq\varepsilon<1$. In all trade-offs one sets $h=\max\left\{\left\lfloor\log\left(\frac{\varepsilon\cdot\Delta}{3}\right)\right\rfloor,0\right\}$, and as a result, the algorithm produces a proper $(1+\varepsilon)\Delta$-edge-coloring of $G$. Our algorithm requires:
    \begin{enumerate}
        \item[(1)] (Using Theorem \ref{Procedure Edge-Coloring properties}(1)) $O\left(\varepsilon^{-4}\cdot\log^4 n\right)$ time using $O\left(m\cdot\varepsilon^{-1}\right)$ processors.
        \item[(2)] (Using Theorem \ref{Procedure Edge-Coloring properties}(2)) $O\left(\varepsilon^{-4}\cdot\log^2 n+\varepsilon^{-6}\cdot\log^2\varepsilon^{-1}\cdot\log n\right)$ time using $O\left(m\cdot\varepsilon^{-1}\right)$ processors.
        \item[(3)] (Using Theorem \ref{Procedure Edge-Coloring properties}(4), when bounding the arboricity $a$ by $\Delta$) $O\left(\varepsilon^{-4-o(1)}\cdot\log^2 n\right)$ time using $O\left(m\cdot\varepsilon^{-1}\cdot\frac{\log^{\delta}\left(\varepsilon^{-1}\right)\cdot\log n}{\log(\varepsilon^{-1}\cdot\log n)}\right)$ processors, for any constant $\delta>0$.
        \item[(4)] (Using Theorem \ref{Procedure Edge-Coloring properties}(5)) $O\left(\varepsilon^{-5}\cdot\log n\cdot\left(\log\varepsilon^{-1}+\frac{\log n}{\log\varepsilon^{-1}\cdot\log(\varepsilon^{-1}\cdot\log n)}\right)+\varepsilon^{-4}\cdot\log^2 n\right)=O(\varepsilon^{-5}\cdot\log^2 n)$ time using  $O\left(m\cdot\left(\varepsilon^{-1}\cdot\log\varepsilon^{-1}+\frac{\varepsilon^{-1/2}\cdot\log n}{\log(\varepsilon^{-1}\cdot\log n)}\right)\right)$ processors.
    \end{enumerate}
\end{theorem}

For constant arbitrarily small $\varepsilon>0$, the result in Theorem \ref{(1+eps)Delta Trade-off}(2) provides 
a $(1+\varepsilon)\Delta$-edge-coloring in $O(\log^2 n)$ time using $O(m)$ processors. 
This is also our best trade-off when $\varepsilon$ is a slowly decreasing function of $n$, 
i.e., $\varepsilon=\frac{1}{\text{poly}(\log\log n)}$.
For $\varepsilon\leq\frac{1}{\log^c n}$, for a constant $c>0$, other trade-offs kick in.

\section{The Edge-Coloring Update Algorithm}\label{Sections: The Edge-Coloring Update Algorithm}
In this section we provide an algorithm for the edge-coloring update problem.
In the edge-coloring update problem, we are given a graph $G=(V,E)$ with an upper bound $\Delta_{\max}$ on its maximum degree, together with a proper $(\Delta_{\max}+1)$-edge-coloring of $G$. During the algorithm, the graph undergoes a sequence of edge insertions and deletions, while maintaining its maximum degree at most $\Delta_{\max}$. We aim to maintain a proper $(\Delta_{\max}+1)$-edge-coloring of the graph throughout these updates.
Specifically, given a new edge $e=(u,v)\notin E$ or an existing edge $e'=(u',v')\in E$, our goal is to maintain a proper $(\Delta_{\max}+1)$-edge-coloring of the graph $G\cup\{e\}$ or $G\setminus\{e'\}$, respectively.\\
\\\textbf{Dealing with an Edge Insertion.}
Given a new edge $e=(u,v)\notin E$ we first construct a maximal fan, with an arbitrarily chosen center among $u$ and $v$ using Lemma \ref{Fan construction} from Section \ref{pram edge coloring}. Next, we compute its $\alpha\beta$-path by applying the connected components algorithm due to~\cite{shiloach1982logn} on the graph $G_{\alpha,\beta}$, which is induced by the edges that are colored with the colors $\alpha$ and $\beta$. Finally, we apply Procedure \textsc{Recolor-Fan} from Section \ref{pram edge coloring} on the computed fan and the $\alpha\beta$-path to produce a proper $(\Delta_{\max}+1)$-edge-coloring.

\begin{lemma}[Complexity of an insertion]
    Let $G=(V,E)$ be a graph with maximum degree at most $\Delta_{\max}$ and let $\varphi$ be a proper $(\Delta_{max}+1)$-edge-coloring of $G$. Let $e=(u,v)\notin E$ be a new edge such that the maximum degree of the graph $G^+=G\cup \{(u,v)\}$ is at most $\Delta_{\max}$. A proper edge-coloring of $G^+$ can be computed in $O(\log n)$ time using $O(n)$ processors.
\end{lemma}

\begin{proof}
    By Lemma \ref{Fan construction} from Section \ref{pram edge coloring}, the construction of a maximal fan $f$, centered at $v$, with the uncolored edge $(u,v)$, requires $O(\log\Delta_{\max})$ time, using $O(\deg(v))$ processors. 
    To compute the $\alpha\beta$-path associated with $f$, we assign to each vertex the processors that might be required by its incident edges in $G_{\alpha,\beta}$ (the subgraph induced by the edges colored with the colors $\alpha$ and $\beta$). Using these processors, we compute the connected component of $v$ in the graph $G_{\alpha,\beta}$. Note that this graph is of size $O(n)$. Therefore, by Lemma \ref{Connected components algorithm}, this computation requires $O(\log n)$ time, using $O(n)$ processors. Finally, we recolor the fan and its $\alpha\beta$-path using Procedure \textsc{Recolor-Fan}. By Theorem \ref{Procedure Recolor-Fan}, this task requires $O(\log n)$ time, using $O(n)$ processors. Overall, we conclude that a proper edge-coloring of $G^+$ can be computed in $O(\log n)$ time using $O(n)$ processors.
\end{proof}

Note that deletion of an edge in the graph does not require any further updates. We summarize this in the following theorem.

\begin{theorem}
    Let $G=(V,E)$ be a graph with maximum degree at most $\Delta_{\max}$ and let $\varphi$ be a proper $(\Delta_{max}+1)$-edge-coloring of $G$. Let $e\notin E$ be a new edge such that the maximum degree of the graph $G^+=G\cup \{e\}$ is at most $\Delta_{\max}$ and let $e'\in E$. A proper edge-coloring of $G^+$ can be computed in $O(\log n)$ time using $O(n)$ processors and a proper edge-coloring of $G^-=G\setminus \{e'\}$ can be computed in $O(1)$ time using $O(1)$ processors.
\end{theorem}

As a corollary, we improve the result of~\cite{liang1996parallel} for the edge-coloring update problem, that achieved an $O(K\cdot\sqrt{\Delta}\cdot\log ^3 \Delta+K\cdot\log n)$ time algorithm, using $O(n\cdot\Delta+\Delta^3)$ processors, for insertion of $K$ new edges.

\begin{corollary}
    Let $G=(V,E)$ be a graph with maximum degree at most $\Delta_{\max}$ and let $\varphi$ be a proper $(\Delta_{max}+1)$-edge-coloring of $G$. Let $G^{+K}$ be a graph obtained by the graph $G$ with a collection of $K$ new edges such that the maximum degree of $G^{+K}$ is at most $\Delta_{max}$. A proper $(\Delta_{\max}+1)$-edge-coloring of $G^{+K}$ can be computed in $O(K\cdot\log n)$ time, using $O(n)$ processors.
\end{corollary}

\section{Acknowledgements}
We are grateful to anonymous reviewers of SODA 2026 for helpful detailed remarks.

\bibliographystyle{alpha}
\bibliography{bibliography.bib}
\newpage
\appendix

\section{Computing a Large Collection of Pairwise-Disjoint Fans}\label{app: Computing a Large Collection of Pairwise-Disjoint Fans}

In this appendix we explain the flaw in the algorithm of Liang et al.~\cite{liang1996parallel}.
Let $G=(V,E)$ be a graph with maximum degree $\Delta$, equipped with a proper partial edge-coloring, and let $F$ be a fixed subset of the set of uncolored edges in $G$, and suppose that $F$ is a matching.
At some point, the algorithm needs to find a large collection of uncolored edges that can be recolored in parallel. That is, a collection of edges, such that fans centered at their endpoints are not intersecting with one another. In~\cite{liang1996parallel}, the authors used the graph $G_3$ for this computation. Recall that the \textit{distance} between a pair of edges $e_1=(v_1,u_1)\in E$ and $e_2=(v_2,u_2)\in E$ in $G$ is defined by $$\text{dist}_G(e_1,e_2)=\min\left\{\text{dist}_G(u_1, u_2),\text{dist}_G(u_1, v_2),\text{dist}_G(v_1, u_2),\text{dist}_G(v_1, v_2)\right\}.$$
The graph $G_3=(V_3,E_3)$ is defined~\cite{liang1996parallel} by $V_3=F$, and for two edges $e_1=(v_1,u_1)\in F$ and $e_2=(v_2,u_2)\in F$, the edge $(e_1,e_2)$ is in $E_3$ if and only if one of the edges in $\{(u_1,u_2),(u_1,v_2),(v_1,u_2),(v_1,v_2)\}$ is in $E$. In other words, $(e_1,e_2)\in E_3$ if and only if $\text{dist}_G(e_1,e_2)=1$ (since $F$ is a matching, $\text{dist}_G(e_1,e_2)\neq 0$). Liang et al.~\cite{liang1996parallel} compute a maximal independent set in $G_3$ in order to compute the collection of uncolored edges that can be recolored in parallel.~\cite{liang1996parallel} did not provide a proof that an independent set in $G_3$ satisfies this property. We next argue that this is not the case. Specifically, consider two edges $e_1=(v,u_1)$, $e_2=(v^*,u^*_1)$ such that $dist_G(e_1,e_2)=2$. Let $f(v)=\langle v,u_1,...,u_k\rangle$ and $f(v^*)=\langle v^*,u^*_1,...,u^*_l\rangle$, for some $k,l\geq 4$ be the respective fans for $v$ and $v^*$. Suppose also that they have a common vertex, i.e., there exists a pair of indexes $i,j$, 
$2\leq i\leq k-2$, $2\leq j\leq l-2$, such that $u_i=u_j^*$. Moreover, we assume that neither among the vertices $u_i,u_{i+1},u_j^*,u^*_{j+1}$ is a special vertex in its respective fan, i.e., $\{u_i,u_{i+1},u_j^*,{u^*_{j+1}\}\cap\{x(v),y(v),z(v),x(v^*),y(v^*),z(v^*)}\}=\emptyset$. Finally, let $\gamma$ be the missing color of $u_i=u_j^*$ in both fans $f(v)$ and $f(v^*)$. It follows that after the exchanging of the $\alpha\beta$-paths in the two fans and respective rotations of these fans, the incident edges $(v,u_i)$, $(u_j^*,v^*)$ end up being colored by $\gamma$, i.e., the resulting edge-coloring is no longer proper. See Figure \ref{fig:intersection fans2} for an illustration. Note also that it might happen that for any endpoint $w$ of $(v,u_1)$ and $w^*$ of $(v^*,u^*_1)$, the two vertices are at distance two from one another. See Figure \ref{fig:distance 2 example} for an illustration.

\begin{figure}
    \centering
    \begin{tikzpicture}
\begin{scope}[yshift=5cm]
        \begin{scope}[xshift=-4cm]
        \node[circle] at ({180}:0.3cm)  {$v$};
        \node[circle] at ({93}:2.35cm)  {$u_1$};
        \node[circle] at ({72}:3.85cm)  {$u_2=u^*_3$};
        \node[circle] at ({50}:2.35cm)  {$u_3$};
        \node[circle] at ({30}:2.35cm)  {$u_4$};
        \node[circle] at ({10}:2.35cm)  {$u_5$};
        \node[circle,fill=pink] at (360:0mm) (center) {};
        \node[circle,fill=green] at ({90}:2cm) (n1) {};
        \node[circle,fill=blue] at ({70}:3.55cm) (n2) {};
        \node[circle,fill=purple] at ({50}:2cm) (n3) {};
        \node[circle,fill=red] at ({30}:2cm) (n4) {};
        \node[circle,fill=purple] at ({10}:2cm) (n5) {};
    
        \draw[dotted, line width=0.7mm] (center)--(n1);
        \draw[line width=0.7mm, green] (center)--(n2);
        \draw[line width=0.7mm, blue] (center)--(n3);
        \draw[line width=0.7mm, purple] (center)--(n4);
        \draw[line width=0.7mm, red] (center)--(n5);

        \node[circle,fill=black] at (3.3,0) (n6) {};
        \node[circle,fill=black] at (3.3,-2) (n7) {};
        \node[circle,fill=black] at (1.8,-3) (n8) {};
        \node[circle,fill=black] at (0,-3) (n9) {};

        \draw[line width=0.7mm, pink] (n4)--(n6);
        \draw[line width=0.7mm, purple] (n6)--(n7);
        \draw[line width=0.7mm, pink] (n7)--(n8);
        \draw[line width=0.7mm, purple] (n8)--(n9);

        \end{scope}

        \begin{scope}[xshift =3cm, yscale=1,xscale=-1]
        \node[circle] at ({180}:0.35cm)  {$v^*$};
        \node[circle] at ({70}:2.42cm)  {$u^*_1$};
        \node[circle] at ({50}:2.42cm)  {$u^*_2$};
        \node[circle] at ({10}:2.42cm)  {$u^*_4$};
        \node[circle] at ({350}:2.4cm)  {$u^*_5$};
        \node[circle] at ({330}:2.4cm)  {$u^*_6$};
        \node[circle,fill=pink] at (360:0mm) (center) {};
        \node[circle,fill=green] at ({70}:2cm) (n2) {};
        \node[circle,fill=red] at ({50}:2cm) (n3) {};
        \node[circle,fill=blue] at ({30}:6.7cm) (n4) {};
        \node[circle,fill=purple] at ({10}:2cm) (n5) {};
        \node[circle,fill=brown] at ({350}:2cm) (n6) {};
        \node[circle,fill=purple] at ({330}:2cm) (n7) {};
    
        \draw[dotted, line width=0.7mm] (center)--(n2);
        \draw[line width=0.7mm, green] (center)--(n3);
        \draw[line width=0.7mm, red] (center)--(n4);
        \draw[line width=0.7mm, blue] (center)--(n5);
        \draw[line width=0.7mm, purple] (center)--(n6);
        \draw[line width=0.7mm, brown] (center)--(n7);

        \node[circle,fill=black] at (3,-2) (n8) {};
        \node[circle,fill=black] at (1.6,-3) (n9) {};
        \node[circle,fill=black] at (-0.2,-3) (n10) {};
        \node[circle,fill=black] at (-2.2,-3) (n11) {};

        \draw[line width=0.7mm, pink] (n6)--(n8);
        \draw[line width=0.7mm, purple] (n8)--(n9);
        \draw[line width=0.7mm, pink] (n9)--(n10);
        \draw[line width=0.7mm, purple] (n10)--(n11);

        \end{scope}
    \end{scope}
    \end{tikzpicture}

    \caption{Two intersecting fans, characterized by $(\alpha,\beta)=(\color{pink}\bullet\color{black}, \color{purple}\bullet\color{black})$ $\langle u_1,u_2,...,u_5\rangle$, centered at $v$, and $\langle u^*_1,u^*_2,...,u^*_6\rangle$ centered at $v^*$, in which the distance between their respective uncolored edges is 2.}
    \label{fig:intersection fans2}
\end{figure}
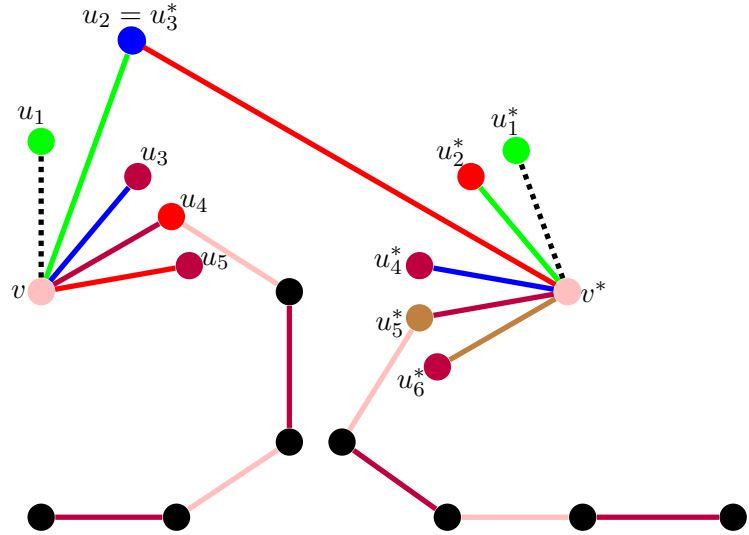

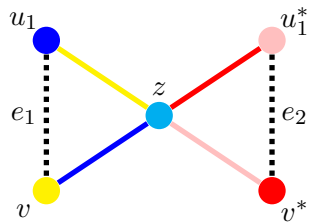
\begin{figure}
    \centering
    \begin{tikzpicture}

        \node[circle] at (0,0.35)  {$z$};
        \node[circle] at (-1.8,0)  {$e_1$};
        \node[circle] at (1.8,0)  {$e_2$};

        \node[circle] at ({145}:2.2cm)  {$u_1$};
        \node[circle] at ({215}:2.2cm)  {$v$};
        \node[circle] at ({35}:2.2cm)  {$u^*_1$};
        \node[circle] at ({325}:2.2cm)  {$v^*$};

        \node[circle,fill=cyan] at (0,0) (center1) {};
        \node[circle,fill=blue] at (-1.5,1) (n1) {};
        \node[circle,fill=yellow] at (-1.5,-1) (n2) {};
        \node[circle,fill=pink] at (1.5,1) (n3) {};
        \node[circle,fill=red] at (1.5,-1) (n4) {};
        
        \draw[line width=0.7mm, yellow] (center1)--(n1);
        \draw[line width=0.7mm, blue] (center1)--(n2);
        \draw[line width=0.7mm, red] (center1)--(n3);
        \draw[line width=0.7mm, pink] (center1)--(n4);
        \draw[dotted, line width=0.7mm] (n1)--(n2);
        \draw[dotted, line width=0.7mm] (n3)--(n4);
        
    \end{tikzpicture}

    \caption{Two uncolored edges of $G$ that might be chosen in an independent set of $G_3$, although maximal fans containing them may intersect at $z$.}
    \label{fig:distance 2 example}
\end{figure}

To overcome this problem, we define the graph $G^{(F)}$ that includes also edges between pairs of edges from $F$ with distance 2 (in $G$), i.e., $G^{(F)}=\left(F,E^{(F)}\right)$, where for $e_1\neq e_2\in F$, an edge $(e_1,e_2)\in E^{(F)}$ if $\text{dist}_G(e_1,e_2)\leq 2$. See Section \ref{sec: Parallel Fan-Recoloring} for all the details.\\

The main impact of this change on our algorithm, is that the maximum degree of $G^{(F)}$ is now bounded by $O(\Delta^2)$ instead of $O(\Delta)$ as in $G_3$ (in~\cite{liang1996parallel}). This larger maximum degree may imply a smaller (by a factor of $\Delta$) "large" independent set in $G^{(F)}$ (than in $G_3$) by a factor of $\Delta$, which increases the number of iterations of the algorithm.

\section{Some Proofs from Section \ref{pram edge coloring}}\label{App: Some Proofs from Section 3}
In this appendix we provide some proofs that were omitted from Section \ref{pram edge coloring}.

\begin{proof}[proof of Lemma \ref{G_3 property}]
    First, observe that for every fan $\langle u_1,...,u_k\rangle$ of $v\in V$, and for each $i\in \{1,2,...,k\}$, $\text{dist}_G(v,u_i)=1$. Let $e=(v,u_1)\in F$ and $e^*=(v^*,u^*_1)\in F$ such that $(e,e^*)\notin E^{(F)}$. By the definition of $G^{(F)}$, $\text{dist}_G(e, e^*)\geq 3$, i.e., in particular, $\text{dist}_G(v, v^*)\geq 3$. Without loss of generality, assume towards contradiction that there are intersecting fans $\langle u_1,u_2,...,u_k\rangle$ centered at $v$, and $\langle u^*_1,u^*_2,...,u^*_l\rangle$ centered at $v^*$. Let $x$ a common vertex of both fans. Then by the triangle inequality, $\text{dist}_G(v,v^*)\leq \text{dist}_G(v,x)+\text{dist}_G(x,v^*)\leq 2$, which is a contradiction. Hence any fan with an uncolored edge $e$ and any fan with an uncolored edge $e^*$ are disjoint.
\end{proof}

\section{Edge-Coloring Paths and Cycles}\label{App: coloring}
In this appendix we present an algorithm for computing an alternating edge-coloring of a path. In addition, we devise another edge-coloring algorithm that we use in our main edge-coloring algorithms in Sections \ref{pram edge coloring} and \ref{Our alg}.\\

Let $P=(V,E)$ be a not necessarily simple path. We also denote this path by $P=(v_0,e_1,v_1,e_2,v_2,$
$...,v_{m-1},e_m,v_m)$. The alternating edge-coloring algorithm proceeds in $h=\lfloor\log m\rfloor$ ($m=|E|$) iterations. In the $j$'th iteration, for each $v_i$ in parallel, we compute the vertices that are at distance $2^j$ from $v_i$ on this path, using the vertices that are at distance $2^{j-1}$ from it, computed in the previous iteration. Using this information, each vertex computes its distance from the endpoint $v_0$ of this path, and according to the distance (modulo 2) of each edge from the vertex $v_0$, its color is determined. The full description of the algorithm is given in Procedure \textsc{Alternating-Coloring}.\\
\\$\textsc{Procedure Alternating-Coloring }(P=(V,E)=(v_0,e_1,v_1,e_2,v_2,...,v_{m-1},e_m,v_m))$
    \begin{description}
        \item[\textbf{Step 1.}] \textbf{For} each $i\in \{1,2,...,m\}$ in parallel \textbf{do}\\
        \makebox[1.2cm]{}\textbf{For} each $1\leq j\leq h$ \textbf{do}\\
        \makebox[1.7cm]{}- compute the (at most 2) vertices of distance $2^j$ from $v_i$ on this path\\ \makebox[1.9cm]{}(if exist) using the vertices of distance $2^{j-1}$ from $v_i$, that were computed\\
        \makebox[1.9cm]{}in the previous iteration
        \item[\textbf{Step 2.}] \textbf{For} each $i\in\{1,2,...,m\}$ in parallel \textbf{do}\\
        \makebox[1.2cm]{}- compute the distance $d(v_0,v_i)$ of $v_i$ from the endpoint $v_0$ of the path using\\
        \makebox[1.4cm]{}the distances that we computed on step 1
        \item[\textbf{Step 3.}] \textbf{For} each $e=(v,u)\in E$ in parallel \textbf{do}\\
        \makebox[1.2cm]{}- color $e$ with the color $\left(\min\{d(v_0,v),d(v_0,u)\}\mod 2\right)$
    \end{description}

We now analyse the algorithm.

\altPath*

\begin{proof}
    First, observe that the algorithm indeed properly edge-colors the graph using two colors. In addition, in step 1 of the algorithm, for each $v_i\in V$ we compute the vertices on this path of distance $2^j$ from $v$ in parallel (for each $1\leq j\leq h$). This process proceeds in $O(\log n)$ rounds, and each round requires $O(1)$ time. Hence, step 1 requires $O(\log n)$ time using $O(m)$ processors. For step 2 we assign a processor to each occurrence of each $v\in V$ that computes the distance of $v$ from $v_0$ in $O(\log n)$ time. Hence step 2 requires $O(\log n)$ time using $O(m)$ processors. Finally, for step 3, we assign a processor to each edge in the graph, that colors the edge in $O(1)$ time. Hence Procedure \textsc{Alternating-Coloring} requires $O(\log n)$ time using $O(m)$ processors.
\end{proof}

For an odd length cycle, essentially the same procedure computes a 2-edge-coloring in which all colors are alternating, except for two certain consecutive edges that are colored by the same color.

\begin{restatable}[Alternating coloring of odd length cycle]{corollary}{cycles}
\label{cycles}
    Let $C=(V,E)$ be an $n$-vertex (not necessarily simple) odd length cycle with $m$ edges.
    There is an algorithm that edge-colors $C$ using two colors in $O(\log n)$ time using $O(m)$ processors, such that the coloring is alternating, except for two consecutive edges, colored by the same color.
\end{restatable}

For the last theorem of this appendix, we recall the algorithm due to Shiloach and Vishkin for computing connected components in a simple graph:
\connComp*
We now derive the following theorem.

\threeEdgeColoring*

\begin{proof}
    We first compute non-isolated connected components of the graph using Lemma \ref{Connected components algorithm} in $O(\log n)$ time using $O(m)$ processors. Observe that $G$ is composed of simple paths and cycles. For each cycle we choose an arbitrary edge and delete it from the graph. Denote the deleted set of edges by $D$. Observe that if $\Delta=1$, then $D=\emptyset$. In this way we are left with only connected components that form simple paths. By Lemma \ref{path coloring}, we can edge-color them (properly) in parallel using 2 alternating colors in $O(\log n)$ time using $O(m)$. Finally, we color the deleted edges in $D$ using the color 3, and obtain a proper $(\Delta+1)$-edge-coloring of $G$.
\end{proof}

\newpage
\section{Maximal Path}\label{maximal path}

In this appendix we present an algorithm for computing a maximal path in a directed graph $G=(V,E)$ with out-degree at most 1. This algorithm requires $O(\log n)$ time using $O(n)$ processors. We use this algorithm for building fans in the $(\Delta+1)$-edge-coloring algorithm in Section \ref{pram edge coloring}.\\

We start with some notation that we will use in the algorithm.
Let $h=\lceil\log n\rceil$.
For a vertex $v\in V$ with out-degree 1, we define $n(v)$ to be its only outgoing neighbor, and otherwise, when its out-degree is 0, we define $n(v)=\perp$. For $v\in V$, we define $P(v)$ to be the path starting at $v_0=v$, and for $i\geq 1$, the $i$'th vertex $v_i$ in this path is the outgoing neighbor of the previous vertex of this path, i.e., $v_i=n(v_{i-1})$ (this path might be infinite). Observe that all the paths starting at $v$ are contained in $P(v)$, and there is a unique maximal simple path starting at $v$. We call this path \emph{the maximal path of} $v$. For a vertex $u\in V$, and for every $v\in V$, we will refer to the position of the first occurrence of $v$ in $P(u)$ as the \emph{index} of $v$ in $P(u)$ (if $v$ does not appear in $P(u)$, the index of $v$ is $-1$). For $v\in V$ and $k\geq 0$, we denote by $P(v)\vert_k$ the restriction of $P(v)=(v_0,v_1,...)$ to length $k$, i.e. $P(v)\vert_k=(v_0,v_1,...,v_k)$ (if $|P(v)|\leq k$, then $P(v)\vert_k=P(v)$). Observe that if $P(v)$ is finite, then $P(v)$ is a simple path that cannot be extended (if we visit a vertex more than once, then we enter into a cycle, and $P(v)$ would be infinite). Hence, in this case $P(v)$ is the maximal path of $v$. Otherwise, $P(v)$ is infinite, and in order to find the maximal path of $v$, we would like to find the first edge that encounters a vertex that was already visited in $P(v)$ before, and restrict $P(v)$ up to this edge. To this end, we will find the vertex with the greatest index $k$ in the graph, and restrict $P(v)$ to length $k$ (the restricted path is the maximal path of $v$).\\

Given a directed graph $G=(V,E)$ with maximum out-degree 1 and a designated root vertex $r$, on step 1 the algorithm computes the path $P(r)\vert_{2^h}$ (since a maximal path is of length at most $n\leq 2^h$, then $P(r)\vert_{2^h}$ must contain the maximal path of $r$). If the out-degree of the last vertex of this path is 0, then this path must be the maximal path of $r$. Otherwise, this path is not simple, and the algorithm will restrict this path (on step 2) to the maximal path of $r$ (by searching for the appropriate index for the restriction).\\
\\$\textsc{Procedure Maximal-Path}\,(G=(V,E), r)$
    \begin{description}
        \item[\textbf{Step 1.}] \textbf{For} each $v\in V$ in parallel \textbf{do}\\
        \makebox[1.2cm]{}- $n_0(v)\leftarrow\begin{cases}
			v, & \text{$n(v)=\perp$}\\
            n(v), & \text{otherwise}
		 \end{cases}$\\
        \makebox[1.2cm]{}- $P_v\leftarrow\begin{cases}
			\langle v\rangle, & \text{$n(v)= \perp$}\\
            \langle v,n(v)\rangle, & \text{otherwise}
		 \end{cases}$\\
        \makebox[1.2cm]{}\textbf{For} each $1\leq i\leq h$ \textbf{do}\Comment{computing the vertices of distance $2^i$ from $v$}\\
        \makebox[1.7cm]{}- $P_v\leftarrow P_v\circ P_{n_{i-1}(v)}$\Comment{$\circ$ is a concatenation of two paths}\\
        \makebox[1.7cm]{}- $n_i(v)\leftarrow n_{i-1}(n_{i-1}(v))$
        \item[\textbf{Step 2.}] \textbf{If} $\deg_{\text{out}}(n_h(r))=0$ \textbf{then}\Comment{$P_r$ is finite, hence it is the maximal path of $r$}\\
        \makebox[1.2cm]{}\textbf{return} $P_r\\$
        \makebox[0.7cm]{}\textbf{else}\Comment{searching for the index of the restriction}\\
        \makebox[1.2cm]{}Allocate an empty array $A$ of size $n$.\\
        \makebox[1.2cm]{}\textbf{For} each $1\leq i\leq n$ in parallel \textbf{do}\\
        \makebox[1.7cm]{}- Compute the $i$'th vertex $v$ of the path $P_r$ using the pointers $n_j(u)$ that\\
        \makebox[1.95cm]{}were computed in step 1 for every $u\in V$ and $0\leq j\leq h$.\\
        \makebox[1.7cm]{}- Save the pair $(v,i)$ in the $i$'th slot $A[i]$ of the array $A$.\\
        \makebox[1.2cm]{}\textbf{For} each $v\in V$ in parallel \textbf{do}\\
        \makebox[1.7cm]{}- Sort the array according to the vertices $v$, and for every vertex $v$ store \\
        \makebox[1.95cm]{}just the smallest index $i$ such that $(v,i)$ is stored in $A$.\\
        \makebox[1.2cm]{}- Compute the maximal index $k$ of pairs left in the array $A$.\\
        \makebox[1.2cm]{}\textbf{return} $P_r\vert_k$
    \end{description}

We start the analysis by showing properties of the pointers $n_i(v)$ and the path $P_v$ for each $v\in V$ and $0\leq i\leq h$.

\begin{claim}\label{maximal path properties}
    Let $G=(V,E)$ be a directed graph with maximum out-degree at most 1 and let $r\in V$. Then for each $v\in V$ and $0\leq i\leq h$, the path $P_v$ computed after the $i$'th iteration of the inner loop in step 1 is the path $P(v)\vert_{2^i}$, and $n_i(v)$ is the last vertex in this path.
\end{claim}

\begin{proof}
    For each $v\in V$ and $0\leq i\leq h$, denote by $P^{(i)}_v$ the path $P_v$ computed after the $i$'th iteration of the inner loop in step 1.
    We will show by induction on $i$ that for each $v\in V$, $P^{(i)}_v=P(v)\vert_{2^i}$, and $n_i(v)$ is the last vertex of this path.\\
    For $i=0$, $P^{(0)}_v=P(v)\vert_1$ and $n_0(v)$ is indeed the last vertex of this path.\\
    For $i>0$, by the induction hypothesis, $P^{(i-1)}_v=P(v)\vert_{2^{i-1}}$, and $n_{i-1}(v)$ is the last vertex of this path, and $P^{(i-1)}_{n_{i-1}(v)}=P(n_{i-1}(v))\vert_{2^{i-1}}$ is the path from $n_{i-1}(v)$ to $n_i(v)=n_{i-1}(n_{i-1}(v))$. Hence $P_v^{(i)}=P_v^{(i-1)}\circ P_{n_{i-1}(v)}^{(i-1)}=P(v)\vert_{2^{i-1}}\circ P(n_{i-1}(v))\vert_{2^{i-1}}=P(v)\vert_{2^i}$ and $n_i(v)$ is the last vertex in this path.
\end{proof}

We are now ready to prove the correctness of the algorithm and analyse its complexity.
\maxPath*

\begin{proof}
    We first show that the algorithm indeed returns the maximal path of $r$.
    By Claim \ref{maximal path properties}, after step 1 of the algorithm, $P_r=P(r)\vert_{2^h}$, and $n_h(r)$ is the last vertex in this path. Recall that $P(r)\vert_{2^h}$ is either the maximal path of $r$ (if the path $P(r)$ is finite) or contains the maximal path of $r$. Hence, if $\deg_{\text{out}}(n_h(r))=0$, then $P_r=P(r)\vert_{2^h}=P(r)$ is  finite and it is the maximal path of $r$. Otherwise, in the for loops in step 2, for each vertex we compute its index in $P(r)$ and restrict the path $P_r$ to the maximal index in the graph. Observe that the restricted path is indeed the maximal path of $r$. This is because the index of the neighbor $u$ of the last vertex $v$ in the restricted path is smaller than the index of the last vertex of the restricted path, i.e., the vertex $u$ already appeared in the path.
    
    We now analyse the complexity of the algorithm. Step 1 of the algorithm assigns a processor to each vertex of the graph. It consists of $O(\log n)$ iterations, and each iteration requires $O(1)$ time. Hence this part requires $O(\log n)$ time using $O(n)$ processors.
    The first for loop in step 2 assigns a processor for each index $i\in\{1,2,...,n\}$, that computes the $i$'th vertex in $P_r$ (if exists) using the pointers $n_i(v)$ we computed in step 1. This process requires $O(\log n)$ time using $O(n)$ processors. 
    Computing the minimal index of each vertex $v\in V$ can be done in parallel in $O(\log n)$ time using $O(n)$ processors (by sorting the array $A$).
    Finally, finding the vertex with the maximal index requires $O(\log n)$ time using $O(n)$ processors as well.
    Hence the algorithm computes a maximal path starting at $r$ in $O(\log n)$ time using $O(n)$ processors.
\end{proof}

\section{Vertex-Coloring 
Algorithms}\label{Vertex-Coloring Algorithm}
In this appendix we adapt some known distributed vertex coloring algorithms to the $\mathrm{CRCW\,\,PRAM}$ model. We start by defining the distributed $\mathsf{LOCAL}$ model of computation.

\textbf{Distributed Synchronous Message Passing Model.}
In the $\mathsf{LOCAL}$ model of distributed computation, we are given an $n$-vertex graph $G$, which abstracts a communication network where each vertex plays the role of a processor, and edges represent communication links. Every vertex is given a unique $\Theta(\log n)$-bit identifier. Initially, each vertex knows its own identifier, as well as $n$ (the number of vertices) and perhaps some other global parameters, such as the maximum degree $\Delta$ of $G$. The computation proceeds in rounds. During each round, the vertices first perform arbitrary local computations and then synchronously broadcast messages to all their neighbors. At the end, each vertex should output its part of the global solution (for instance, its own color or, in the context of edge-coloring, the colors of the edges incident to it). The only measure of efficiency for such an algorithm is the worst-case number of communication rounds.

The first algorithm that we adapt combines a forests decomposition technique together with 3-vertex-coloring of forests.
As a first step towards the $(\Delta+1)$-vertex-coloring, we describe a 3-vertex-coloring algorithm due to Cole and Vishkin~\cite{cole1986deterministic} of \emph{oriented forests}.

\begin{definition}[Oriented Forest]
    A directed graph $F=(V,E)$ is an \emph{oriented forest} if the out-degree of each $v\in V$ is at most 1, and the underlying undirected graph of $F$ is acyclic. (An underlying undirected graph of a directed graph $F=(V,E)$ is a graph $G=(V,E')$, where $E'=\{(u,v)\mid \langle u,v\rangle\in E\}$.)
    \end{definition}

\subsection{3-Vertex-Coloring of Oriented Forests}\label{sec: 3-Vertex-Coloring of Oriented Forests}
In this subsection we present the algorithm due to~\cite{cole1986deterministic} for 3-vertex-coloring oriented forests. As a first step, we present a 6-vertex-coloring algorithm for oriented forests, and then show how it can be further improved into a 3-vertex-coloring. The 6-vertex-coloring algorithm proceeds in phases. The idea of the algorithm is to start with some proper vertex-coloring, that might use a large palette, and in each phase reduce the number of colors used by the constructed vertex-coloring. The coloring stays proper after every phase. 
We start by presenting the "color-reducing" algorithm, called Procedure $\textsc{Forest-Reduce-Colors}$, that is performed in each phase of the 6-vertex-coloring algorithm. In this algorithm each vertex in the graph, in parallel, defines its new color. In order to keep the coloring proper, each vertex will have the property that its new color is different from the color of its parent in the forest.\\
Before we provide the full description of the algorithm, we present some notation that we will use in the algorithm.
Let $F=(V,E)$ be an oriented forest, and $\varphi$ be a proper vertex-coloring of $F$. For a vertex $v\in V$, we denote by $p(v)$ the successor (parent) of $v$ in the input forest $F$. For a color (integer) $c$, denote by $c_i$ the $i$'th bit in the binary representation of $c$. Note that for a root $v$, of a tree connected component in $F$, we have $p(v)=\emptyset$. We now present Procedure \textsc{Forest-Reduce-Colors}. See Figure \ref{fig:procedure forest reduce colors}, for an illustration of the algorithm.\\
\\$\textsc{Procedure}$ $\textsc{Forest-Reduce-Colors}$ $\left(F=(V,E), \varphi\right)$
\begin{description}
    \item{\textbf{Step 1.}} \textbf{For} each $v\in V$ in parallel \textbf{do}\\
    \makebox[1.15cm]{}- Find the smallest index $i$, such that $\varphi(v)_i\neq \varphi(p(v))_i$.\\
    \makebox[1.15cm]{}\Comment{if $p(v)=\emptyset$, define $i=0$}\\
    \makebox[1.15cm]{}- Define $\varphi(v)\leftarrow (i,\varphi(v)_i)$. \Comment{$\varphi(v)\leftarrow 2\cdot i+\varphi(v)_i$}
\end{description}

\begin{figure}
    \centering
    \begin{tikzpicture}

        \begin{scope}[xshift=-4cm]
        \node[circle,fill=green] at (0,0) (center) {};
        \node at (0,0) {$v$};
        \node[circle,fill=red] at (0,1.2) (n1) {};
        \node[circle,fill=white,draw=black] at (1,-1.2) (n2) {};
        \node[circle,fill=white,draw=black] at (-1,-1.2) (n3) {};
        \node[circle,fill=white,draw=black] at (0,-1.2) (n4) {};
        \node[circle,fill=white,draw=black] at (-2,0) (n5) {};
        \node[circle,fill=blue] at (0,2.4) (n7) {};

        \node at (1.5,0)  {$\varphi(v)=100110_2$};
        \node at (1.8,1.2)  {$\varphi(p(v))=111110_2$};
        \node at (2.04,2.4)  {$\varphi(p(p(v)))=100100_2$};

        \node[circle] at (1.56,-0.2)  {\tiny 5};
        \node[circle] at (1.74,-0.2)  {\tiny 4};
        \node[circle] at (1.94,-0.2)  {\tiny 3};
        \node[circle] at (2.13,-0.2)  {\tiny 2};
        \node[circle] at (2.33,-0.2)  {\tiny 1};
        \node[circle] at (2.52,-0.2)  {\tiny 0};
        
        \node[circle] at (2.1,1)  {\tiny 5};
        \node[circle] at (2.28,1)  {\tiny 4};
        \node[circle] at (2.48,1)  {\tiny 3};
        \node[circle] at (2.68,1)  {\tiny 2};
        \node[circle] at (2.87,1)  {\tiny 1};
        \node[circle] at (3.05,1)  {\tiny 0};

        \node[circle] at (2.62,2.2)  {\tiny 5};
        \node[circle] at (2.79,2.2)  {\tiny 4};
        \node[circle] at (2.99,2.2)  {\tiny 3};
        \node[circle] at (3.17,2.2)  {\tiny 2};
        \node[circle] at (3.36,2.2)  {\tiny 1};
        \node[circle] at (3.55,2.2)  {\tiny 0};

        \draw[draw=red] (1.95,-0.04) ellipse (0.1cm and 0.3cm);
        \draw[draw=red] (2.5,1.16) ellipse (0.1cm and 0.3cm);

        \draw[draw=blue] (2.89,1.16) ellipse (0.1cm and 0.3cm);
        \draw[draw=blue] (3.37,2.36) ellipse (0.1cm and 0.3cm);
    
        \draw[-{Stealth},line width=0.7mm] (center)--(n1);
        \draw[-{Stealth},line width=0.7mm] (n2)--(center);
        \draw[-{Stealth},line width=0.7mm] (n3)--(center);
        \draw[-{Stealth},line width=0.7mm] (n4)--(center);
        \draw[-{Stealth},line width=0.7mm] (n5)--(n1);
        \draw[-{Stealth},line width=0.7mm] (n1)--(n7);

        \end{scope}

        \begin{scope}[xshift=3cm]
        \node[circle,fill=pink] at (0,0) (center) {};
        \node at (0,0) {$v$};
        \node[circle,fill=orange] at (0,1.2) (n1) {};
        \node[circle,fill=white,draw=black] at (1,-1.2) (n2) {};
        \node[circle,fill=white,draw=black] at (-1,-1.2) (n3) {};
        \node[circle,fill=white,draw=black] at (0,-1.2) (n4) {};
        \node[circle,fill=white,draw=black] at (-2,0) (n5) {};
        \node[circle,fill=purple] at (0,2.4) (n7) {};

        \node at (1.95,0)  {$\varphi'(v)=(3,0)=110_2$};
        \node at (2.1,1.2)  {$\varphi'(p(v))=(1,1)=11_2$};
        \node at (2.3,2.4)  {$\varphi'(p(p(v)))=(0,0)=0_2$};

        \draw[-{Stealth},line width=0.7mm] (center)--(n1);
        \draw[-{Stealth},line width=0.7mm] (n2)--(center);
        \draw[-{Stealth},line width=0.7mm] (n3)--(center);
        \draw[-{Stealth},line width=0.7mm] (n4)--(center);
        \draw[-{Stealth},line width=0.7mm] (n5)--(n1);
        \draw[-{Stealth},line width=0.7mm] (n1)--(n7);

        \end{scope}

        \draw [-{Stealth[length=5mm]}] (-0.9,0.6) -- (1.1,0.6);

    \end{tikzpicture}

    \caption{An illustration Procedure \textsc{Forest-Reduce-Colors}}
    \label{fig:procedure forest reduce colors}
\end{figure}
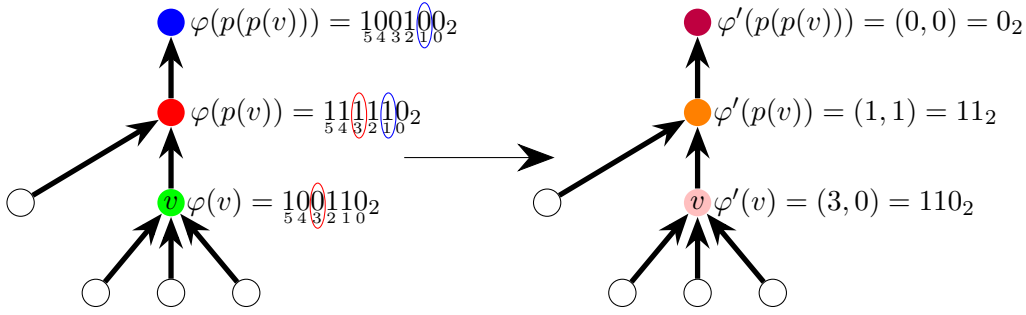

We first analyse the complexity of Procedure \textsc{Forest-Reduce-Colors}.

\begin{lemma}[Complexity of Procedure \textsc{Forest-Reduce-Colors}]\label{forest-reduce-colors complexity}
    Let $F=(V,E)$ be an $n$-vertex $m$-edge oriented forest, and let $\varphi$ be a proper vertex-coloring of $F$.
    Procedure \textsc{Forest-Reduce-Colors} requires $O\left(1\right)$ time using $O\left(m\right)$ processors.
\end{lemma}

\begin{proof}
    Let $v\in V$ be a non-isolated vertex of $F$. We first bound the time required to compute the smallest index $i$, such that $c(v)_i\neq c(p(v))_i$. This can be done by computing $\mathsf{XOR}$ of the two colors, and finding the least significant bit of the result. The latter is a single-word operation, and can be executed in $O(1)$ time in $\mathrm{CRCW}$ $\mathrm{PRAM}$.
    The computation of the new color requires $O(1)$ time and $O(1)$ processors for each non-isolated vertex. The number of non-isolated vertices is linear in the number of edges of $F$. Hence overall, applying this procedure on all the vertices in parallel requires $O(1)$ time using $O(m)$ processors.
\end{proof}

Using Procedure \textsc{Forest-Reduce-Colors}, we now present the 6-vertex-coloring algorithm. As we described above, the algorithm starts with some proper vertex-coloring. To this end, we will define an initial $n$-vertex-coloring, in which each vertex $v$ gets a unique color according to its unique id, $\text{id}(v)$. The algorithm is composed of $\log^* n$ phases. In each of them, it improves the coloring using Procedure \textsc{Forest-Reduce-Colors}.\\
This pseudocode is provided in Procedure \textsc{6-Vertex-Coloring-Forests}.\\
\\$\textsc{Procedure}$ $\textsc{6-Vertex-Coloring-Forests}$ $\left(F=(V,E)\right)$
\begin{description}
    \item{\textbf{Step 1.}} \textbf{For} each $v\in V$ in parallel \textbf{do}\\
    \makebox[1.15cm]{} Define $\varphi(v)\leftarrow \text{id}(v)$
    \item{\textbf{Step 2.}} \textbf{For} $i=1,2,...,\log^*n$ \textbf{do}\\
    \makebox[1.15cm]{} \textsc{Forest-Reduce-Colors}$(F,\varphi)$
\end{description}

For $i\in\{0,1,2,...,\log^* n\}$, we denote by $\varphi_i$ the coloring $\varphi$ after the $i$'th iteration of the for loop in step 2. We start by analysing the size of the palette used by each coloring $\varphi_i$. See, e.g.,~\cite{doi:10.1137/1.9780898719772}, Chapter 7, for a proof of the next lemma.

\begin{lemma}[Size of palettes]\label{Bound on the coloring size}
    Let $F=(V,E)$ be an $n$-vertex oriented forest.
    For each $i\in\{0,1,2,...,\log^* n\}$, we have $|\varphi_i|=O\left(\log^{(i)}n\right)$.
\end{lemma}

We are now ready to analyse the complexity of Procedure \textsc{6-Vertex-Coloring-Forests}.

\begin{lemma}[Complexity of Procedure \textsc{6-Vertex-Coloring-Forests}]\label{6-vertex-coloring complexity}
    Let $F=(V,E)$ be an $n$-vertex $m$-edge oriented forest.
    Step 1 of Procedure \textsc{6-Vertex-Coloring-Forests} requires $O(1)$ time and $O(n)$ processors. Step 2, i.e., the rest of it, requires $O\left(\log^* n\right)$ time using $O\left(m\right)$ processors. (Note that $m$ might be significantly smaller than $n$.)
\end{lemma}

\begin{proof}
    First, observe that the construction of $\varphi$ in step 1 requires $O(1)$ time using $O(n)$ processors. By Lemma \ref{forest-reduce-colors complexity}, for each $i\in\{1,2,...,\log^*n\}$, the $i$'th iteration of the for loop in step 2 requires $O\left(1\right)$ time using $O(m)$ processors. Hence, overall, step 2 of Procedure \textsc{Forest-Reduce-Colors} requires $\sum_{i=1}^{\log^*n}O(1)=O\left(\log^* n\right)$ time using $O\left(m\right)$ processors.
\end{proof}

The correctness of this algorithm follows from the correctness of the algorithm of Cole and Vishkin~\cite{cole1986deterministic}. We summarize this result in the following theorem.

\begin{theorem}[Properties of Procedure \textsc{6-Vertex-Coloring-Forests}]
    Let $F=(V,E)$ be an $n$-vertex $m$-edge oriented forest.
    Procedure \textsc{6-Vertex-Coloring-Forests} computes a 6-vertex-coloring of $F$ in $O\left(\log^* n\right)$ time using $O\left(n\right)$ processors. Moreover, step 1 of the procedure requires $O(1)$ time using $O(n)$ processors, while its step 2 (the rest of the procedure) requires $O\left(\log^* n\right)$ time using $O\left(m\right)$ processors.
\end{theorem}

We now move on and describe how to improve this 6-vertex-coloring to the desired 3-vertex-coloring. To this end, we now present a simple algorithm that given an integer $k\geq 4$ and a $k$-vertex-coloring of an oriented forest, reduces one color of the coloring, i.e., the algorithm computes a proper $(k-1)$-vertex-coloring of the forest. The algorithm consists of two steps. At the first step, it recolors all the vertices to the color of their parent. Observe that this step keeps the coloring proper, and provides the property that for each vertex, all its neighbors are colored using at most two different colors (one color for its parent, and one color for its children - which is its original color). Using this property, all the vertices that are colored $k$, choose in parallel a color from $\{1,2,3\}$ that is not used by their neighbors. In that way, we get rid of the color $k$, and get a new proper $(k-1)$-vertex-coloring. The pseudocode of this algorithm is provided in Procedure \textsc{Forest-Reduce-Color}.\\
\\$\textsc{Procedure}$ $\textsc{Forest-Reduce-Color}$ $\left(F=(V,E),\varphi\right)$
\begin{description}
    \item{\textbf{Step 1.}} \textbf{For} each $v\in V$ in parallel \textbf{do}\\
    \makebox[1.15cm]{} Let $\varphi'(v)\leftarrow \varphi(v)$\\
    \makebox[1.15cm]{} Define $\varphi(v)\leftarrow \varphi(p(v))$ \Comment{if $p(v)=\emptyset$}, $v$ chooses an arbitrary color in \myindent{6.8}$\{1,2,3\}\setminus \{\varphi(v)\}$\\
    \makebox[1.15cm]{} \textbf{If} $\varphi(v)=|\varphi|$ \textbf{do}\\
    \makebox[1.5cm]{} Define $\varphi(v)$ to be an arbitrary color from $\{1,2,3\}\setminus\{\varphi'(v),\varphi(p(v))\}$ 
\end{description}

Observe that this algorithm can be implemented by assigning a processor to each forest edge (or equivalently, to each non-isolated vertex), that does all the computations in $O(1)$ time. We summarize this result in the next lemma.

\begin{lemma}[Complexity of Procedure \textsc{Forest-Reduce-Color}]\label{Forest-reduce-color complexity}
    Let $F=(V,E)$ be an $n$-vertex oriented forest, $k\geq 4$ an integer, and $\varphi$ be a $k$-vertex-coloring of $F$.
    Procedure \textsc{Forest-Reduce-Color} requires $O\left(1\right)$ time using $O\left(m\right)$ processors.
\end{lemma}

Now we are ready to combine all the parts that we described above, to devise an algorithm for 3-vertex-coloring oriented forests.\\
\\$\textsc{Procedure}$ $\textsc{3-Vertex-Coloring-Forests}$ $\left(F=(V,E)\right)$
\begin{description}
    \item{\textbf{Step 1.}} Compute a 6-vertex-coloring $\varphi$ of $F$ using Procedure \textsc{6-Vertex-Coloring-Forests}$(F)$
    \item{\textbf{Step 2.}} \textbf{For} $k=6,5,4$ \textbf{do}\\
    \makebox[1.15cm]{} \textsc{Forest-Reduce-Color}$(F,\varphi)$
\end{description}

By Lemma \ref{6-vertex-coloring complexity}, step 1 requires $O(\log^* n)$ using $O(n)$ processors, and each iteration of step 2 requires $O(1)$ time using $O(n)$ processors. Hence, Procedure \textsc{3-Vertex-Coloring-Forests} requires $O(\log^* n)$ time using $O(n)$ processors. Moreover, all steps of the algorithm except for step 1 of Procedure \textsc{6-Vertex-Coloring-Forests} (that initialize each color as identity number) requires $O(m)$ processors.
The correctness of the algorithm follows from that of the algorithm of Cole and Vishkin~\cite{cole1986deterministic}. We summarize this result in the following theorem.

\begin{theorem}[Properties of 3-vertex-coloring]\label{3-vertex-coloring properties}
    Let $F=(V,E)$ be an $n$-vertex $m$-edge oriented forest.
    Procedure \textsc{3-Vertex-Coloring-Forests} computes a 3-vertex-coloring of $F$ in $O\left(\log^* n\right)$ time using $O\left(n\right)$ processors. Moreover, all steps of the algorithm, except for the step that initializes the colors, requires $O(m)$ processors.
\end{theorem}

\subsection{An Adaptation of Linial's Distributed Algorithm to $\mathrm{PRAM}$}
Linial's Algorithm~\cite{linial1987distributive} provides $O(\Delta^2)$-vertex-coloring in $\log^* n +O(1)$ distributed time. We now argue that it can be efficiently implemented in $\mathrm{PRAM}$, and analyse its complexity. 

The algorithm is based upon a construction of $\Delta$-union-free set systems due to~\cite{erdHos1985families}.

\begin{definition}[$\Delta$-union-free set systems]
    For positive integer parameters $m$ and $\Delta$, a set system $\mathcal{F}$ over groundset $\{1,2,...,m\}$ is called \emph{$\Delta$-union-free} if for any $S_0\in \mathcal{F}$, and $S_1,S_2,...,S_\Delta\in \mathcal{F}\setminus\{S_0\}$, we have $S_0\not\subseteq\bigcup_{i=1}^{\Delta}S_i$.
\end{definition}
Erd\H{o}s et al.~\cite{erdHos1985families} provided two constructions of $\Delta$-union-free set systems with similar parameters. One of them is probabilistic, and another one is algebraic. We will now sketch their algebraic construction.\\

For a prime number $q$, consider the ring of degree-$d$ polynomials (for parameters $d$ and $q$ that will be determined below) over the field $GF(q)$ of characteristics $q$. We will build $N=q^{d+1}$ sets over the groundset $\{1,2,...,m\}$, $m=q^2$, in the following way. Note that there are $N=q^{d+1}$ degree-$d$ polynomials as above. For each such a polynomial $p(\cdot)$, let $S_p=\{(i,p(i))\mid i\in GF(q)\}$. Note that two such polynomials may intersect in at most $d$ points. Thus by setting $q>\Delta\cdot d$, we guarantee that the family is $\Delta$-union-free. We will set $q=O(\Delta\cdot d)$ (prime). Given that $N=q^{d+1}$, we conclude that $q\approx O\left(\frac{\Delta\cdot\log N}{\log(\Delta\cdot\log N)}\right)$.\\

Given a proper $N$-vertex-coloring $\varphi$ of a graph $G$ with maximum degree at most $\Delta$, one associates a set $S_c\in \mathcal{F}$ with every color $c\in \{1,2,...,N\}$. Each vertex $v$ then finds a color $\alpha\in S_{\varphi(v)}\setminus\bigcup_{u\in N(v)}S_{\varphi(u)}$, and sets $\alpha$ to be its new color. It is easy to see (cf.~\cite{10.5555/2534493}, Chapter 3.10) that as a result, one obtains a proper $O(q^2)=O\left(\frac{\Delta^2\cdot\log^2 N}{\log^2(\Delta\cdot\log N)}\right)$-vertex-coloring of the graph within one round of distributed computation. By repeating this for $\log^* n+O(1)$ rounds, one ends up with a proper $O(\Delta^2)$-vertex-coloring.\footnote{The algorithm described here leads directly to $O(\Delta^{2} \cdot \log \Delta)$-coloring. To obtain $O(\Delta^{2})$-coloring, one uses the scheme once again with $d = 2$. This additional step does not affect the asymptotic complexity of the resulting algorithm.}\\

To implement this in $\mathrm{PRAM}$ model, we designate $q$ subsets of processors for every vertex $v$. Each such a subset contains $\deg(v)$ processors. Now, for every $i\in GF(q)$, there are $\deg(v)$ processors that check (in $O(1)$ time) if there exists a neighbor $u$ of $v$ whose set $S_{\varphi(u)}$ contains the element $(i,p_{\varphi(v)}(i))$. Within additional $O(\log q)$ time, the processors then find an index $i$ such that $(i,p_{\varphi(v)}(i))$ does not belong to $\bigcup_{u\in N(v)}S_{\varphi(u)}$.\\

Overall this requires $O(\log\Delta+\log\log N)$ time, and $O(|E|\cdot q)=O\left(|E|\cdot\frac{\Delta\cdot\log n}{\log(\Delta\cdot\log n)}\right)$ processors. The operation is then iterated $\log^*n+O(1)$ time, but the number of initial colors $N_1=n,N_2,...$ rapidly decreases. As a result, the overall time complexity of Linial's algorithm is $O(\log\Delta\cdot\log^*n+\log\log n)$, and it employs $O\left(|E|\cdot\frac{\Delta\cdot\log n}{\log(\Delta\cdot\log n)}\right)$ processors.

We summarize this discussion with the following theorem.
\begin{theorem}[An adaptation of~\cite{linial1987distributive}]\label{An adaptation of [Lin87]}
    Given an $n$-vertex $m$-edge graph $G=(V,E)$ with maximum degree $\Delta$, the algorithm computes a proper $O\left(\Delta^2\right)$-vertex-coloring in $O\left(\log\Delta\cdot\log^*n+\log\log n\right)$ time, using $O\left(m\cdot\frac{\Delta\cdot\log n}{\log(\Delta\cdot\log n)}\right)$ processors.
\end{theorem}

Barenboim and Elkin~\cite{barenboim2008sublogarithmic} extended this algorithm to graphs with bounded arboricity $a$. Specifically, they showed that given an orientation of all graph edges in which each vertex has $A=O(a)$ parents, a similar algorithm (which they called Arb-Linial) computes an $O(A^2)$-coloring of the graph in $\log^* n+O(1)$ rounds.
It is easy to see that Arb-Linial can also be implemented in $\mathrm{PRAM}$ in a similar way. We summarize this in the following corollary.

\begin{corollary}[Adaptation of~\cite{barenboim2008sublogarithmic}, Algorithm Arb-Linial]\label{Adaptation of [BE08], Algorithm Arb-Linial}
    Given an $n$-vertex $m$-edge graph $G=(V,E)$ with arboricity $a$ and an orientation of edges of $G$ such that every vertex has at most $A$ parents, Algorithm Arb-Linial computes an $O(A^2)$-coloring of $G$ in $O(\log A\cdot\log^*n+\log\log n)$ time using $O\left(m\cdot\frac{A\cdot\log n}{\log(A\cdot\log n)}\right)$ processors.
\end{corollary}

One can also trade the number of processors for running time. Given a parameter $z$, $1\leq z\leq q$, we designate $z$ subsets of $\deg(v)$ processors each, to every vertex $v$. Each subset is now in charge for up to $\left\lceil\frac{q}{z}\right\rceil$ elements $i\in GF(q)$. Hence each subset will now need $O\left(\frac{q}{z}\right)$ time to check if there exists an element $i\in GF(q)$ (among elements of $GF(q)$ for which this subset of processors is in charge), such that there exists a neighbor $u$ of $v$ whose subset $S_{\varphi(u)}$ contains $(i,p_{\varphi(v)}(i))$.
At this point we have $q$ elements $0$ or $1$, indicating for every $i\in GF(q)$ whether $(i,p_{\varphi(v)}(i))$ appears in a subset $S_{\varphi(u)}$ of one of the neighbors of $v$.
We have $Z=z\cdot\deg(v)$ processors that need to find a zero in this array. Each processor is in charge for a segment of $\frac{q}{Z}$ elements. Within $O\left(\frac{q}{Z}\right)$ time, each of the $Z$ processors finds out if its designated segment contains a zero. Within $O(1)$ additional time (recall that we consider $\mathrm{CRCW\,\,PRAM}$), the computation is over, and an element $i\in GF(q)$ such that $(i,p_{\varphi(v)}(i))$ does not appear in $\bigcup_{u\in\Gamma(v)} S_{\varphi(u)}$ is found.
Thus, the running time for this step is $O\left(\frac{q}{z\cdot\deg(v)}\right)$,
i.e., the overall time is $O\left(\frac{q}{z}\right)$, and the overall number of processors is $\sum_{v\in V}z\cdot\deg(v)=O(|E|\cdot z)$. Substituting $q=O\left(\frac{\Delta\cdot\log n}{\log(\Delta\cdot\log n)}\right)$, the running time becomes $O\left(\frac{\Delta\cdot\log n}{z\cdot\log(\Delta\cdot\log n)}+\log\log n+\log\Delta\right)$. Summing up over $\log^*n$ phases (in which the number of initial colors $N_1$, $N_2$, ... rapidly decreases), we obtain an overall time of $O\left(\frac{\Delta\cdot\log n}{z\cdot\log(\Delta\cdot\log n)}+\log\log n+\log\Delta\cdot\log^*n\right)$.
\begin{corollary}[An adaptation of~\cite{linial1987distributive}, a trade-off]\label{An adaptation of [Lin87], a trade-off}
    Given an $n$-vertex $m$-edge graph $G=(V,E)$ with maximum degree $\Delta$, and a parameter $1\leq z\leq q=O\left(\frac{\Delta\cdot\log n}{\log(\Delta\cdot\log n)}\right)$, our algorithm computes a proper $O(\Delta^2)$-vertex-coloring in 
    $$O\left(\log\Delta\cdot\log^* n+\log\log n+\frac{\Delta\cdot\log n}{z\cdot\log(\Delta\cdot\log n)}\right)$$
    time, using $O\left(m\cdot z\right)$ processors.
\end{corollary}

Corollary \ref{Adaptation of [BE08], Algorithm Arb-Linial} extends similarly, with the parameter $z$ in the range 
$1\leq z=O\left(\frac{A\cdot\log n}{\log(A\cdot\log n)}\right).$

\subsection{Extension of Linial's Algorithm to Defective Coloring}\label{sec: Extension of Linial's Algorithm to Defective Coloring}
\begin{definition}[Defective coloring]
    A vertex-coloring $\varphi$ is called \emph{$D$-defective} (for a parameter $D$) if for every vertex $v$, there are at most $D$ neighbors of $v$ colored by $\varphi(v)$.
\end{definition}

Barenboim et al.~\cite{doi:10.1137/12088848X} extended Linial's algorithm to defective coloring. Specifically, they showed that for any parameter $p\leq \Delta$, a $\left(\frac{\Delta}{p}\right)$-defective $O(p^2)$-coloring can be computed in $\log^*n+O(1)$ distributed rounds. Next, we sketch this extension, and argue that it can also be efficiently implemented in $\mathrm{PRAM}$.
\begin{definition}[$\rho$-cover]
    For a set $S_0$, and $\Delta$ other sets $S_1,S_2,...,S_{\Delta}$, and a parameter $\rho$, we say that sets $S_1,S_2,...,S_{\Delta}$ \emph{$\rho$-cover} $S_0$ if every element $x\in S_0$, appears in at least $\rho$ set among $S_1,S_2,...,S_{\Delta}$.  
\end{definition}

\begin{definition}[$\Delta$-union $(\rho+1)$-cover-free]
        For a pair of parameters $\Delta$ and $\rho$, a set system $\mathcal{F}$ is \emph{$\Delta$-union $(\rho+1)$-cover-free} if for any set $S_0\in \mathcal{F}$ and sets $S_1,S_2,...,S_{\Delta}\in\mathcal{F}\setminus\{S_0\}$, the sets $S_1,S_2,...,S_{\Delta}$ do not $\rho$-cover $S_0$. 
\end{definition}

\textbf{The basic step.} Suppose that we have a $D$-defective $N$-coloring $\varphi$ of a graph $G$, and a $\Delta$-union $(\rho+1)$-cover-free family $\mathcal{F}$ with $N$ sets, over a groundset $[m]$, for some $D,N,\Delta,\rho$ and $m$. It is easy to see that within one round of distributed computation one can then obtain a $(D+\rho)$-defective $m$-coloring of $G$.\\
The algorithm of~\cite{doi:10.1137/12088848X} starts with an $O(\Delta^2)$-coloring (0-defective) computed via Linial's algorithm, and refines it by iterating the basic step (see above) with appropriate parameters for $O(\log^*\Delta)$ rounds. Consider again the set system $\mathcal{F}=\mathcal{F}_{q,d}$, in which each set contains $q$ elements. (Each set $S_p$ corresponds to a degree-$d$ polynomial $p(\cdot)$ over $GF(q)$. It is equal to $\{(i,p(i))\mid i\in GF(q)\}$.) We have $|\mathcal{F}|=N=q^{d+1}$. 

To hit each element of a set $S_0\in \mathcal{F}$ for $\rho+1$ times, one needs at least $\Delta\geq\frac{q(\rho+1)}{d}$ other sets $S_1,S_2,...,S_{\Delta}$. So we set $\Delta$ slightly smaller than $\frac{q(\rho+1)}{d}$, to guarantee that the family is $\Delta$-union $(\rho+1)$-cover-free. The size of the groundset satisfies $m=q^2$, and thus it follows that  
\begin{equation}\label{eq: def coloring}
    m\leq\left(\frac{\Delta+1}{\rho+1}\right)^2\log^2 N.
\end{equation}
Similarly to Linial's algorithm, it is easy to adapt this scheme to $\mathrm{PRAM}$ setting. Again, we designate $q$ disjoint subsets of processors to every vertex $v$. Each such a set consists of $\deg(v)$ processors.  For a value $(i,p_{\varphi(v)}(i))\in S_{\varphi(v)}$, the $\deg(v)$ processors designated by $v$ to the value $i$ check in how many sets $S_{\varphi(u)}$, $u\in N(v)$, appears the pair $(i,p_{\varphi(v)}(i))$. If this number is $\rho$ or smaller, they raise a flag. Finally, processors designated for $v$ find $i\in GF(q)$ with a raised flag. (Its existence is guaranteed by the analysis.) The overall time is $O(\log q)$, and the number of processors is $O(|E|\cdot q)$. As $m=q^2$, by inequality (\ref{eq: def coloring}) we conclude that the running time is $O\left(\log\Delta+\log\log N\right)=O\left(\log\Delta+\log\log n\right)$, and the number of processors is $O\left(|E|\cdot\frac{\frac{\Delta}{\rho}\cdot\log n}{\log\left(\frac{\Delta}{\rho}\cdot\log n\right)}\right)$. As the algorithm performs $O(\log^*\Delta)$ iterations of this step (with decreasing palettes), it can be verified that the overall time is $O(\log\Delta\cdot\log^*\Delta+\log\log n)$ with $O\left(|E|\cdot\frac{\Delta\cdot\log n}{\log(\Delta\cdot\log n)}\right)$ processors (in addition to the time $O(\log\Delta\cdot\log^* n+\log\log n)$ needed to compute an $O(\Delta^2)$-coloring via Linial's algorithm).

We summarize this discussion with the following theorem.
\begin{theorem}[An adaptation of defective coloring~\cite{doi:10.1137/12088848X}]\label{adaptation of BEK14}
    Given an $n$-vertex $m$-edge graph $G=(V,E)$ with maximum degree $\Delta$, and a parameter $p$, $1\leq p\leq\Delta$, a $\left(\frac{\Delta}{p}\right)$-defective $O\left(p^2\right)$-vertex-coloring can be computed in $O\left(\log\Delta\cdot\log^*n+\log\log n\right)$ time, using $O\left(m\cdot\frac{\Delta\cdot\log n}{\log(\Delta\cdot\log n)}\right)$ processors.
\end{theorem}

Similarly to Theorem \ref{An adaptation of [Lin87]} and Corollary \ref{Adaptation of [BE08], Algorithm Arb-Linial}, this analysis extends to a trade-off between running time and number of processors.

\begin{theorem}[An adaptation of~\cite{doi:10.1137/12088848X}, a trade-off]\label{An adaptation of [BEK14], a trade-off}
Given an $n$-vertex $m$-edge graph $G=(V,E)$ with maximum degree $\Delta$, and parameters $p$ and $z$, $1\leq p\leq\Delta$, $1\le z=O\left(\frac{\Delta\cdot\log n}{\log(\Delta\cdot\log n)}\right)$, a $\left(\frac{\Delta}{p}\right)$-defective $O(p^2)$-vertex-coloring can be computed in  $$O\left(\frac{\Delta\cdot\log n}{z\cdot\log(\Delta\cdot\log n)}+\log\Delta\cdot\log^*n+\log\log n\right)$$ time, using $O\left(m\cdot z\right)$ processors.
\end{theorem}

One can further improve the number of processors in Theorem~\ref{adaptation of BEK14}. Note that the only step in the algorithm that employs $O\left(|E|\cdot \frac{\Delta\cdot\log n}{\log\left(\Delta\cdot\log n\right)}\right)$ processors is the step that computes an $O\left(\Delta^{2}\right)$-coloring (via Linial's algorithm), while the other steps employ only $O\left(|E|\cdot\frac{\frac{\Delta}{\rho}\cdot\log n}{\log\left(\frac{\Delta}{\rho}\cdot\log n\right)}\right)$ processors. On the other hand, it is not hard to see that the distributed algorithm of Barenboim et al. (see~\cite{doi:10.1137/12088848X}, Section 3.2) can be applied directly on the na\"ive $n$-coloring of the vertex set $V$, i.e., the coloring that assigns every vertex $v \in V$ the color $\mathrm{Id}(v)$. In other words, the step that computes an $O\left(\Delta^{2}\right)$-coloring can be skipped altogether. The resulting coloring is still an $O\left(\left(\frac{\Delta+1}{\rho+1}\right)^2\right)$-coloring, $O(\rho)$-defective, while the distributed running time grows slightly: instead of $\frac{1}{2}\cdot\log^{*} n + O(\log^{*} \Delta)$ of the original version of the algorithm of~\cite{doi:10.1137/12088848X}, this version has running time of $O\left(\log^{*} n\right)$. The $\mathrm{PRAM}$ running time of this version is $O\left(\log\Delta\cdot\log^* n+\log\log n\right)$, because we now apply $O(\log^* n)$ recoloring steps (that employ $(\rho+1)$-union-free set systems), as opposed to $O(\log^{*} \Delta)$ recoloring steps in the version described above. On the other hand, the number of processors is now just $O\left(|E|\cdot\frac{\frac{\Delta}{\rho}\cdot\log n}{\log\left(\frac{\Delta}{\rho}\cdot\log n\right)}\right)$ (instead of $O\left(|E|\cdot \frac{\Delta\cdot\log n}{\log\left(\Delta\cdot\log n\right)}\right)$). We now set $p=O(\frac{\Delta}{\rho})$, and obtain an $O(p^{2})$-coloring, $O(\frac{\Delta}{p})$-defective, using $O\left(|E|\cdot \frac{p\cdot\log n}{\log\left(p\cdot\log n\right)}\right)$ processors. 

To summarize:
\begin{theorem}[A variant of defective coloring of~\cite{doi:10.1137/12088848X}]\label{A variant of defective coloring of BEK14}
    Given an $n$-vertex $m$-edge graph $G=(V,E)$ with maximum degree $\Delta$, and a parameter $p$, $1\leq p\leq\Delta$, a $\left(\frac{\Delta}{p}\right)$-defective $O\left(p^2\right)$-vertex-coloring can be computed in $O\left(\log\Delta\cdot\log^*n+\log\log n\right)$ time, using $O\left(m\cdot\frac{p\cdot\log n}{\log(p\cdot\log n)}\right)$ processors.
\end{theorem}

One can also have a trade-off version of this result, analogous to Theorem~\ref{An adaptation of [BEK14], a trade-off}.

\begin{theorem}[An adaptation of~\cite{doi:10.1137/12088848X}, a better trade-off]\label{An adaptation of [BEK14], a better trade-off}
Given an $n$-vertex $m$-edge graph $G=(V,E)$ with maximum degree $\Delta$, and parameters $p$ and $z$, $1\leq p\leq\Delta$, $1\le z=O\left(\frac{p\cdot\log n}{\log(p\cdot\log n)}\right)$, a $\left(\frac{\Delta}{p}\right)$-defective $O(p^2)$-vertex-coloring can be computed in  $$O\left(\frac{p\cdot\log n}{z\cdot\log(p\cdot\log n)}+\log\Delta\cdot\log^*n+\log\log n\right)$$ time, using $O\left(m\cdot z\right)$ processors.
\end{theorem}

\subsection{Algebraic Color Reduction}\label{app: Colors Reducing Procedure}
In this section we present a version of the distributed algorithm due to Barenboim et al.~\cite{barenboim2018locally}. The algorithm receives as an input a graph $G=(V,E)$ with maximum degree $\Delta$, equipped with a $k$-vertex-coloring $\varphi$, for $k=O(\Delta^2)$, and returns an $O\left(\Delta\right)$-vertex-coloring of $G$. We refer to this algorithm as Procedure \textsc{Algebraic-Color-Reduction}.\\

In the algorithm, we first choose a prime parameter $\sqrt{k}\leq p=O\left(\Delta\right)$, and consider each color $\varphi(v)$, for $v\in V$, as a pair $\langle a_v,b_v\rangle$, where $a_v,b_v\in\{0,1,...,p-1\}$. Then, in parallel, each vertex $v\in V$, checks whether it has a neighbor $u$ with $b_v=b_u$. If it has such a neighbor, it defines $\varphi(v)=\langle a_v,a_v+b_v\,(\text{mod } p)\rangle$. Otherwise, define $\varphi(v)=\langle 0,b_v\rangle$. We continue with this process for $p$ iterations. Barenboim et al.~\cite{barenboim2018locally} showed that this procedure produces a proper vertex-coloring, and that it is guaranteed that by the end of this process, for each vertex $v\in V$, we have $a_v=0$, i.e., we are left with $p=O(\Delta)$-vertex-coloring of the graph. 

In order to implement this algorithm in $\mathrm{PRAM}$, we assign a processor to each edge $(u,v)\in E$, that checks, in parallel, whether $b_v=b_u$, and informs $u$ and $v$. Then each vertex, in parallel, defines its new color according to the results of its incident edges. Since this process is applied iteratively $O(p)$ times, we conclude that at the end of this process, we obtain a $p=O(\Delta)$-vertex-coloring of $G$ in $O(p\log\Delta)=O(\Delta\log\Delta)$ time using $O(m)$ processors. We summarize this result in the next theorem. ($O(\log \Delta)$ time is required for a vertex $v$ to determine if there exists a neighbor $u$ of $v$ such that the respective processor associated with the edge $(u,v)$ raised a flag indicating that $b_u=b_v$.)

\begin{theorem}[An adaptation of the algorithm of~\cite{barenboim2018locally}]\label{Properties of Procedure Algebraic-Color-Reduction}
    Given an $n$-vertex $m$-edge graph $G=(V,E)$ with maximum degree $\Delta$ and a proper $O(\Delta^2)$-vertex-coloring of $G$, a proper $O(\Delta)$-vertex-coloring of $G$ can be computed in $O(\Delta\cdot\log\Delta)$ time using $O(m)$ processors.
\end{theorem}

\subsection{$(\Delta+1)$-Vertex-Coloring Algorithm for General Graphs}\label{(Delta+1)-Vertex-Coloring Algorithm For General Graphs}
In this subsection we present a $\mathrm{CRCW\,\, PRAM}$ algorithm for $(\Delta+1)$-vertex-coloring general graphs. This algorithm is based on algorithms of~\cite{panconesi2001some} and~\cite{barenboim2018locally}.\\

We start with describing a merging phase of the $(\Delta+1)$-vertex-coloring algorithm, in which we receive two proper vertex-colorings of two edge-disjoint subgraphs of the input graph $G$, and merge them into one proper vertex-coloring of the union of these subgraphs. As a first part of describing this merging algorithm, we present a simple procedure that given a graph $G=(V,E)$ with maximum degree $\Delta$, an integer $k>\Delta+1$, and a proper $k$-vertex-coloring $\varphi$ of $G$, reduces one color from the coloring. That is, the algorithm produces a proper $(k-1)$-vertex-coloring of $G$. The algorithm simply recolors every vertex $v$ that is colored $k$, using a color from $\{1,2,...,\Delta+1\}\setminus \varphi(N(v))\neq \emptyset$, where $\varphi(N(v))$ is the set of colors used by the neighbors of $v$. The pseudocode of the algorithm is given in Procedure \textsc{Reduce-Color}.\\
\\$\textsc{Procedure}$ $\textsc{Reduce-Color}$ $\left(G=(V,E),\varphi\right)$
\begin{description}
    \item{\textbf{Step 1.}} \textbf{For} each $v\in V$ in parallel \textbf{do}\\
    \makebox[1.15cm]{} \textbf{If} $\varphi(v)=|\varphi|$ \textbf{do}\\
    \makebox[1.5cm]{} Define $\varphi(v)$ to be an arbitrary color from $\{1,2,...,\Delta+1\}\setminus \varphi(N(v))$\\
\end{description}
First observe that since each vertex $v\in V$ has at most $\Delta$ neighbors, then $\{1,2,...,\Delta+1\}\setminus \varphi(N(v))\neq \emptyset$, and step 1 is well-defined. In the next theorem we analyse the complexity of Procedure \textsc{Reduce-Color}.

\begin{theorem}[Properties of Procedure \textsc{Reduce-Color}]\label{Procedure vertex Reduce-Color properties}
    Let $G=(V,E)$ be an $m$-edge graph with maximum degree $\Delta$, $k>\Delta+1$ be an integer, and $\varphi$ be a $k$-vertex-coloring of $G$. 
    Then Procedure \textsc{Reduce-Color} computes a proper $(k-1)$-vertex-coloring of $G$ in $O\left(\log\Delta\right)$ time using $O\left(m\right)$ processors.
\end{theorem}

\begin{proof}
    Let $\varphi'$ be the vertex-coloring produced by Procedure \textsc{Reduce-Color}. It is easy to see that this is a proper $(k-1)$-vertex-coloring (cf., e.g.,~\cite{doi:10.1137/1.9780898719772}, Chapter 7.1).
    
    We now analyse the complexity of the algorithm. For each vertex $v\in V$ that is colored $k$, we compute its new color by assigning a processor to each neighbor of $v$, sort these vertices according to their color, and choose a free color in $\{1,2,...,\Delta+1\}\setminus \varphi(N(v))$. This process requires $O(\log\Delta)$ time using $O(\deg(v))$ processors, one for each neighbor of $v$. Overall, Procedure \textsc{Reduce-Color} requires $O\left(\log \Delta\right)$ time using $\sum_{v\in V}O(\deg(v))=O\left(m\right)$ processors.
\end{proof}

Next, we present the merging algorithm. This algorithm receives as input two proper $(\Delta+1)$-vertex-colorings $\varphi_1$ and $\varphi_2$, of two edge-disjoint graphs $G_1=(V,E_1)$ and $G_2=(V,E_2)$, respectively, and returns a combined proper $(\Delta+1)$-vertex-coloring $\varphi$ of their union $G=(V,E_1\cup E_2)$, where $\Delta$ is a bound on the maximum degree of $G$. The algorithm first defines a new color for each vertex, as a pair of its colors in $\varphi_1$ and $\varphi_2$. This coloring is a proper $O(\Delta^2)$-vertex-coloring. Then, we reduce all colors greater than $\Delta+1$ in two steps. First, using Procedure \textsc{Algebraic-Color-Reduction} from Appendix \ref{app: Colors Reducing Procedure}, we compute a proper $O(\Delta)$-vertex-coloring, and then by iteratively applying Procedure \textsc{Reduce-Color}, we reduce the $O(\Delta)$ extra colors, and receive a $(\Delta+1)$-vertex-coloring. See Procedure \textsc{Merge} for a detailed description of the algorithm.\\
\\$\textsc{Procedure}$ $\textsc{Merge}$ $\left(G_1,G_2,\varphi_1,\varphi_2,\Delta\right)$
\begin{description}
    \item{\textbf{Step 1.}} \textbf{For} each $v\in V$ in parallel \textbf{do}\\
    \makebox[1.15cm]{} Define $\varphi(v)\leftarrow(\varphi_1(v),\varphi_2(v))$\Comment{$\varphi(v)\leftarrow\varphi_1(v)\cdot|\varphi_2|+\varphi_2(v)$}
    \item{\textbf{Step 2.}} $\varphi\leftarrow\textsc{Algebraic-Color-Reduction}(G,\varphi)$\Comment{See Appendix \ref{app: Colors Reducing Procedure}}
    \item{\textbf{Step 3.}} \textbf{while} $|\varphi|>\Delta+1$ \textbf{do}\\
    \makebox[1.15cm]{} \textsc{Reduce-Color}$(G_1\cup G_2,\varphi)$
\end{description}

Let $G_1=(V,E_1)$ and $G_2=(V,E_2)$ be two edge-disjoint graphs that are properly vertex-colored with the colorings $\varphi_1$ and $\varphi_2$, respectively. We first show that the coloring $\varphi$ produced by Procedure \textsc{Merge} is indeed a proper $(\Delta+1)$-vertex-coloring of $G_1\cup G_2$, and analyse the complexity of the procedure.

\begin{theorem}[Properties of Procedure \textsc{Merge}]\label{Procedure Merge properties}
    Let $G_1=(V,E_1)$ and $G_2=(V,E_2)$ two edge-disjoint $n$-vertex graphs, let $\Delta$ be a bound on the maximum degree of $G=G_1\cup G_2$, and let $m$ be the number of edges in $G$. Let $\varphi_1$ (respectively, $\varphi_2$) be a proper $(\Delta+1)$-vertex-coloring of $G_1$ (resp., $G_2$). Procedure \textsc{Merge} computes a proper $(\Delta+1)$-vertex-coloring $\varphi$ of $G$ in $O(\Delta\cdot\log\Delta)$ time using $O(m)$ processors.
\end{theorem}

\begin{proof}
    Let $(u,v)\in E_1\cup E_2$. Since $\varphi_1$ and $\varphi_2$ are proper vertex-colorings of $G_1$ and $G_2$, respectively, then either $\varphi_1(u)\neq\varphi_1(v)$ or $\varphi_2(u)\neq\varphi_2(v)$. Hence $\varphi(v)$ (defined on step 1) satisfies $\varphi(u)=(\varphi_1(u),\varphi_2(u))\neq(\varphi_1(v),\varphi_2(v))=\varphi(v)$. Therefore, since $|\varphi_1|\leq \Delta+1$ and $|\varphi_2|\leq \Delta+1$, $\varphi$ is a proper $O(\Delta^2)$-vertex-coloring of $G$. Hence, by Theorem \ref{Properties of Procedure Algebraic-Color-Reduction}, Procedure \textsc{Algebraic-Color-Reduction} computes a proper $O(\Delta)$-vertex-coloring. Since we apply Procedure \textsc{Reduce-Color} on the coloring $\varphi$ until it uses at most $\Delta+1$ colors, and since Procedure \textsc{Reduce-Color} reduces one color and returns a proper vertex-coloring of $G$, the returned coloring $\varphi$ is indeed a proper $(\Delta+1)$-vertex-coloring of $G$.\\
    
    Next, we analyse the complexity of the algorithm. For step 1 we assign a processor to each $v\in V$ that defines its new colors. This part requires $O(1)$ using $O(n)$ processors. Next, by Theorem \ref{Properties of Procedure Algebraic-Color-Reduction} (from Appendix \ref{app: Colors Reducing Procedure}), step 2 (Procedure \textsc{Algebraic-Color-Reduction}) requires $O(\Delta\cdot\log\Delta)$ time using $O(m)$ processors.
    Observe that before step 3, $|\varphi|=O(\Delta)$. By Theorem \ref{Procedure vertex Reduce-Color properties}, each execution of Procedure \textsc{Reduce-Color} requires $O(\log \Delta)$ time using $O(m)$ processors. Hence, overall Procedure \textsc{Merge} requires $O(\Delta\cdot\log\Delta)$ time using $O(m)$ processors.
\end{proof}

We are now ready to describe the main vertex-coloring algorithm of this section. It is based on the distributed algorithm~\cite{barenboim2014combinatorial}. We start by presenting the idea of the algorithm. Let $G=(V,E)$ an input graph. The algorithm is based on a divide-and-conquer approach. The algorithm first partitions the graph $G$ into $\Delta$ edge-disjoint oriented subforests $F_1,F_2,...,F_{\Delta}$. Then, in parallel, it rapidly 3-vertex-colors each of them, using Procedure \textsc{3-Vertex-Coloring-Forest} from Appendix \ref{sec: 3-Vertex-Coloring of Oriented Forests}. Finally, the algorithm combines these $\Delta$ 3-vertex-colorings into a $(\Delta+1)$-vertex-coloring of the input graph $G$ using Procedure \textsc{Merge}.

We now provide a more detailed description of each of its stages.
\begin{description}
    \item[Partition:] In this stage, we partition the edge-set $E$ into $\Delta$ subsets $E_1,E_2,...,E_{\Delta}$ that define the $\Delta$ oriented forests $G^{(0)}_1=(V,E_1),G^{(0)}_2=(V,E_2),...,G^{(0)}_{\Delta}=(V,E_{\Delta})$. For this part, we assign to each edge an index $i\in \{1,2,...,\Delta\}$ that specifies its subset in the partition. The assignment of these indexes will be performed by the lower-id endpoint of the edge. Each vertex $v\in V$ will assign an index to the edges of the form $(v,u)$ such that $id(v)<id(u)$. The orientation of the edges will be towards their higher-id endpoint. 
    \item[Color:] As we will see later, each of $G^{(0)}_1,G^{(0)}_2,...,G^{(0)}_{\Delta}$ is indeed a forest. Hence, we can compute, in parallel, 3-vertex-colorings $\varphi^{(0)}_1,\varphi^{(0)}_2,...,\varphi^{(0)}_{\Delta}$ of them using Procedure \textsc{3-Vertex-Coloring-Forests} from Section \ref{sec: 3-Vertex-Coloring of Oriented Forests}.
    \item[Merge:] In this part we will merge the $\Delta$ colorings computed in the previous part using Procedure \textsc{Merge}. We merge the colorings $\varphi^{(0)}_1,\varphi^{(0)}_2,...,\varphi^{(0)}_{\Delta}$ in pairs in $h=\lceil\log\Delta\rceil$ iterations. In each iteration we split the graphs into pairs, and merge the colorings of these pairs. After $h=\lceil\log\Delta\rceil$ iterations, we will get a proper $(\Delta+1)$-vertex-coloring of $G$.
\end{description}
We now present the full description of this algorithm.\\
\\$\textsc{Procedure}$ $\textsc{Vertex-Coloring}$ $\left(G=(V,E)\right)$
\begin{description}
    \item{\textbf{Step 1.}} \textbf{For} each $v\in V$ in parallel \textbf{do}\\
    \makebox[1cm]{}- Assign a unique index $\text{ind}(v,u)$ from $\{1,2,...,\Delta\}$ to every edge $(v,u)$ \myindent{1.12} incident in $v$, such that $v$ is its lower-id endpoint\\
    \makebox[1cm]{}- For every $i\in\{1,2,...,\Delta\}$, let $E_i=\{e\in E\,|\,\text{ind}(e)=i\}$, and $G^{(0)}_i=(V,E_i)$. 
    \myindent{1}- Orient each edge towards its higher-id endpoint.\\
    \makebox[1cm]{}- For every $i\in\{\Delta+1,\Delta+2,...,2^h\}$, let $G^{(0)}_i=(V,\emptyset)$. \\\makebox[1cm]{}\Comment{empty graphs, to simplify indexing}
    \item{\textbf{Step 2.}} \textbf{For} each $i=1,2,...,2^h$ in parallel \textbf{do}\\
    \makebox[1.15cm]{} $\varphi^{(0)}_i\leftarrow$\textsc{3-Vertex-Coloring-Forests}$\left(G^{(0)}_i\right)$\\\makebox[1cm]{}\Comment{for $i>\Delta$, the coloring $\varphi^{(0)}_i$ uses only the color $0$}
    \item{\textbf{Step 3.}} \textbf{For} $i=1,2,...,h$ \textbf{do}\\
    \makebox[1.15cm]{} \textbf{For} $j=1,2,...,2^{h-i}$ \textbf{do}\\
    \makebox[1.6cm]{}$\varphi^{(i)}_j\leftarrow\textsc{Merge}\left(G^{(i-1)}_{2j-1}, G^{(i-1)}_{2j},\varphi^{(i-1)}_{2j-1},\varphi^{(i-1)}_{2j},\Delta\right)$
\end{description}

Next, we analyse Procedure \textsc{Vertex-Coloring}. We start by showing that the graphs $G^{(0)}_1,G^{(0)}_2,...,G^{(0)}_{\Delta}$ that are defined on step 1 are indeed oriented forests.

\begin{lemma}[Forests decomposition]\label{Forests decomposition}
    Let $G=(V,E)$ be a graph. The subgraphs $G^{(0)}_1,G^{(0)}_2,...,G^{(0)}_{\Delta}$ defined on step 1 of Procedure \textsc{Vertex-Coloring} are oriented forests.
\end{lemma}

\begin{proof}
    Consider a vertex $v\in V$ and an index $i\in \{1,2,...,\Delta\}$. Note that there is at most one edge incident on $v$ that is oriented from $v$ in $G^{(0)}_i$. Indeed, such an edge received the index $i$, and $v$ is its lower-id endpoint. Since each edge received a unique index, there is at most one such edge. Also, since all the edges in the graph are oriented towards their higher-id endpoints, the graph $G_i^{(0)}$ contains no oriented cycles. Note that in any unoriented cycle, there is a vertex with outdegree 2. Hence the graph $G_i^{(0)}$ also contains no unoriented cycles. Hence $G^{(0)}_i$ is a forest.
\end{proof}

We next show that the coloring $\psi$ is a proper $(\Delta+1)$-vertex-coloring.

\begin{lemma}[A bound on the number of colors]\label{delta+1 subcolorings}
    Let $G=(V,E)$ be a graph with maximum degree $\Delta\geq 2$. For each $i\in\left\{0,1,...,h\right\}$ and $j\in\left\{1,2,...,2^{h-i}\right\}$, the coloring $\varphi^{(i)}_j$ is a proper $(\Delta+1)$-vertex-coloring of $G^{(i)}_j$.
\end{lemma}

\begin{proof}
    We prove the lemma by induction on $i$.\\
    For $i=0$, by Lemma \ref{Forests decomposition}, for each $j\in\left\{1,2,...,\Delta\right\}$, $G^{(0)}_j$ is an oriented forest, and by Theorem \ref{3-vertex-coloring properties}, $\varphi^{(0)}_j$ is indeed a proper $3$-vertex-coloring of $G^{(0)}_j$, which is in particular, a proper $(\Delta+1)$-vertex-coloring of $G^{(0)}_j$.\\
    Assume that for each $j\in\left\{1,2,...,\frac{\Delta}{2^i}\right\}$ the coloring $\varphi^{(i)}_j$ is a proper $(\Delta+1)$-vertex-coloring of $G^{(i)}_j$, for some $i\in\{1,2,...,h-1\}$. Let $j\in\left\{1,2,...,\frac{\Delta}{2^{i+1}}\right\}$, and  $\varphi^{(i+1)}_{j}=\textsc{Merge}\left(G^{(i)}_{2j-1},G^{(i)}_{2j},\varphi^{(i)}_{2j-1},\varphi^{(i)}_{2j},\Delta\right)$, as defined in step 3. By Theorem \ref{Procedure Merge properties}, $\varphi^{(i+1)}_{j}$ is indeed a proper $(\Delta+1)$-vertex-coloring of $G^{(i+1)}_j$.\\
    Hence, we conclude that the coloring $\varphi^{\left(h\right)}_1$ that is returned by the algorithm is a proper $(\Delta+1)$-vertex-coloring of $G^{\left(h\right)}_1=G$.
\end{proof}

Observe that one can omit the assumption that $\Delta\geq 2$. Otherwise, the graph is a union of a matching and an independent set, and it can be colored using $\Delta+1\leq 2$ colors using $O(n)$ processors in $O(1)$ time.
Finally, we analyse the complexity of the algorithm.

\begin{lemma}[Complexity of Procedure \textsc{Vertex-Coloring}]
    Let $G=(V,E)$ be an $n$-vertex $m$-edge graph with maximum degree $\Delta$. Procedure \textsc{Vertex-Coloring} requires $O(\log\log n+\Delta\cdot\log^2\Delta)$ time using $O(m)$ processors.
\end{lemma}

\begin{proof}
    First, we analyse the assignment of indexes to the edges.
    To this end, we assign a processor to each $v\in V$ and a neighbor $u$ of $v$. These processors sort the neighbors of $v$ that have a higher-id than $v$, and assign to the edge that connects $v$ with a higher-id neighbor $u$ its index in the sorted order. This process requires $O(\log\Delta)$ time using $O(m)$ processors overall.
    Note that on step 2, all the forests $G_i^{(0)}$, $i\in\{1,2,...,2^h\}$, are edge-disjoint. Initializing vertex-colors in all these forests requires $O(1)$ time using $O(n)$ processors. Denote by $m_i$ the number of edges in $G^{(0)}_i$, for every $i\in\{1,2,...,2^h\}$. The rest of the parallel executions of Procedure \textsc{3-Vertex-Coloring-Forests} on all these forests $G^{(0)}_1,...,G^{(0)}_{2^h}$ requires (by Theorem \ref{3-vertex-coloring properties}) $O(\log^* n)$ time, and $O\left(\bigcup_{i=1}^{2^h}m_i\right)=O(m)$ processors.
    Finally, by Theorem \ref{Procedure Merge properties}, for each $i\in\left\{1,2,...,\lceil\log\Delta\rceil\right\}$, the merging process in step 3 requires $O(\Delta\cdot\log\Delta)$ time using $O(m)$ processors. Hence, overall Procedure \textsc{Vertex-Coloring} requires $O(\log^* n+\Delta\cdot\log^2\Delta)$ time using $O(m+n)=O(m)$ processors. (We assume, without loss of generality, that the graph is connected.)
\end{proof}

We summarize the main result of this section in the next theorem.

\begin{theorem}[Properties of Procedure \textsc{Vertex-Coloring}]\label{procedure vertex-coloring properties}
    Let $G=(V,E)$ be an $n$-vertex $m$-edge graph with maximum degree $\Delta$. Procedure \textsc{Vertex-Coloring} computes a proper $(\Delta+1)$-vertex-coloring of $G$ in $O(\log^* n+\Delta\cdot\log^2\Delta)$ time using $O(m)$ processors.
\end{theorem}

\subsection{A Simple Coloring of Graphs with Bounded Arboricity}\label{sec: A Simple Coloring of Graphs with Bounded Arboricity}
In this section we adapt a distributed algorithm of Barenboim and Elkin~\cite{barenboim2008sublogarithmic} for $O(a)$-vertex-coloring of graphs with arboricity at most $a$. \emph{Arboricity} of a graph $G=(V,E)$ is defined by $a(G)=\max_{U\subseteq V,|U|\geq 2}\left\{\frac{|E(U)|}{|U|-1}\right\}$. The algorithm requires $O(a\cdot\log n)$ rounds. 

\begin{definition}[$H$-decomposition~\cite{barenboim2008sublogarithmic}]
    For a parameter $A$, an \emph{$H$-decomposition} of a graph $G=(V,E)$ is a partition $V=\bigcup_{i=1}^{\ell}H_i$, for some $\ell\geq 1$, that satisfies that for every index $i\in \{1,2,...,\ell\}$, and every vertex $v\in H_i$, $v$ has at most $A$ neighbors in $\bigcup_{j=i}^{\ell}H_j$.
\end{definition}

\begin{theorem}[$\mathrm{PRAM}$ computation of $H$-decomposition]\label{H-decomposition pram}
    Let $G=(V,E)$ an $n$-vertex $m$-edge graph with arboricity $a$. A computation of an $H$-decomposition of $G$ with out-degree $A=O(a)$ requires $O(\log n\cdot\log\Delta)$ time using $O(m)$ processors.
\end{theorem}

\begin{proof}
    The algorithm is composed of $\ell=O(\log n)$ phases. Suppose that at the beginning of a phase $i$, every vertex $v$ knows its degree in $V\setminus \bigcup_{j=1}^{i-1}H_j$. Then the vertices with degree at most $A$ join the set $H_i$. For every vertex $v$ that joins $H_i$, each of its neighbors in $V\setminus \bigcup_{j=1}^{i}H_j$ decrements its degree by 1. By designating a processor to each edge, this can be done in $O(\log\Delta)$ time using $O(n+m)=O(m)$ processors.

\end{proof}

The $O(a)$-vertex-coloring algorithm starts with computing an $H$-decomposition of the input graph $G=(V,E)$ with out-degree $A=O(a)$.
Once the $H$-decomposition is in place, we compute in parallel an $(A+1)$-vertex-coloring $\varphi_i$ of every $H$-set $H_i$, for all $i\in\{1,2,...,\ell\}$.
In distributed setting this can be done in $\Tilde{O}\left(\sqrt{A}\right)+\log^* n$ time via Barenboim's algorithm~\cite{barenboim2016deterministic}. In $\mathrm{PRAM}$ setting we have seen (see Theorem \ref{procedure vertex-coloring properties}) that this can be done in $O(A\cdot\log^2 A+\log^* n)$ time and $O(m)$ processors.
Finally, there is a recoloring step, during which the vertices compute a proper $(A+1)$-vertex-coloring $\psi$ in $O(A\cdot\log n)=O(a\cdot\log n)$ distributed rounds. (In $\mathrm{PRAM}$ setting each of these rounds corresponds to a phase, whose time and work complexities we analyze below.)

The vertices $v$ of $H_{\ell}$ retain their colors, i.e., set $\psi(v)=\varphi_{\ell}(v)$. Next, iteratively we run $A+1$ recoloring phases for each $H_{j}$, $j\in\{1,2,...,\ell-1\}$, starting at $H_{\ell-1}$ until $H_1$. Consider some index $1\leq j<\ell$, and suppose that all vertices of $\bigcup_{h=j+1}^{\ell}H_h$ were already recolored, and we are now recoloring the vertices of $H_j$. We spend one phase for each of the $A+1$ color classes of $\varphi_j$. Suppose that for some $i\in\{0,1,...,A\}$, vertices of the first $i$ color classes of $\varphi_j$ were already recolored, and we are now recoloring vertices $v\in H_j$ with $\varphi_j(v)=i+1$. Each of these vertices has at most $A$ recolored neighbors $\left(\text{in }\bigcup_{h=j}^{\ell}H_h\right)$. Hence there is an available color for $v$ in $[A+1]\setminus\left\{\psi(u)\mid \text{$u\in \bigcup_{h=j}^{\ell}H_h$ is a neighbor of $v$ that was already recolored}\right\}$, and $v$ sets $\varphi(u)$ to be such a color. (All vertices $v\in H_j$ with $\varphi_j(v)=i+1$ do this in parallel. Note that since they form an independent set, the resulting coloring is proper.) 

Next, we analyze the $\mathrm{PRAM}$ running time of this algorithm. Each time a vertex $v$ colors itself with a color $\psi(v)$, it uses $\deg(v)$ processors associated with it to eliminate $\psi(v)$ from palettes of its neighbors. Finding an available color can then be done in $O(\log A)=O(\log a)$ time using $O(A)=O(a)$ processors (per vertex). Hence, overall $(A+1)\cdot\log n$ recoloring phases require $O(A\cdot\log A\cdot\log n)=O(a\cdot\log a\cdot\log n)$ time and $O(|E|)$ processors. Hence, the overall $\mathrm{PRAM}$ time for this step is $O(a\cdot\log a\cdot\log n)$, using $O(|E|)$ processors.
Hence, the overall time is $O\left(a\cdot\log^2 a+\log\log n\right)+O\left(a\cdot\log a\cdot\log n\right)=O\left(a\cdot\log a\cdot\log n\right)$.

We summarize this discussion with the following theorem.
\begin{theorem}[An adaptation of~\cite{barenboim2008sublogarithmic}]\label{Adaptation of BE08}
    Given an $n$-vertex $m$-edge graph $G=(V,E)$ with arboricity $a$, an $O\left(a\right)$-vertex-coloring can be computed in $O\left((a\cdot\log a+\log\Delta)\cdot\log n\right)$ time, using $O\left(m\right)$ processors.
\end{theorem}

\subsection{Computing $\Delta^{1+o(1)}$-Vertex-Coloring in $\mathrm{PRAM}$ Time $O(\log n\cdot\log ^{O(1)}\Delta)$}\label{app:Delta^{1+o(1)}-Vertex-Coloring}
We adapt the distributed algorithm of Barenboim and Elkin~\cite{barenboim2011deterministic}, that given a graph with arboricity $a$ and a parameter $\varepsilon>0$, computes an $O(a^{1+\varepsilon})$-vertex-coloring in $O\left(\frac{1}{\varepsilon}\cdot a^{2\varepsilon}\cdot\log n\right)$ rounds, to the $\mathrm{PRAM}$ setting.

We start with the following definition. 
\begin{definition}[Arbdefective Coloring~\cite{barenboim2011deterministic}]
    A vertex-coloring $\psi$ of a graph $G$ is called \emph{$t$-coloring $q$-arbdefective} (for a pair of parameters $t$ and $q$) if it uses $t$ colors and each color class of $\psi$ induces a graph with arboricity at most $q$. We also say that $q$ is the \emph{arbdefect} of the coloring $\psi$.
\end{definition} 
The algorithm starts with computing an $H$-decomposition with out-degree $A=O(a)$. This step requires $O(\log n)$ distributed rounds, and also (by Theorem \ref{H-decomposition pram}) $O(\log n\cdot\log \Delta)$ $\mathrm{PRAM}$ time with $O(|E|)$ processors. The algorithm now computes (in parallel) a $c\cdot A^{2\varepsilon}$-coloring, $(A^{1-\varepsilon})$-defective $\varphi_i$ (using the parameter $p=A^{\varepsilon}$), of $H_i$, for every $i\in\{1,2,...,\ell\}$, for a fixed constant $c>0$. (Recall that the maximum degree within each $H_i$ is at most $A$.) By Theorem \ref{A variant of defective coloring of BEK14} (in Appendix \ref{sec: Extension of Linial's Algorithm to Defective Coloring}), this step requires $O(\log a\cdot\log^*n+\log\log n)$ time and $O\left(|E|\cdot\frac{a^{\varepsilon}\cdot\log n}{\log(a^{\varepsilon}\cdot\log n)}\right)$ processors. 

Now the algorithm recolors the graph in the following way in $A^{\varepsilon}$ colors. (The new coloring will be called $\psi$.) The algorithm recolors the vertices in each $H_i$ ($i\in\{1,2,...,\ell\}$), starting with recoloring $H_{\ell}$, then $H_{\ell-1}$, and finally it recolors $H_1$. It spends $c\cdot A^{2\varepsilon}$ phases on each one of them. For each $j\in\{1,2,...,\ell\}$, each of the $c\cdot A^{2\varepsilon}$ phases dedicated to $H_j$ is spent on recoloring vertices $v\in H_j$ with a given color class of $\varphi_j$. Specifically, vertices of the first $A^{\varepsilon}$ color classes of $H_{\ell}$ retain their colors, i.e., they set $\psi(v)=\varphi_{\ell}(v)$. Then vertices of the color class $A^{\varepsilon}+1$ with respect to $\varphi_{\ell}$ recolor themselves in parallel. Each $v$ with $\varphi_{\ell}(v)=A^{\varepsilon}+1$ has up to $A$ neighbors that were already recolored (in a color from $\left\{1,2,...,A^{\varepsilon}\right\}$). It now selects a color that is used by at most $A^{1-\varepsilon}$ of its recolored neighbors, and sets $\psi(v)$ to be this color (such a color exists by the pigeonhole principle). Then vertices of the color class $A^{\varepsilon}+2$ with respect to $\varphi_{\ell}$ do the same, etc.,..., and finally, vertices of $\varphi_{\ell}$-color $c\cdot A^{2\varepsilon}$ recolor themselves. Once this is done, vertices of $H_{\ell-1}$-color 1 do the same, then with $\varphi_{\ell-1}$-color 2, etc.,..., and finally, vertices of $\varphi_{\ell-1}$-color $c\cdot A^{2\varepsilon}$. Then the algorithm does the same with $H_{\ell-2},...,H_1$. It is not hard to see that the resulting coloring $\psi$ is $A^{\varepsilon}$-coloring, $O(A^{1-\varepsilon})$-arbdefective. (Indeed, when a vertex $v\in H_i$ is colored by $\psi(v)$, there are at most $A^{1-\varepsilon}$ neighbors that are already recolored with the same $\psi$-color. In addition, there are up to $A^{1-\varepsilon}$ neighbors of $v$ with the same $\varphi_i$-color. Each of these neighbors recolors itself simultaneously with $v$, and may end up be colored by $\psi(v)$ as well.) At this point the algorithm recurses on each of the $A^{\varepsilon}$ color classes in parallel. It does so for $O\left(\frac{1}{\varepsilon}\right)$ recursion levels, up until we are left with $O\left(2^{\frac{1}{\varepsilon}}\cdot A^{1-\varepsilon}\right)$ subgraphs with arboricity $O(A^{\varepsilon})$ each. We then use the algorithm from Appendix \ref{sec: A Simple Coloring of Graphs with Bounded Arboricity} to $O(A^{\varepsilon})$-color each such subgraph within $O\left(\left(a^{\varepsilon}\cdot\log a\cdot\varepsilon+\log\Delta\right)\cdot\log n\right)$ time, $O(|E|)$ processors (see Theorem \ref{Adaptation of BE08}). Overall we obtain $O\left(2^{\frac{1}{\varepsilon}}\cdot A\right)=O\left(2^{\frac{1}{\varepsilon}}\cdot a\right)$-vertex-coloring in $\mathrm{PRAM}$ time $O\left(\frac{1}{\varepsilon}\cdot a^{2\varepsilon}\cdot\log a \cdot\log n+\log\Delta\cdot\log n\right)$, using $O\left(|E|\cdot\frac{a^{\varepsilon}\cdot\log n}{\log\left(a^{\varepsilon}\cdot\log n\right)}\right)$ processors. By setting $\varepsilon=\frac{c\cdot\log\log a}{\log a}$, for a constant parameter $c>0$, we obtain $a\cdot2^{O\left(\frac{\log a}{\log\log a}\right)}=a^{1+o(1)}$-vertex-coloring in $O\left(\frac{\log a}{\log\log a}\cdot\log^{2c+1} a\cdot\log n+\log\Delta\cdot\log n\right)$ time, with $O\left(|E|\cdot\frac{\log^c a\cdot\log n}{\log(\log a\cdot\log n)}\right)=O\left(|E|\cdot\frac{\log^c a\cdot\log n}{\log\log n}\right)$ processors. As a result we can also obtain an independent set of size $\Omega\left(\frac{n}{a^{1+o(1)}}\right)$ within the same time, using the above number of processors.
Since for every graph $a\leq \Delta$, it follows that the algorithm provides also a $\Delta^{1+o(1)}$-vertex-coloring in $O\left(\log^{1+\delta}\Delta\cdot\log n\right)$ time, $O\left(|E|\cdot\frac{\Delta^{\delta}\cdot\log n}{\log(\Delta\cdot\log n)}\right)$ processors, for an arbitrarily small constant $\delta>0$. And as a result we can also obtain an independent set of size $\Omega\left(\frac{n}{\Delta^{1+o(1)}}\right)$ within the same time, using the above number of processors.

We summarize this discussion with the following theorems.
\begin{theorem}[An adaptation of~\cite{barenboim2011deterministic}]\label{Adaptation of BE11 arboricity}
    Given an $n$-vertex $m$-edge graph $G=(V,E)$ with maximum degree $\Delta$ and arboricity $a$, a $a^{1+o\left(1\right)}$-vertex-coloring can be computed in $O\left(\left(\log^{2+\delta} a+\log\Delta\right)\cdot\log n\right)$ time (for an arbitrary small constant $\delta>0$), using $O\left(m\cdot\frac{\log^{\delta}a\cdot\log n}{\log(a\cdot\log n)}\right)$ processors. As $a\leq \Delta$, this is also a $\Delta^{1+o\left(1\right)}$-vertex-coloring, in $O\left(\log^{2+\delta} \Delta\cdot\log n\right)$ time (for an arbitrary small constant $\delta>0$), using $O\left(m\cdot\frac{\log^{\delta}\Delta\cdot\log n}{\log(\Delta\cdot\log n)}\right)$ processors.
\end{theorem}

\subsection{Arbdefective Coloring in Bounded-Degree Graphs}\label{Arbdefective Coloring in Bounded-Degree Graphs}
In this section we adapt to $\mathrm{PRAM}$ setting a distributed algorithm due to~\cite{barenboim2018locally}. In this algorithm we are given and graph $G$ with maximum degree $\Delta$, and a parameter $\rho$, $1\leq\rho\leq\Delta$. The algorithm computes an $O(\rho)$-arbdefective, $O\left(\frac{\Delta}{\rho}\right)$-coloring of $G$ in $O\left(\frac{\Delta}{\rho}\right)$ distributed rounds. We then employ this algorithm in Appendix \ref{sec: Adapting Barenboim's Algorithm to PRAM Model} to adapt Barenboim's distributed algorithm~\cite{barenboim2016deterministic} for $O(\Delta)$-vertex-coloring in $O\left(\sqrt{\Delta}\right)+\log^*n$ time to $\mathrm{PRAM}$ setting. 

The algorithm starts with computing a $\rho$-defective, $O\left(\left(\frac{\Delta}{\rho}\right)^2\right)$-coloring $\varphi$ in $\log^*n +O(1)$ distributed rounds. (See Theorem \ref{adaptation of BEK14}.) Now every vertex $v$ represents its color $\varphi(v)$ as a pair $\langle a,b\rangle=\langle a,b\rangle_{\varphi(v)}$, $a,b=O\left(\frac{\Delta}{\rho}\right)$. Specifically, let $c>0$ be a constant such that $\varphi(v)$ employs at most $c^2\cdot\left(\frac{\Delta}{\rho}\right)^2$ colors. Then let $p$ be a prime such that $2c\cdot\frac{\Delta}{\rho}<p\leq 4c\cdot\frac{\Delta}{\rho}+2$. (It exists by Bertrand-Chebyshev's principle.) Then $\langle a,b\rangle_{\varphi(v)}$ is the representation of $\varphi(v)$ in the basis $p$. 

The algorithm sets $\psi_1(v)=\varphi(v)$, and starts iterating for $p$ rounds. The color that $v$ has at the beginning of round $i$ (among these $p$ rounds, $1\leq i\leq p$) is denoted by $\psi_i(v)$. If $\varphi_i(v)=\langle 0,b\rangle_{\psi(v)}$, this color is said to be \emph{final} or \emph{finalized}, and otherwise it is not yet final. If $v$ changes its color from a not final one to a final one on round $i$, we say that it \emph{finalizes its color} on round $i$. For a pair of neighbors $u,v$ we say that their colors are \emph{in conflict} with one another on some round $i$ if $\psi_i(u)=\langle a_u,b_u\rangle$, $\psi_i(v)=\langle a_v,b_v\rangle$, and $b_u=b_v$.

On each round $i$ every vertex $v$ with a not final color $\psi_i(v)=\langle a,b\rangle$ checks how many neighbors $u$ of $v$ whose original color is different from that of $v$ (i.e., $\varphi(u)\neq\varphi(v)$) are in conflict with $v$. If the number of (such) conflicts is at most $\rho$, than $v$ finalizes its color, i.e., sets $\psi_{i+1}(v)=\langle 0,b\rangle$. Otherwise it sets $\psi_{i+1}(v)=\langle a,b+a\rangle$. (The summation is modulo $p$.) This completes the description of the (distributed) algorithm (due to~\cite{barenboim2018locally}).

Next, we sketch its analysis. Fix a vertex $v$ and consider a neighbor $u$ of $v$ with $\varphi(u)\neq \varphi(v)$. Denote $\langle a_v,b_v\rangle=\varphi(v)$, $\langle a_u,b_u\rangle=\varphi(u)$. Consider the period of time starting with the beginning of the algorithm and until $v$ finalizes (or until the end of $p$ rounds). Within this period of time $u$ may be in conflict with $v$ at most once until $u$ finalizes, and at most once after that. (This is because equation $b_v+i\cdot a_v\equiv b_u+i\cdot a_u\, (\text{mod }p)$ has one solution, and this is also the case for equation $b_v+i\cdot a_v\equiv b'\, (\text{mod }p)$ for any $b'\in GF(p)$.) Hence, overall $v$ may have at most $2\Delta$ conflicts with its neighbors within $p$ rounds, before it finalizes. On the other hand, on each round on which $v$ does not finalize, it has more than $\rho$ conflicts. Thus, it cannot stay unfinalized for more than $\frac{2\Delta}{\rho}$ rounds. Since $p>\frac{2\Delta}{\rho}$, it follows that $v$ (and every other vertex) eventually finalizes.

Now we argue that the resulting coloring $\psi=\psi_{p+1}$ is $O\left(\frac{\Delta}{\rho}\right)$-coloring, $O(\rho)$-arbdefective. For the sake of this argument, we orient every edge $(v,u)$ with $\psi(v)=\psi(u)$ in the following way: if $v$ (respectively, $u$) finalizes on a round $i_v$ (resp., $i_u$), and $i_v>i_u$, orient the edge towards $u$ (i.e., as $\langle v,u\rangle$). If $i_v=i_u$ or if their original colors satisfy $\varphi(v)=\varphi(u)$, then orient the edge towards the endpoint with larger Id (i.e., essentially arbitrarily).

Observe that $\psi$ employs $p=O\left(\frac{\Delta}{\rho}\right)$ colors (as $\psi(v)=\langle 0,b_v\rangle$, for every $v\in V$, and $b_v\in GF(p)$). Also, we now argue that under the above orientation, every vertex $v$ has at most $O(\rho)$ outgoing neighbors $u$ with $\psi(u)=\psi(v)$.

Indeed, recall that $\varphi$ is a $\rho$-defective coloring, and thus $v$ may have at most $\rho$ outgoing neighbors $u$ with $\psi(v)=\psi(u)$ and $\varphi(v)=\varphi(u)$. In addition, for each outgoing neighbor $u$ of $v$ with $\psi(v)=\psi(u)$ and $\varphi(v)\neq \varphi(u)$, the neighbor $u$ finalized its color before (or together with) $v$. Thus, $v$ was in conflict with $u$ when $v$ finalized its color (as $\psi(v)=\psi(u)$). Since $v$ did finalize its color on that round, it was in at most $\rho$ conflicts of that point, and thus there are at most $\rho$ such outgoing neighbors $u$. Hence, arbdefect of the coloring $\psi$ is at most $2\rho=O(\rho)$.

Finally, we discuss the $\mathrm{PRAM}$ implementation of this algorithm. As we have seen in Appendix \ref{sec: Extension of Linial's Algorithm to Defective Coloring} (Theorem \ref{A variant of defective coloring of BEK14}), computing the defective coloring $\varphi$ requires $O(\log\Delta
\cdot\log^* n+\log\log n)$ time, with $O\left(m\cdot\frac{\frac{\Delta}{\rho}\cdot\log n}{\log\left(\frac{\Delta}{\rho}\cdot\log n\right)}\right)$ processors. Each of the $p=O\left(\frac{\Delta}{\rho}\right)$ rounds of the main loop (that converts the initial coloring $\psi_0=\varphi$ into the ultimate coloring $\psi_p=\psi$) can be implemented in the following way: we designate $\deg(v)$ processors to every vertex $v$. These processors count the number of conflicts within $O(\log\deg(v))=O(\log \Delta)$ time, using $O(m)$ processors. Hence the overall time is $O(p\cdot\log\Delta)=O\left(\frac{\Delta}{\rho}\cdot\log\Delta\right)$, and the number of processors is $O\left(m\cdot\frac{p\cdot\log n}{\log(p\cdot\log n)}\right)$.
\begin{theorem}[An adaptation of~\cite{barenboim2018locally}]\label{Adaptation of [BEG18]}
    For any $n$-vertex $m$-edge graph with maximum degree $\Delta$, and a parameter $1\leq\rho\leq\Delta$, an $O\left(\frac{\Delta}{\rho}\right)$-coloring $O(\rho)$-arbdefective can be computed in $$O\left(\frac{\Delta}{\rho}\cdot\log \Delta+\log\Delta\cdot\log^* n+\log\log n\right)=O\left(\log\Delta\cdot\left(\frac{\Delta}{\rho}+\log^* n\right)+\log\log n\right)$$ $\mathrm{PRAM}$ time, using $O\left(m\cdot\frac{\frac{\Delta}{\rho}\cdot\log n}{\log\left(\frac{\Delta}{\rho}\cdot\log n\right)}\right)$ processors.
\end{theorem}

To optimize the number of processors, we use the trade-off of Theorem \ref{An adaptation of [BEK14], a better trade-off} with the parameter $z$ given by $z=\frac{\sqrt{p}\cdot\log n}{\log(p\cdot\log n)}$, $p=\frac{\Delta}{\rho}$.
Then computing the defective coloring requires
$$
O\left(\frac{p\cdot\log n}{z\cdot\log(p\cdot\log n)}+\log\Delta\cdot\log^*n+\log\log n\right)=O\left(\left(\sqrt{p}+\log^*n\right)\cdot\log\Delta+\log\log n\right)
$$
time, but uses
$O(m\cdot z)=O\left(m\cdot\frac{\sqrt{p}\cdot\log n}{\log(p\cdot\log n)}\right)$
processors. Together with the $p$ rounds of the main loop, the overall running time is
\begin{align*}
    &O\left((\sqrt{p}+\log^*n)\cdot\log\Delta+\log\log n\right)+p\cdot\log\Delta)=\\&O\left((p+\log^*n)\cdot\log\Delta+\log\log n\right)=\\&O\left(\frac{\Delta}{\rho}\cdot\log\Delta+\log\Delta\cdot\log^*n+\log\log n\right),
\end{align*}
and the number of processors is
$$O\left(m\cdot\frac{\sqrt{p}\cdot\log n}{\log(p\cdot\log n)}\right)=O\left(m\cdot \frac{\sqrt{\frac{\Delta}{\rho}}\cdot\log n}{\log\left(\frac{\Delta}{\rho}\cdot\log n\right)}\right).$$

This improves the bound on the number of processors from Theorem \ref{Adaptation of BE08}, while other parameters stay the same.

\begin{corollary}
    For an $n$-vertex $m$-edge graph with maximum degree $\Delta$, and a parameter $1\leq\rho\leq \Delta$, an $O(\rho)$-arbdefective $O\left(\frac{\Delta}{\rho}\right)$-coloring can be computed in $$O\left(\frac{\Delta}{\rho}\cdot\log\Delta+\log\Delta\cdot\log^*n+\log\log n\right)$$ time, using $O\left(m\cdot \frac{\sqrt{\frac{\Delta}{\rho}}\cdot\log n}{\log\left(\frac{\Delta}{\rho}\cdot\log n\right)}\right)$ processors.
\end{corollary}

A particularly important setting is when $\rho=\sqrt{\Delta}$. (Note that if $\log^* n>\sqrt{\Delta}$ then the term $\log\log n$ dominates the running time.) We then obtain:
\begin{corollary}\label{cor: Adaptation of [BEG18]}
    For any $n$-vertex $m$-edge graph with maximum degree $\Delta$, an $O\left(\sqrt{\Delta}\right)$-arbdefective $O\left(\sqrt{\Delta}\right)$-coloring can be computed in $O\left(\sqrt{\Delta}\cdot\log\Delta+\log\log n\right)$ $\mathrm{PRAM}$ time, using $O\left(m\cdot\frac{\Delta^{1/4}\cdot\log n}{\log\left(\Delta\cdot\log n\right)}\right)$ processors.
\end{corollary}

\subsection{Adapting Barenboim's Algorithm to $\mathrm{PRAM}$ Model}\label{sec: Adapting Barenboim's Algorithm to PRAM Model}
In this section we adapt a distributed algorithm due to Barenboim~\cite{barenboim2016deterministic} for $O(\Delta)$-vertex-coloring in $\Tilde{O}\left(\sqrt{\Delta}\right)+O\left(\log^* n\right)$ rounds to the $\mathrm{PRAM}$ model. 

We start with describing the algorithm on the distributed model, and sketching its analysis. Then we adapt it to the $\mathrm{PRAM}$ model. 

The algorithm starts with computing an $O(p)$-coloring, $\left(\frac{\Delta}{p}\right)$-arbdefective $\hat{\varphi}$ in $O\left(\frac{\Delta}{p}\right)+\log^* n$ time by the algorithm of~\cite{barenboim2018locally} (see Theorem \ref{Adaptation of [BEG18]}), for a parameter $p$, $1\leq p\leq\Delta$. (We will use $p=\sqrt{\Delta}$.) Let $G_1,G_2,..,G_p$ be the $p$ subgraphs of $G$ induced by the $p$ color classes of $\hat{\varphi}$, i.e., for every $i\in \{1,2,...,p\}$, the subgraph $G_i$ is induced by vertices $v$ with $\hat{\varphi}(v)=i$. Observe that as part of this computation we also obtain a forest decomposition of each $G_i$ into $\left(\frac{\Delta}{p}\right)$ forests (and that arboricity of each $G_i$ is at most $\left(\frac{\Delta}{p}\right)$). We now invoke in parallel the Arb-Linial algorithm on each of them, and obtain an $O\left(\left(\frac{\Delta}{p}\right)^2\right)$-vertex-coloring $\varphi_i$ for each $G_i$ within additional $\log^* n+O(1)$ rounds. Note that when $p=O\left(\sqrt{\Delta}\right)$, $\varphi_i$ is an $O(\Delta)$-coloring of $G_i$, for every $i\in\{1,2,...,p\}$.

Barenboim's algorithm now computes a (proper) $O(\Delta)$-vertex-coloring $\psi$ of the entire graph $G$ within $p$ iterations. (For a general $p$, the number of colors it uses is $\Lambda(\Delta,p)=O\left(\Delta+\frac{\Delta}{p}\cdot\left(\sqrt{\Delta}+\frac{\Delta}{p}\right)\right)$.) Specifically, it starts with initializing the colors $\psi(v)=\varphi_1(v)$ for all vertices $v\in V(G_1)$. It then proceeds with recoloring vertices of $G_2$ in one round, and then vertices of $G_3$ within an additional round,$\ldots$, etc., and finally it recolors the vertices of $G_p$, again in one round. We next describe one single iteration $t$, $2\leq t\leq p$, of this loop. On this iteration the algorithm recolors vertices of $G_t$ in a way consistent with colors of already recolored vertices (that belong to $G_1\cup...\cup G_{t-1}$). In addition, the new coloring $\psi$ of vertices of $G_t$ satisfies that for any $(u,v)\in E$, $u,v\in V(G_t)$, $\psi(u)\neq\psi(v)$. The coloring $\psi$ uses $\Lambda=\Lambda(\Delta,p)$ colors.

At the beginning of this iteration we have a proper $\Lambda$-coloring $\psi$ of $G_1\cup...\cup G_{t-1}$, and a $c\cdot\left(\frac{\Delta}{p}\right)^2$-coloring $\varphi=\varphi_t$ of $G_t$, for a universal constant $c>9$.
In addition, recall that we are given a forest decomposition of $G_t$ into $a=\frac{\Delta}{p}$ forests, and every vertex $v\in V(G_t)$ knows its $a$ parents under this decomposition. Let also $L(v)$ be the set of $\psi$-colors used by neighbors $u$ of $v$ that belong to $G_1\cup...\cup G_{t-1}$.

We set a parameter $\mu$ to be a prime between $\sqrt{\Delta}+\sqrt{c}\cdot\frac{\Delta}{p}$ and $2\left(\sqrt{\Delta}+\sqrt{c}\cdot\frac{\Delta}{p}\right)$. (Such a prime exists by Bertrand-Chebyshev's principle.) The coloring $\psi$ uses
\begin{equation}\label{eq: Lambda}
    \Lambda=\Delta+(2a+1)\cdot\mu
\end{equation} 
colors. Each vertex $v\in V(G_t)$ constructs a set of $\mu$ univariate polynomials $p_{0}^{(\varphi(v))}(x),...,p_{\mu-1}^{(\varphi(v))}(x)$ over $GF(\mu)$. For every $i\in GF(\mu)$, the polynomial $p_{i}^{(\varphi(v))}(x)$ is given by $p_{i}^{(\varphi(v))}(x)=i+a_v\cdot x+b_v\cdot x^2\, (\text{mod } \mu)$, where $\varphi(v)=\langle a_v,b_v\rangle_{\varphi(v)}$. (That is, $\langle a_v,b_v\rangle_{\varphi(v)}$ is the representation of $\varphi(v)$ on the basis $\left\lceil\sqrt{c}\cdot a\right\rceil$.) Observe that for a pair of neighbors $u,v\in V(G_t)$, we have $\varphi(u)\neq \varphi(v)$, and thus $\langle a_v,b_v\rangle_{\varphi(v)}\neq\langle a_u,b_u\rangle_{\varphi(u)}$. Hence their respective sets of polynomials $\left\{p_{0}^{(\varphi(v))}(x),...,p_{\mu-1}^{(\varphi(v))}(x)\right\}$ and $\left\{p_{0}^{(\varphi(u))}(x),...,p_{\mu-1}^{(\varphi(u))}(x)\right\}$ are disjoint. (As for any pair of such polynomials, $j+a_v\cdot x+b_v\cdot x^2$ and $j'+a_u\cdot x+b_u\cdot x^2$, either $a_v\neq a_u$ or $b_v\neq b_u$, or both.) We say that a polynomial $p$ \emph{generates} a color $\gamma\in\{0,1,...,\mu^2-1\}$ if the representation $\langle y,x\rangle$ of $\gamma$ in the basis $\mu$ satisfies $y=p(x)$.

For every vertex $v\in V(G_t)$ and index $i\in\{0,1,...,\mu-1\}$, let 
\begin{equation}\label{eq: generated colors}
    \mathcal{P}_i(v)=\left\{\left\langle p_i^{(\varphi(v))}(k),k\right\rangle\mid k=0,1,...,\mu-1\right\}
\end{equation}
be the set of colors generated by the polynomial $p_i^{(\varphi(v))}(x)$. Also, let $L_i(v)=L(v)\cap \mathcal{P}_i(v)$ be the set of colors used by recolored neighbors of $v$ that are generated by $p_i^{(\varphi(v))}(x)$, and let $\hat{L_i}(v)=\mathcal{P}_i(v)\setminus L_i(v)$ be the set of colors generated by the polynomial $p_i^{(\varphi(v))}(x)$, which are not used by recolored neighbors of $v$.

The vertex $v$ selects an index $i$, $i\in\{0,1,...,\mu-1\}$, so that the set $L_i(v)$ has minimum size, and sends it to all its neighbors. Observe that $|L(v)|\leq \Delta$. Observe also that for every $i\neq j$, $i,j\in\{0,1,...,\mu-1\}$, we have $L_i(v)\cap L_j(v)=\emptyset$, as otherwise we have $p_i^{(\varphi(v))}(k)=p_j^{(\varphi(v))}(k)$, for some $k\in\{0,1,...,\mu-1\}$. But then 
$$i+a_v\cdot k+b_v\cdot k^2\equiv j+a_v\cdot k+b_v\cdot k^2 \,(\text{mod } \mu),$$
contradiction. It follows that the size of $L_i(v)$ for the index $i\in\{0,1,...,\mu-1\}$ that minimizes its size is at most $\frac{\Delta}{\mu}$. Hence, the size of the corresponding set $\hat{L_i}(v)=\mathcal{P}_i(v)\setminus L_i(v)$ is at least 
$$\mu-\frac{\Delta}{\mu}\geq \mu-\sqrt{\Delta}\geq\sqrt{\Delta}+\sqrt{c}\cdot\frac{\Delta}{p}-\sqrt{\Delta}=\sqrt{c}\cdot\frac{\Delta}{p}>2\cdot\frac{\Delta}{p}.$$ 
(We used here $\mu\geq \sqrt{\Delta}+\sqrt{c}\cdot\frac{\Delta}{p}\geq\sqrt{\Delta}$ and $c>9$.) 

This means that there are more than $2\cdot\frac{\Delta}{p}$ colors generated by the polynomial $p_i^{(\varphi(v))}(x)$ (for the selected index $i$) that are not used by the already-recolored neighbors of $v$ (neighbors from $\bigcup_{j=1}^{t-1}G_j$). Also, for every parent $u$ of $v$ in $G_t$ (and there are at most $\frac{\Delta}{p}$ such parents), the polynomial $p_i^{(\varphi(v))}(x)$ selected by $v$ is different from the polynomial $p_{j_u}^{(\varphi(u))}(x)$ selected by $u$, and thus, these two polynomials may intersect in at most two points (as these are degree-2 polynomials). In other words, among more than $2\cdot\frac{\Delta}{p}$ colors in $\hat{L}_i(v)$ (generated by $p_i^{(\varphi(v))}(x)$ and not used by recolored neighbors of $v$), there is at least one color which cannot be generated by polynomials $p_{j_u}^{(\varphi(u))}(x)$, selected by parents $u$ of $v$ in $G_t$. The vertex $v$ now colors itself with such a color $\left\langle p_i^{(\varphi(v))}(k),k\right\rangle\in\hat{L_i}(v)$, for the smallest $k$ among such colors.

The discussion above implies that the resulting coloring $\psi$ is proper, and that it is computed within overall $O(p)$ rounds. To analyse the number of employed colors, observe that the smallest $k$ as above is at most $2\cdot\frac{\Delta}{p}+\frac{\Delta}{\mu}$. (Recall that $|L_i(v)|\leq\frac{\Delta}{\mu}$.) Hence the representation $\left\langle p_i^{(\varphi(v))}(k),k\right\rangle$ of this color on the basis $\mu$ gives rise to a color which is at most $$\mu-1+\mu\cdot\left(\frac{2\Delta}{p}+\frac{\Delta}{\mu}\right)\leq\Delta+\mu\cdot\left(\frac{2\Delta}{p}+1\right)=\Lambda(\Delta,p).$$ (See (\ref{eq: Lambda}).) Hence, the coloring $\psi$ employs $\Lambda(\Delta,p)$ colors. As $\mu=O\left(\sqrt{\Delta}+\frac{\Delta}{p}\right)$, we have $\Lambda=\Delta+O\left(\left(\frac{\Delta}{p}\right)^2\right)+\left(\frac{\Delta^{\frac{3}{2}}}{p}\right)$. By setting $p=\left(\sqrt{\Delta}\right)$, we obtain an $O(\Delta)$-vertex-coloring in $O\left(\sqrt{\Delta}+\log^* n\right)$ rounds.

To implement this algorithm in $\mathrm{PRAM}$ setting, we first compute an $O\left(\sqrt{\Delta}\right)$-coloring $\varphi$, $O\left(\sqrt{\Delta}\right)$-arbdefective within $O\left(\sqrt{\Delta}\cdot\log\Delta+\log\log n\right)$ time, using $O\left(m\cdot\frac{\Delta^{1/4}\cdot\log n}{\log(\Delta\cdot\log n)}\right)$ processors (see Corollary \ref{cor: Adaptation of [BEG18]}). Then, in parallel, we compute an $O(\Delta)$-coloring $\varphi_i$ for each color class $G_i$. By Corollary \ref{An adaptation of [Lin87], a trade-off}, for a parameter $z=O\left(\frac{\sqrt{\Delta}\cdot\log n}{\log(\Delta\cdot\log n)}\right)$, this step requires
$$O\left(\log\Delta\cdot\log^*n+\log \log n+\frac{\sqrt{\Delta}\cdot\log n}{z\cdot\log(\Delta\cdot\log n)}\right)$$
time, using $O\left(|E|\cdot z\right)\cdot O\left(\sqrt{\Delta}\right)$
processors ($O\left(|E|\cdot z\right)$ processors for each of the $O\left(\sqrt{\Delta}\right)$ subgraphs). We set $z = \log\Delta$, and obtain overall running time
$$O\left(\log\Delta\cdot\log^*n + \log\log n+\frac{\sqrt{\Delta}\cdot\log n}{\log\Delta\cdot\log(\Delta\cdot\log n)}\right)=O\left(\frac{\sqrt{\Delta}\cdot\log n}{\log\Delta\cdot\log(\Delta\cdot\log n)}\right),$$
and $O\left(|E|\cdot\left(\sqrt{\Delta}\cdot\log\Delta+\frac{\Delta^{1/4}\cdot\log n}{\log(\Delta\cdot\log n)}\right)\right)$
processors.

Then we have $O(p)=O\left(\sqrt{\Delta}\right)$ recoloring iterations. To implement these iterations, for a fixed vertex $v$, we designate $\mu$ processors to every edge $(v,u)$. These processors are indexed $0,1,...,\mu-1$. Processors indexed $i$ (for every $i$) compute the set $L_i(v)=\mathcal{P}_i(v)\cap L(v)$ in $O(1)$ time. (The set $\mathcal{P}_i(v)$ (see Equation (\ref{eq: generated colors})) is computed by these processors in $O(1)$ time. For each recolored neighbor $u$ of $v$, the $i$'th processor associated with the edge $(v,u)$ removes the color $\psi(u)$ of $u$ from $\mathcal{P}_i(v)$.) Then $\hat{L_i}(v)$ is computed within the same time with these processors. Finally, for each parent $u$ of $v$ in its subgraph $G_t$, (of arboricity at most $\frac{\Delta}{p}$), $\mu$ processors associated with the edge $(v,u)$ solve in $O(1)$ time the equation $p_i^{(\varphi(v))}(x)=p_{j_u}^{(\varphi(u))}(x)$ (where $i$ (respectively, $j_u$) is the index selected by $v$ (resp., by $u$)), and remove its (at most) two solutions from $\hat{L_i}(v)$. (To solve these equation, processor $q\in\{0,1,...,\mu-1\}$ computes $p_i^{(\varphi(v))}(q)$ and $p_{j_u}^{(\varphi(u))}(q)$, and checks if they are equal.) Finding the minimum $k\in\{0,1,...,\mu-1\}$ among colors $\left\langle p_i^{(\varphi(v))}(k),k\right\rangle$, which are still in the set $\hat{L_i}(v)$ requires now $O(\log \Delta)$ time. Hence overall, the recoloring process requires $O\left(\log\Delta\cdot\sqrt{\Delta}\right)$ time, using $O(\mu\cdot m)=O\left(\sqrt{\Delta}\cdot m\right)$ processors. \\
To summarize:
\begin{theorem}[An adaptation of~\cite{barenboim2016deterministic}]\label{Adaptation of Bar16}
    An $O(\Delta)$-vertex-coloring can be computed in $$O\left(\sqrt{\Delta}\cdot\left(\log\Delta+\frac{\log n}{\log\Delta\cdot\log(\Delta\cdot\log n)}\right)\right)$$ $\mathrm{PRAM}$ time using $O\left(m\cdot\left(\sqrt{\Delta}\cdot\log\Delta+\frac{\Delta^{1/4}\cdot\log n}{\log(\Delta\cdot\log n)}\right)\right)$ processors.
\end{theorem}

\subsection{Large Independent Set}\label{app: Large Independent Set}
We start by stating a result due to Goldberg and Spencer~\cite{goldberg1989constructing} that computes a maximal independent set in a simple graph.
\begin{lemma}[Maximal independent set algorithm~\cite{goldberg1989constructing}]\label{Maximal independent set algorithm}
    Let $G=(V,E)$ an $n$-vertex $m$-edge graph. There is a deterministic $\mathrm{EREW\,\, PRAM}$ algorithm that finds a maximal independent set in $G$ in $O\left(\log^3 n\right)$ time using $O\left(\frac{n+m}{\log n}\right)$ processors.
\end{lemma}

For our purposes, a "large" independent set is sufficient (see Definition \ref{def: independent set}).
Any vertex-coloring algorithm can be used to compute such an independent set. Namely, since each color class in a proper vertex-coloring is an independent set, we can compute the largest color class in the $n$-vertex graph and return it. 
Denote the time required for computing a $\lambda$-vertex-coloring of an $n$-vertex $m$-edge graph with maximum degree $\Delta$ and arboricity $a$ by $VCT_{\lambda}(n,m,\Delta,a)$ and the number of processors it requires by $VCP_{\lambda}(n,m,\Delta,a)$.
We analyse the computation of a large independent set, using vertex-coloring algorithm, in the next lemma.

\begin{lemma}
    Let $G=(V,E)$ be an $n$-vertex $m$-edge graph with maximum degree $\Delta$ and arboricity $a$. An independent set $I$ of $G$ of size $|I|\geq\frac{n}{\lambda}$, i.e., a $\lambda$-large independent set, can be computed in $O(VCT_{\lambda}(n,m,\Delta,a)+\log n)$ time using $O(VCP_{\lambda}(n,m,\Delta,a)+m)$ processors.
\end{lemma}

\begin{proof}
    As was described above, the algorithm first computes a $\lambda$-vertex-coloring $\varphi$ of $G$, and then returns the largest color class. The computation of the largest color class can be done by sorting the vertices according to their color, and computing the range of indexes of each color in the sorted order. Hence, the whole process requires $O(VCT_{\lambda}(n,m,\Delta,a)+\log n)$ time using $O(VCP_{\lambda}(n,m,\Delta,a)+m)$ processors.
\end{proof}

Using the different vertex-coloring algorithms we presented in this section, we summarize the different bounds on the complexity of computing large independent sets that we get in the next theorem.

\independentSetAlg*

\end{document}